\documentclass[11pt]{article}
\usepackage{graphicx}
\usepackage{bm}
\usepackage{dsfont}
\usepackage{amsmath}
\usepackage{amssymb}
\usepackage{amsfonts}
\usepackage{amsthm}
\usepackage{mathtools}
\usepackage{hyperref}
\usepackage{braket}
\usepackage[normalem]{ulem}
\usepackage{wrapfig}
\usepackage{tikz}

\usepackage[left=1in,right=1in,top=1in,bottom=1in]{geometry}
\usepackage{multirow}
\usepackage{csquotes}
\usepackage{subfig}
\usepackage{multicol}
\usepackage{balance}
\usepackage{algorithm}
\usepackage{algpseudocode}
\usepackage[backend = biber, style = alphabetic,maxbibnames=9]{biblatex}
\hypersetup{colorlinks = true, allcolors = [rgb]{0, 0, 0.7}}

\newtheorem{definition}{Definition}
\newtheorem{theorem}{Theorem}
\newtheorem{proposition}{Proposition}
\newtheorem{lemma}{Lemma}
\newtheorem{corollary}{Corollary}
\newtheorem{conjecture}{Conjecture}
\newtheorem{remark}{Remark}

\def\autorefapp#1{\hyperref[#1]{Appendix~\ref{#1}}}

\renewcommand{\title}[1]{\vbox{\center\bf{\Large #1}}\vspace{5mm}}
\renewcommand{\author}[1]{\vbox{\center{#1}}\vspace{5mm}}
\newcommand{\address}[1]{\vbox{\center{\it #1}}}

\newcommand\emails[1]{\begingroup
	\renewcommand\thefootnote{}\footnote{#1}
	\addtocounter{footnote}{-1}\endgroup}

\def\ep{\varepsilon}

\def\tr{{\rm tr}}

\def\iden{\mathbb{I}}

\def\and{\quad {\rm and} \quad}

\def\ni{\noindent}
\def\nn{\nonumber\\}

\newcommand{\id}{\mathbb I}
\def\cV{\mathcal{V}}
\def\spn{\mathrm{span}}
\def\ketbra#1{|{#1}\rangle\!\langle{#1}|}
\def\g{G}
\def\1D{\mathrm{1D}}
\def\cg{\mathrm{CG}}
\def\DL{\mathrm{DL}}
\def\maxdeg{\kappa}
\def\mindeg{\vartheta}
\def\maxh{h}

\newcommand{\StarGraphRQC}{
    \begin{tikzpicture}
		\foreach \y in {0.5, 1, ..., 2.5} 
		  \draw [thick] (0.5, \y) -- (5, \y);
        
		\node [fill = gray, minimum width = 0.5cm, minimum height = 0.75cm, anchor = south west, rounded corners] at (1, 0.5 - 0.125) {};

		\node [fill = white, minimum width = 0.75cm, minimum height = 0.25cm, anchor = south west] at (2 - 0.125, 1 - 0.125) {};
		\node [fill = gray, minimum width = 0.5cm, minimum height = 1.25cm, anchor = south west, rounded corners] at (2, 0.5 - 0.125) {};

		\node [fill = white, minimum width = 0.75cm, minimum height = 0.75cm, anchor = south west] at (3 - 0.125, 1 - 0.125) {};
		\node [fill = gray, minimum width = 0.5cm, minimum height = 1.75cm, anchor = south west, rounded corners] at (3, 0.5 - 0.125) {};

        \node [fill = white, minimum width = 0.75cm, minimum height = 1.25cm, anchor = south west] at (4 - 0.125, 1 - 0.125) {};
		\node [fill = gray, minimum width = 0.5cm, minimum height = 2.25cm, anchor = south west, rounded corners] at (4, 0.5 - 0.125) {};

        \begin{scope}[shift = {(4, 0)}]
            \foreach \y in {0, 0.5, 1, 1.5}
                \draw[thick] (1, 1 + \y) .. controls (1.5, 1 + \y) and (1.5, 0.5 + \y).. (2, 0.5 + \y);
            \draw[thick] (1, 0.5) .. controls (1.5, 0.5) and (1.5, 2.5).. (2, 2.5);
        \end{scope}

        \node [anchor = east] () at (0.5, 0.5) {$v_1$};
        \node [anchor = east] () at (0.5, 1.0) {$v_2$};
        \node [anchor = east] () at (0.5, 1.5) {$v_3$};
        \node [anchor = east] () at (0.5, 2.0) {$v_4$};
        \node [anchor = east] () at (0.5, 2.5) {$v_5$};

    \end{tikzpicture}
}

\newcommand{\PathGraphRQC}{
    \begin{tikzpicture}
        \foreach \y in {0.5, 1, ..., 2.5} 
            \draw [thick] (0.5, \y) -- (5.25, \y);
    
        \foreach \x/\y in {1/0.5, 2/1, 3/1.5, 4/2} 
            \node [fill = gray, minimum width = 0.5cm, minimum height = 0.75cm, anchor = south west, rounded corners] at (\x, \y - 0.125) {};
    
        \node [anchor = east] () at (0.5, 0.5) {$v_1$};
        \node [anchor = east] () at (0.5, 1.0) {$v_2$};
        \node [anchor = east] () at (0.5, 1.5) {$v_3$};
        \node [anchor = east] () at (0.5, 2.0) {$v_4$};
        \node [anchor = east] () at (0.5, 2.5) {$v_5$};
    
    \end{tikzpicture}
}

\newcommand{\PathGraphRQCPerm}{
    \begin{tikzpicture}
		\foreach \y in {0.5, 1, ..., 2.5}
        {
            \draw [thick] (0.5, \y) -- (1.5, \y);
        }

        \foreach \x/\y in {1/0.5, 2.5/0.5, 4/0.5, 5.5/0.5} 
		  \node [fill = gray, minimum width = 0.5cm, minimum height = 0.75cm, anchor = south west, rounded corners] at (\x, \y - 0.125) {};

        \foreach \x in {2.5, 4, 5.5}
        {
            \foreach \y in {1.5, 2, 2.5}
            {
                \draw [thick] (\x, \y) -- (\x + 0.5, \y);
            }
        }

        \foreach \x in {0.5, 2, 3.5, 5, 6}
        {
            \begin{scope}[shift = {(\x, 0)}]
                \foreach \y in {0, 0.5, 1, 1.5}
                    \draw[thick] (1, 1 + \y) .. controls (1.5, 1 + \y) and (1.5, 0.5 + \y).. (2, 0.5 + \y);
                \draw[thick] (1, 0.5) .. controls (1.5, 0.5) and (1.5, 2.5).. (2, 2.5);
            \end{scope}
        }

        \node [anchor = east] () at (0.5, 0.5) {$v_1$};
        \node [anchor = east] () at (0.5, 1.0) {$v_2$};
        \node [anchor = east] () at (0.5, 1.5) {$v_3$};
        \node [anchor = east] () at (0.5, 2.0) {$v_4$};
        \node [anchor = east] () at (0.5, 2.5) {$v_5$};
    
    \end{tikzpicture}
}

\newcommand{\PathGraphRQCPermSWAP}{
    \begin{tikzpicture}
		\foreach \y in {0.5, 1, ..., 2.5}
        {
            \draw [thick] (0.5, \y) -- (1.5, \y);
        }

        \foreach \x/\y in {1/0.5, 3/0.5, 5/0.5, 7/0.5} 
		  \node [fill = gray, minimum width = 0.5cm, minimum height = 0.75cm, anchor = south west, rounded corners] at (\x, \y - 0.125) {};

        \foreach \y in {1.5, 2, 2.5}
        {
            \draw [thick] (1.5, \y) -- (2, \y);
        }

        \foreach \x in {3, 5, 7}
        {
            \foreach \y in {1.5, 2, 2.5}
            {
                \draw [thick] (\x, \y) -- (\x + 1, \y);
            }
        }

        \foreach \x in {1.5, 3.5, 5.5, 7.5}
        {
            \draw[thick, dotted] (\x, 0.5) .. controls (\x + 0.25, 0.5) and (\x + 0.25, 1).. (\x + 0.5, 1);
            \draw[thick, dotted] (\x, 1) .. controls (\x + 0.25, 1) and (\x + 0.25, 0.5).. (\x + 0.5, 0.5);
        }

        \foreach \x in {1, 3, 5, 7, 8}
        {
            \begin{scope}[shift = {(\x, 0)}]
                \foreach \y in {0, 0.5, 1, 1.5}
                    \draw[thick] (1, 1 + \y) .. controls (1.5, 1 + \y) and (1.5, 0.5 + \y).. (2, 0.5 + \y);
                \draw[thick] (1, 0.5) .. controls (1.5, 0.5) and (1.5, 2.5).. (2, 2.5);
            \end{scope}
        }

        \node [anchor = east] () at (0.5, 0.5) {$v_1$};
        \node [anchor = east] () at (0.5, 1.0) {$v_2$};
        \node [anchor = east] () at (0.5, 1.5) {$v_3$};
        \node [anchor = east] () at (0.5, 2.0) {$v_4$};
        \node [anchor = east] () at (0.5, 2.5) {$v_5$};

    \end{tikzpicture}
}

\newcommand{\PathGraphRQCStarPerm}{
    \begin{tikzpicture}
		\foreach \y in {1, 1.5, 2, 2.5}
        {
            \draw [thick] (0.5, \y) -- (1.5, \y);
        }

        \draw [thick] (0.5, 0.5) -- (6 + 1, 0.5);

        \foreach \x/\y in {1/0.5, 2.5/0.5, 4/0.5, 5.5/0.5} 
		  \node [fill = gray, minimum width = 0.5cm, minimum height = 0.75cm, anchor = south west, rounded corners] at (\x, \y - 0.125) {};

        \foreach \x in {2.5, 4, 5.5}
        {
            \foreach \y in {1.5, 2, 2.5}
            {
                \draw [thick] (\x, \y) -- (\x + 0.5, \y);
            }
        }

        \foreach \x in {0.5, 2, 3.5, 5}
        {
            \begin{scope}[shift = {(\x, 0)}]
                \foreach \y in {0.5, 1, 1.5}
                    \draw[thick] (1, 1 + \y) .. controls (1.5, 1 + \y) and (1.5, 0.5 + \y).. (2, 0.5 + \y);
                \draw[thick] (1, 1) .. controls (1.5, 1) and (1.5, 2.5).. (2, 2.5);
            \end{scope}
        }

        \begin{scope}[shift = {(6, 0)}]
            \foreach \y in {0, 0.5, 1, 1.5}
                \draw[thick] (1, 1 + \y) .. controls (1.5, 1 + \y) and (1.5, 0.5 + \y).. (2, 0.5 + \y);
            \draw[thick] (1, 0.5) .. controls (1.5, 0.5) and (1.5, 2.5).. (2, 2.5);
        \end{scope}

        \node [anchor = east] () at (0.5, 0.5) {$v_1$};
        \node [anchor = east] () at (0.5, 1.0) {$v_2$};
        \node [anchor = east] () at (0.5, 1.5) {$v_3$};
        \node [anchor = east] () at (0.5, 2.0) {$v_4$};
        \node [anchor = east] () at (0.5, 2.5) {$v_5$};

    \end{tikzpicture}
}

\newcommand{\CyclicPermutation}{
    \begin{tikzpicture}[baseline = {(current bounding box.center)}]
        \begin{scope}[shift = {(-0.5, 0)}]
            \foreach \y in {0, 0.5, 1, 1.5}
                \draw[thick] (1, 1 + \y) .. controls (1.5, 1 + \y) and (1.5, 0.5 + \y).. (2, 0.5 + \y);
            \draw[thick] (1, 0.5) .. controls (1.5, 0.5) and (1.5, 2.5).. (2, 2.5);
        \end{scope}

        \node [anchor = east] () at (0.5, 0.5) {$v_1$};
        \node [anchor = east] () at (0.5, 1.0) {$v_2$};
        \node [anchor = east] () at (0.5, 1.5) {$v_3$};
        \node [anchor = east] () at (0.5, 2.0) {$v_4$};
        \node [anchor = east] () at (0.5, 2.5) {$v_5$};
        
    \end{tikzpicture}
}

\newcommand{\StarGraphRQCLemSix}{
    \begin{tikzpicture}[baseline = {(current bounding box.center)}]
		\foreach \y in {0.5, 1, ..., 2.5} 
		  \draw [thick] (0.5, \y) -- (5, \y);
        
		\node [fill = gray, minimum width = 0.5cm, minimum height = 0.75cm, anchor = south west, rounded corners] at (1, 0.5 - 0.125) {};

		\node [fill = white, minimum width = 0.75cm, minimum height = 0.25cm, anchor = south west] at (2 - 0.125, 1 - 0.125) {};
		\node [fill = gray, minimum width = 0.5cm, minimum height = 1.25cm, anchor = south west, rounded corners] at (2, 0.5 - 0.125) {};

		\node [fill = white, minimum width = 0.75cm, minimum height = 0.75cm, anchor = south west] at (3 - 0.125, 1 - 0.125) {};
		\node [fill = gray, minimum width = 0.5cm, minimum height = 1.75cm, anchor = south west, rounded corners] at (3, 0.5 - 0.125) {};

        \node [fill = white, minimum width = 0.75cm, minimum height = 1.25cm, anchor = south west] at (4 - 0.125, 1 - 0.125) {};
		\node [fill = gray, minimum width = 0.5cm, minimum height = 2.25cm, anchor = south west, rounded corners] at (4, 0.5 - 0.125) {};

        \node [anchor = east] () at (0.5, 0.5) {$v_0$};
        \node [anchor = east] () at (0.5, 1.0) {$v_1$};
        \node [anchor = east] () at (0.5, 1.5) {$v_2$};
        \node [anchor = east] () at (0.5, 2.0) {$v_3$};
        \node [anchor = east] () at (0.5, 2.5) {$v_4$};

    \end{tikzpicture}
}

\newcommand{\StarGraphRQCLemSixDag}{
    \begin{tikzpicture}[baseline = {(current bounding box.center)}]
		\foreach \y in {0.5, 1, ..., 2.5} 
		  \draw [thick] (0.5, \y) -- (5, \y);
        
		\node [fill = gray, minimum width = 0.5cm, minimum height = 0.75cm, anchor = south west, rounded corners] at (1, 0.5 - 0.125) {};

		\node [fill = white, minimum width = 0.75cm, minimum height = 0.25cm, anchor = south west] at (2 - 0.125, 1 - 0.125) {};
		\node [fill = gray, minimum width = 0.5cm, minimum height = 1.25cm, anchor = south west, rounded corners] at (2, 0.5 - 0.125) {};

		\node [fill = white, minimum width = 0.75cm, minimum height = 0.75cm, anchor = south west] at (3 - 0.125, 1 - 0.125) {};
		\node [fill = gray, minimum width = 0.5cm, minimum height = 1.75cm, anchor = south west, rounded corners] at (3, 0.5 - 0.125) {};

        \node [fill = white, minimum width = 0.75cm, minimum height = 1.25cm, anchor = south west] at (4 - 0.125, 1 - 0.125) {};
		\node [fill = gray, minimum width = 0.5cm, minimum height = 2.25cm, anchor = south west, rounded corners] at (4, 0.5 - 0.125) {};

        \node [anchor = east] () at (0.5, 0.5) {$v_0$};
        \node [anchor = east] () at (0.5, 1.0) {$v_4$};
        \node [anchor = east] () at (0.5, 1.5) {$v_3$};
        \node [anchor = east] () at (0.5, 2.0) {$v_2$};
        \node [anchor = east] () at (0.5, 2.5) {$v_1$};

    \end{tikzpicture}
}

\newcommand{\PathGraphRQCLemSix}{
    \begin{tikzpicture}[baseline = {(current bounding box.center)}]
        \foreach \y in {0.5, 1, ..., 2.5} 
            \draw [thick] (0.5, \y) -- (5, \y);
    
        \foreach \x/\y in {1/0.5, 2/1, 3/1.5, 4/2} 
            \node [fill = gray, minimum width = 0.5cm, minimum height = 0.75cm, anchor = south west, rounded corners] at (\x, \y - 0.125) {};

        \begin{scope}[shift = {(4, 0)}]
            \foreach \y in {0, 0.5, 1, 1.5}
                \draw[thick] (1, 0.5 + \y) .. controls (1.5, 0.5 + \y) and (1.5, 1 + \y).. (2, 1 + \y);
            \draw[thick] (1, 2.5) .. controls (1.5, 2.5) and (1.5, 0.5).. (2, 0.5);
        \end{scope}
    
        \node [anchor = east] () at (0.5, 0.5) {$v_0$};
        \node [anchor = east] () at (0.5, 1.0) {$v_4$};
        \node [anchor = east] () at (0.5, 1.5) {$v_3$};
        \node [anchor = east] () at (0.5, 2.0) {$v_2$};
        \node [anchor = east] () at (0.5, 2.5) {$v_1$};
    
    \end{tikzpicture}
}

\newcommand{\PathGraphRQCLemSixFin}{
    \begin{tikzpicture}[baseline = {(current bounding box.center)}]
        \foreach \y in {0.5, 1, ..., 2.5} 
            \draw [thick] (1.5, \y) -- (6, \y);
    
        \foreach \x/\y in {5/0.5, 4/1, 3/1.5, 2/2} 
            \node [fill = gray, minimum width = 0.5cm, minimum height = 0.75cm, anchor = south west, rounded corners] at (\x, \y - 0.125) {};

        \begin{scope}[shift = {(-0.5, 0)}]
            \foreach \y in {0, 0.5, 1, 1.5}
                \draw[thick] (1, 1 + \y) .. controls (1.5, 1 + \y) and (1.5, 0.5 + \y).. (2, 0.5 + \y);
            \draw[thick] (1, 0.5) .. controls (1.5, 0.5) and (1.5, 2.5).. (2, 2.5);
        \end{scope}
    
        \node [anchor = east] () at (0.5, 0.5) {$v_0$};
        \node [anchor = east] () at (0.5, 1.0) {$v_4$};
        \node [anchor = east] () at (0.5, 1.5) {$v_3$};
        \node [anchor = east] () at (0.5, 2.0) {$v_2$};
        \node [anchor = east] () at (0.5, 2.5) {$v_1$};
    
    \end{tikzpicture}
}

\begin{document}

\begin{center}
\vspace*{2cm}
\title{Local random quantum circuits form approximate\\[4pt] designs on arbitrary architectures}

\author{Shivan Mittal${}^a$ and Nicholas Hunter-Jones${}^{a,b,c}$}
\address{${}^a$Department of Physics, University of Texas at Austin, Austin, TX 78712\\[6pt]
${}^b$Department of Computer Science, University of Texas at Austin, Austin, TX 78712\\[6pt]
${}^c$Stanford Institute for Theoretical Physics, Stanford, CA 94305}

\emails{\hspace*{-5mm} \href{mailto:shivan@utexas.edu }{\tt shivan@utexas.edu }}
\emails{\hspace*{-5mm} \href{mailto:nickrhj@utexas.edu}{\tt nickrhj@utexas.edu}}

\end{center}

\begin{abstract}
We consider random quantum circuits (RQC) on arbitrary connected graphs whose edges determine the allowed $2$-qudit interactions. Prior work has established that such $n$-qudit circuits with local dimension $q$ on $\1D$, complete, and $D$-dimensional graphs form approximate unitary designs, that is, they generate unitaries from distributions close to the Haar measure on the unitary group $U(q^n)$ after polynomially many gates. Here, we extend those results by proving that RQCs comprised of $O(\mathrm{poly}(n,k))$ gates on a wide class of graphs form approximate unitary $k$-designs. We prove that RQCs on graphs with spanning trees of bounded degree and height form $k$-designs after $O(|E|n\,\mathrm{poly}(k))$ gates, where $|E|$ is the number of edges in the graph. Furthermore, we identify larger classes of graphs for which RQCs generate approximate designs in polynomial circuit size. For $k \leq 4$, we show that RQCs on graphs of certain maximum degrees form designs after $O(|E|n)$ gates, providing explicit constants. We determine our circuit size bounds from the spectral gaps of local Hamiltonians. To that end, we extend the finite-size (or Knabe) method for bounding gaps of frustration-free Hamiltonians on regular graphs to arbitrary connected graphs. We further introduce a new method based on the Detectability Lemma for determining the spectral gaps of Hamiltonians on arbitrary graphs. Our methods have wider applicability as the first method provides a succinct alternative proof of \href{https://doi.org/10.1007/s00220-009-0873-6}{Commun.\ Math.\ Phys.\ 291, 257 (2009)} and the second method proves that RQCs on {\it any} connected architecture form approximate designs in quasi-polynomial circuit size.
\end{abstract}

\section{Introduction \label{sec:intro}}
Random quantum circuits (RQCs) serve as both a potential candidate for demonstrating exponential quantum advantage and a solvable model of local quantum chaotic dynamics. For example, they are used in demonstrating advantage of quantum over classical computing \cite{BoixoSergio2018CQ,AruteFrank2019QS,BoulandAdam2019OT,WuYulin2021StrongQuantum,ZhuQinling2022QC,MorvanA2023PT}, as analytically tractable models to study out-of-equilibrium physics and entanglement generation in many-body quantum systems \cite{OliveiraRoberto2006EG,DahlstenOscarCO2007TE,ZnidaricMarko2008EC,NahumAdam2017QE,NahumAdam2018OS,vonKeyserlingkCurtW2018OH,BensaJas2021FL}, as encoding circuits for quantum error correcting codes \cite{BrownWinton2013SR}, as models for scrambling and decoupling quantum information \cite{BrownWinton2012SS,BrownWinton2015DW}, and (thus) as models for information dynamics inside black holes \cite{HaydenPatrick2007BH,BrandaoFernandoGSL2021MO}. In most applications, one is interested in the ensemble averages of $k$-degree polynomials in the entries of a randomly selected unitary matrix and its complex conjugate. Some examples of such polynomials include $k^{\text{th}}$ statistical moment of observables and $(k/2)$-point out-of-time-ordered correlation functions. However, approximating typical unitary operators on $n$-qudit Hilbert spaces  is computationally and physically intractable because it requires $O(2^{2n})$ gates in a quantum circuit \cite{KnillEmanuel1995AB}. 
This exponential complexity can be avoided by using RQCs because they can implement unitary time evolutions that are sampled from a distribution close to the Haar distribution on the unitary group with only polynomial in $n$ gates \cite{DankertChristoph2009EA,GrossDavid2007ED}. In particular, RQCs on $n$ qudits with local dimension $q$ after $O(\mathrm{poly}(n, k) \log(1/\varepsilon))$ gates generate distributions over the unitary group $U(q^n)$ such that ensemble averages of degree $k$ polynomials in the entries of the unitary matrices are $\varepsilon$-close to the same averages computed using the Haar measure on $U(q^n)$ \cite{BrandaoFernandoGSL2016LR,HunterJonesNicholas2019UD,HaferkampJonas2021IS,HaferkampJonas2022RQ,HarrowAram2023AU}. Due to this property, the RQC-generated distributions over $U(q^n)$ are termed $\varepsilon$-approximate unitary $k$-designs, or it is simply said that RQCs form $\varepsilon$-approximate unitary $k$-designs.

Suppose vertices and edges of a graph denote qudits and pairs of qudits on which gates can act in the circuit, respectively. We call that graph is the architecture of the random quantum circuit. Then the mentioned references show that $O(\mathrm{poly}(n,k))$ (omitting $\log(1/\varepsilon)$ dependence which is fundamental) gates suffice to form approximate $k$-designs for RQCs with the following architectures: 1D line with open and closed boundary conditions, complete graph, and $D$-dimensional lattice in arbitrary dimension. Here, we extend that literature by proving that RQCs on a large class of graphs form approximate designs in $O(\mathrm{poly}(n,k))$ circuit size. The class of graphs for which we show this result directly includes all previous architectures and, more generally, bounded degree graphs with $O(\log{(n)})$ height spanning trees as well as other graphs that can be ``compressed'' to such graphs (as we define later). Furthermore, we show that $O(n^{O(\log(n))}\mathrm{poly}(k))$ gates suffice for RQCs with \textit{arbitrary} architectures to form approximate unitary $k$-designs. The polynomial dependence on $k$ in our bounds is inherited from existing results in 1D.

For several applications of RQCs, approximating low $k$ moments is sufficient. For example, purity of subsystems and the first two moments of observables inform us about entanglement generation and out-of-equilibrium behavior of random dynamics modeled by RQCs. Anti-concentration of RQC output distributions can be determined from the variance of the measurement probabilities and is thus a second moment quantity, potentially independent of the circuit architecture \cite{DalzellAlex2022RQ}. Anti-concentration, in turn, provides evidence for hardness \cite{BoulandAdam2019OT} and, concurrently, the classical tractability of Random Circuit Sampling \cite{AharonovDorit2023AP}, a task proposed to demonstrate quantum supremacy and implemented by Google and USTC. In these applications of approximating averages of low-degree polynomials of unitaries, the scaling of the number of gates in $n$ for fixed $k$ is important. Our results indicate that studying out-of-equilibrium behavior, entanglement generation, post-thermalization dynamics \cite{CHJR20} and the output distributions \cite{NietnerAlex2023FR} of random quantum circuits is possible on a wide class of architectures with only $O(\mathrm{poly}(n))$ gates. On the other hand, the scaling in $k$ of the required number of gates is important for proving the strong version of Brown-Susskind conjecture \cite{BrandaoFernandoGSL2021MO,HaferkampJonas2022RQ,oszmaniec2022saturation}, which states that circuit complexity (minimum circuit size required to construct a unitary operator using a universal gate set) of most unitary operators generated by RQCs increases linearly with the circuit size. Since our bound on the circuit size inherits its $k$-dependence directly from the $\1D$ case, if the strong version of the Brown-Susskind conjecture were proven for the $\1D$ case, our results would imply a proof for RQCs on arbitrary connected architectures.

For the first few moments, we provide a combination of analytical and numerical results. We use the finite-size criteria due to Knabe alongside our improved numerical scheme to find rigorous circuit size bounds for forming designs on arbitrary bounded degree connected architectures. Using the same approach, we provide a short proof of Ref.~\cite{HL08} showing the convergence of RQCs on complete graph architecture to unitary $2$-designs in $O(n^2)$ size. For general moments, we provide analytical size bounds for forming approximate $k$-designs on arbitrary connected graphs. Here, our main technical contribution is a lower bound on the spectral gap of a local Hamiltonian defined on an arbitrary connected graph of qudits with pairwise interaction. Determining the spectral gap of a local Hamiltonian is a hard problem and, hence, there does not exist a single procedure to find it other than to diagonalize the Hamiltonian, which is computationally expensive. We describe a novel technique to lower bound the spectral gap of {\it frustration-free} Hamiltonians on arbitrary connected graphs when the local terms of the Hamiltonian are invariant under left/right multiplication by the two-site permutation operator and the projector on to the ground states commutes with the $n$-site cyclic permutation operator. Our approach relies on recursive application of the Detectability Lemma and the Quantum Union Bound in a novel way that is independent of the techniques in \cite{AnshuAnurag2016SP,anshu2020higherD}, along with the properties of spanning trees of connected graphs. We begin with stating the definitions and notations that we use in \autoref{sec:not_and_def} followed by the motivation and informal summary of our results \autoref{sec:mot_and_res}. In later sections (\autoref{sec:knabe_any_g}--\autoref{sec:short_hl08}), we provide formal statements and proofs of our results. We conclude with an outlook in \autoref{sec:outlook}. In the appendices, we collect numerical calculations of spectral gaps, semi-classical approximations of spectral gaps, and some proofs for results in \autoref{sec:det_lem_app}.

\section{Notation and Definitions \label{sec:not_and_def}}

In this section, we give basic definitions and notations that we will use throughout the text. Let $n$ denote the number of qudits and let $q$ be their local dimension. Consider a graph $G(V, E)$, where $V$ denotes the set of vertices and $E$ denotes the set of unordered pairs of vertices that share an edge in the graph. Vertices in $V$ are identified with qudits in a quantum circuit (or rather their corresponding Hilbert spaces). We will simplify notation and denote a graph by $G$, unless the need to specify $V$ and $E$ arises. Moving forward, we will only consider connected graphs. We will denote the complex conjugate and Hermitian adjoint of a matrix $A$ by $\overline{A}$ and $A^\dagger$, respectively.

There are various notions of random quantum circuits that primarily differ in the choice of gate set and the pairs of qudits acted upon by unitary gates at each time step. We consider ``local'' random quantum circuits that are defined as follows.
\begin{definition}
    \label{def:loc_rqc}
    A local random quantum circuit on a graph $G(V, E)$ is defined to be a quantum circuit in which at each time step an edge $(i,j)$ is chosen uniformly at random from $E$, a $2$-site unitary gate $U$ is chosen randomly with respect to the Haar measure on $\mathcal{U}(q^2)$, and $U$ is applied to the circuit on qudits $i$ and $j$. We refer to $G(V, E)$ as the architecture of the local random quantum circuit. Random quantum circuits of size $t$ correspond to $t$ steps of this random process.
\end{definition}

We find it more convenient to talk about the size of a random circuit, i.e.\ the number of constituent gates, instead of the depth as circuit depth becomes somewhat ambiguous when dealing with more general circuit architectures. 

The approximate unitary design property of a distribution on the group of $n$-qudit unitaries is a statement about its convergence to the Haar measure on that group. Its definition in terms of the diamond norm is as follows.

\begin{definition}
    \label{def:str_aud}
    An RQC with architecture defined by a graph $G$ is said to form an $\varepsilon$-approximate unitary $k$-design if for any $\varepsilon > 0$, there exists a minimum size $\tau$ of the RQC such that for all $t \geq \tau$ the quantum channel $\Phi(G, n, k, t) (\ \cdot\ ) := \int (d\nu_{\mathrm{RQC}} (G))^{*t}\, U^{\otimes k} (\ \cdot\ ) U^{\dagger \otimes k}$ computed using the $t$-fold convolution of the probability measure over the unitary group induced by one step of the RQC, $(d\nu_{\mathrm{RQC}}(G))^{*t}$, and the quantum channel $\Phi(n, k)(\ \cdot\ ) := \int d\mu_{\mathrm{Haar}}\, U^{\otimes k} (\ \cdot\ ) U^{\dagger \otimes k}$ computed using the Haar measure over the unitary group, $d\mu_{\mathrm{Haar}}$, are $\varepsilon$ close in diamond norm, that is,
    \begin{equation}
        \big\| \Phi(G, n, k, t) - \Phi(n, k) \big\|_{\diamond} \leq \frac{\varepsilon}{q^{nk}}\,.
    \end{equation}
\end{definition}

An alternative version to \autoref{def:str_aud}, expressed in terms of a more easily computable norm, the operator norm, is presented next. We will compute/estimate the operator norm to upper bound the circuit size after which local RQCs form approximate unitary designs.

\begin{definition}
    \label{def:weak_aud}
    An RQC with architecture defined by graph $G$ is said to form an $\varepsilon$-approximate unitary $k$-design if for any $\varepsilon > 0$, there exists a minimum size $\tau$ of the RQC such that for all $t \geq \tau$ the moment operator $\Phi(G, n, k, t) := \int (d\nu_{\mathrm{RQC}} (G))^{*t}\, U^{\otimes k} \otimes \overline{U}^{\otimes k}$ computed using the $t$-fold convolution of the probability measure over the unitary group induced by one step of the RQC, $(d\nu_{\mathrm{RQC}}(G))^{*t}$, and the moment operator $\Phi(n, k) := \int d\mu_{\mathrm{Haar}}\, U^{\otimes k} \otimes \overline{U}^{\otimes k}$ computed using the Haar measure over the unitary group, $d\mu_{\mathrm{Haar}}$, are $\varepsilon$ close in operator norm, that is,
    \begin{equation}
        \big\|\Phi(G, n, k, t) - \Phi(n, k) \big\|_{\infty} \leq \frac{\varepsilon}{q^{2nk}}\,.
    \end{equation}
\end{definition}

\ni As the operator norm of the moment operators bounds the diamond norm of the corresponding channels at the expense of a factor of $q^{nk}$, proving that an RQC architecture satisfies \autoref{def:weak_aud} implies the approximate design condition in \autoref{def:str_aud}.\\

To compute/estimate the operator norm in \autoref{def:weak_aud}, we will benefit from defining a rescaled version of the moment operator and a Hamiltonian.

\begin{definition}
    \label{def:resc_mop}
    Consider an RQC with architecture defined by a graph $G(V, E)$ and the corresponding moment operators, $\Phi(G, n, k, t)$ and $\Phi(n, k)$. We define a rescaled moment operator as follows,
    \begin{align}
        M(G, n, k, t) = |E| \Phi(G, n, k, t),
    \end{align}
    and, in particular,
    \begin{align}
        \label{eq:def_mop}
        M(G, n, k, 1) = \sum_{(i, j) \in E} \int_{\mathcal{U}(q^2)} d\mu_{\mathrm{Haar}} U_{(i, j)} ^{\otimes k, k} \otimes \id_{[n] \backslash \{i, j\}} ^{\otimes k, k}\,
    \end{align}
    where $(\ \cdot\ )^{\otimes k, k} := (\ \cdot\ )^{\otimes k} \otimes \overline{(\ \cdot\ )}^{\otimes k}$, subscripts denote the Hilbert spaces on which the operators act, $\id_{X}$ denotes the identity operator on all sites with indices in the set $X$ and $|E|$ denotes the number of edges in the graph $G(V, E)$. Each term in the sum in Eq.~\eqref{eq:def_mop} is referred to as a local term of the moment operator $M(G, n, k, 1)$. We simplify notation by defining $M(G, n, k) := M(G, n, k, 1)$.
\end{definition}

The Hamiltonian is the rescaled moment operator of \autoref{def:resc_mop} with its spectrum inverted about its maximum eigenvalue. Defining such a Hamiltonian is useful because then we may use existing techniques to bound spectral gaps of Hamiltonians in order to upper bound the operator norm in \autoref{def:weak_aud}.

\begin{definition}
    \label{def:ham_gnk}
    We define a frustration-free local Hamiltonian denoted by $H(G, n, k)$ as follows,
    \begin{equation}
        \label{eq:def_ham}
        H(G, n, k) := |E|\ \id_{[n]} ^{\otimes k, k} - M(G, n, k) = \sum_{(i, j) \in E} \left(\id_{[n]} ^{\otimes k, k} - \int_{\mathcal{U}(q^2)} d\mu_{\mathrm{Haar}}\, U_{(i, j)} ^{\otimes k, k} \otimes \id_{[n] \backslash \{i, j\}} ^{\otimes k, k}\right)\,,
    \end{equation}
    where the notation is borrowed from the previous definitions. Each term (enclosed in parenthesis) in the sum in Eq.~\eqref{eq:def_ham} is referred to as a local term of the Hamiltonian $H(G, n, k)$. We denote the spectral gap of $H(G, n, k)$ by $\Delta(G, n, k)$. We refer the reader to Ref.~\cite{BrandaoFernandoGSL2016LR} for the proof of frustration-freeness of $H(G, n, k)$.
\end{definition}

The primary existing technique of lower bounding the spectral gap of local Hamiltonians that we use is due to Knabe \cite{knabe1988energy}. It relates the spectral gap of the Hamiltonian restricted to a finite subset of total sites to that of the Hamiltonian on all sites. Thus, the following definitions of ``neighborhood'' of a vertex in a graph and restriction of the Hamiltonian to that neighborhood will prove useful.

\begin{definition}
    Consider a graph $G(V, E)$. The neighborhood $N(v)$ of any vertex $v \in V$ is defined as follows,
    \begin{equation}
        N(G, v) := \{u : (v, u) \in E \text{ or } u = v\}\,.
    \end{equation}
    We define the restriction of $H(G, n, k)$ to the neighborhood of a vertex $v$ in the natural way by $H(N(G, v), |N(G, v)|, k)$, and denote its spectral gap by $\Delta(N(G, v), |N(G, v)|, k)$, where $|N(G, v)|$ denotes the number of vertices in $N(G, v)$, or, equivalently, one plus the degree of $v$.
\end{definition}

In \autoref{sec:det_lem_app}, we will introduce another method to determine spectral gaps of Hamiltonians of the form given in \autoref{def:ham_gnk} that relies on the Detectability Lemma and, its converse, the Quantum Union Bound as defined below.

\begin{lemma}[Detectability Lemma \cite{AharonovDorit2008DL,AnshuAnurag2016SP}]
    \label{def:dl}
    Consider a set of projectors $\{Q_i\}_{i = 1} ^m$ and a Hamiltonian $H = \sum_{i = 1} ^m Q_i$, with spectral gap $\Delta$. Assume that each $Q_i$ does not commute with at most $g$ other projectors in $\{Q_i\}_{i = 1} ^m$. For any $\pi \in S_m$, define $\ket{\phi} := \prod_{i = 1} ^m (1 - Q_{\pi(i)}) \ket{\psi}$, then
    \begin{align}
        \label{eq:dl}
        \left\lVert \prod_{i = 1} ^m (1 - Q_{\pi(i)}) \ket{\psi} \right\rVert^2 \leq \cfrac{1}{1 + \cfrac{\bra{\phi} H \ket{\phi}}{g^2 \langle \phi | \phi \rangle}} \leq \cfrac{1}{1 + \cfrac{\Delta}{g^2}}\,.
    \end{align}
\end{lemma}
\begin{lemma}[Quantum Union Bound \cite{GaoJingliang2015Quantum,AnshuAnurag2016SP}] 
    \label{def:qub}
    Consider a set of projectors $\{Q_i\}_{i = 1} ^m$ and a Hamiltonian $H = \sum_{i = 1} ^m Q_i$ with spectral gap $\Delta$. For any $\pi \in S_m$
    \begin{align}
        \label{eq:qub}
        \left\lVert \prod_{i = 1} ^m (1 - Q_{\pi(i)}) \ket{\psi} \right\rVert^2 \geq 1 - 4 \bra{\psi} H \ket{\psi} \geq 1 - 4 \Delta\,.
    \end{align}
\end{lemma}
\ni We will refer to the left hand sides of Eqs.~\eqref{eq:dl} and~\eqref{eq:qub} as the Detectability Lemma norm. We will refer to the operator inside those norms as the Detectability Lemma operator.

Now we describe the landscape of results in the current literature, motivate an open question and report our partial progress.

\section{Motivation and Results \label{sec:mot_and_res}}
In the seminal work of Ref.~\cite{BrandaoFernandoGSL2016LR}, it was proved that local and parallel/brickwork RQCs on 1D graphs form $\varepsilon$-approximate unitary $k$-designs in $O(n^2 k^{11})$ circuit size (or, equivalently $O(n k^{11})$ circuit depth). Subsequently, the size bound for brickwork RQCs on 1D graphs was improved in Ref.~\cite{HaferkampJonas2022RQ} to $O(n^2 k^{5+o(1)})$. For large local dimension, an upper bound of $O(n^2 k)$ on the design size of brickwork RQCs on 1D graphs was proved in Ref.~\cite{HunterJonesNicholas2019UD} and improved in Ref.~\cite{HaferkampJonas2021IS} for the local dimensions larger than $\Omega(k^2)$. The upper bounds in the last two references are almost optimal in both $n$ and $k$ by comparison with the lower bound of $\Omega(n k / \log(nk))$ proved in Ref.~\cite{BrandaoFernandoGSL2016LR}. Beyond results for RQCs on 1D graphs, Ref.~\cite{HaferkampJonas2021IS} proved that $O(n^2 \log(n) {\rm poly}(k))$ size local RQCs on complete graphs form approximate designs. Furthermore, an upper bound of $O(n^{1 + 1/D} \mathrm{poly}(k))$ on the size of parallel RQCs on hypercubic lattices in $D$ dimensions was proved in Ref.~\cite{HarrowAram2023AU}. Therefore, the existing literature considers RQCs on three broad classes of graphs---1D, complete and $D$-dimensional graphs---and concludes that those RQCs form $\varepsilon$-approximate unitary $k$-designs for sizes that scale in $n$ and $k$ as given above. Thus, we are naturally led to ask,
    \begin{displayquote}
        At what circuit size do random quantum circuits on arbitrary connected graphs form $\varepsilon$-approximate unitary $k$-designs?
    \end{displayquote}
    We will follow the approach of Refs.~\cite{ZnidaricMarko2008EC,BrownWinton2010CR,BrandaoFernandoGSL2010EQ,BrandaoFernandoGSL2016LR} and derive circuit size upper bounds from the spectral gap $\Delta(G, n, k)$ of the local Hamiltonian $H(G, n, k)$. We recall an observation used explicitly in Ref.~\cite{oszmaniec2022saturation} that one can use the spectral gap for $\1D$ RQC to lower bound the spectral gap for graphs that contain a Hamiltonian path. Since the respective Hamiltonians share the same ground space, the spectral gap can only increase by adding local terms to the Hamiltonian for 1D graph consisting of $n$ sites (similar observations were also made in Refs.~\cite{BrandaoFernandoGSL2010EQ,BrownWinton2012SS,Onorati17}). 
    Nevertheless, this observation naively offers little insight into the nature of spectral gaps for arbitrary graphs and cannot control the gaps for graphs without a Hamiltonian path.
    
    Intuitively, one expects that RQCs which allow for mixing across any pair of qudits and across only nearest-neighbor qudits on a line (that is, RQCs on complete and 1D graphs) to exhibit maximal and minimal mixing properties, respectively. Therefore, one might think that RQCs on all other connected graphs would exhibit mixing properties that interpolate between those of RQCs on 1D and complete graphs. This empirical reasoning leads us to refine our original question to the following,
    \begin{displayquote}
        Do local RQCs on arbitrary $n$-vertex connected graphs form $\varepsilon$-approximate unitary $k$-designs at least as fast as RQCs on $\1D$ graphs, that is, after $O(n^2)$ gates? 
    \end{displayquote}
    We provide partial progress towards answering the above question. 
    Consider an arbitrary $n$-vertex connected graph $G(V, E)$. 
    We seek to find if local RQCs with architecture given by $G$ form $\varepsilon$-approximate unitary $k$-designs in $O(n^2)$ size. It suffices to show that the ratio of the spectral gap $\Delta(G, n, k)$ of the Hamiltonian $H(G, n, k)$ and the number of edges is $\Delta(G, n, k) / |E| = \Omega(1/n)$.

    As defined, our random circuits involve the application of 2-site unitaries drawn from the Haar measure on the 2-site unitary group $\mathcal{U}(q^2)$. All of our results immediately extend to random quantum circuits constructed from gates drawn randomly from any universal gate set $\mathsf{G}$, consisting of single and two-qubit gates. This follows from the independence of the local gap for universal gate sets closed under inverses and consisting of algebraic entries \cite{BourgainGamburd}, which guarantees that if RQCs with Haar-random gates form $k$-designs in $O({\rm poly}(n,k))$ size, then so do RQCs with gates drawn from $\mathsf{G}$, only at the expense of a gate-set dependent constant. Moreover, following the results of Refs.~\cite{Varju13,OSH20} we can drop the restrictions on the universal gate set at the further expense of an $n \log^2(k)$ factor.\\

    \ni \textbf{Result~1:}\ We prove a Knabe bound for the spectral gap of a frustration-free Hamiltonian on any connected graph. Specifically, in \autoref{thm:knabe_any_connected_graph} we relate $\Delta(G, n, k)$ to the spectral gap of the Hamiltonian restricted to the neighborhood of any vertex in $G$. 
    For any vertex $v$ in $G$, we denote its neighborhood by $N(G, v)$, the restriction of the Hamiltonian to that neighborhood by $H(N(G, v), |N(G, v)|, k)$, and its spectral gap by $\Delta(N(G, v), |N(G, v)|, k)$. 
    \autoref{thm:knabe_any_connected_graph} requires
    \begin{equation}
        \frac{\min_{v \in V} \Delta\big(N(G, v), |N(G, v)|, k\big) - 1/2}{|E|} = \Omega\left(\frac{1}{n}\right)\,,
    \end{equation}
    for $\Delta(G, n, k) / |E| = \Omega(1/n)$. 
    Since the neighborhood $N(G, v)$ of any vertex $v$ in a connected graph $G$ is a star graph, \autoref{thm:knabe_any_connected_graph} reduces the original problem of finding a lower bound on $\Delta(G, n, k)$ to that of finding a greater than $1/2$ lower bound on the spectral gap of the Hamiltonian on a star graph.\\
    
    \ni \textbf{Result~2:}\ By computing star graph gaps, we prove that random quantum circuits on any graph of certain bounded degrees form approximate $k$-designs for $k\leq 4$ after $c |E|(2nk+\log(1/\ep))$ gates, giving explicit (good) constants $c$. The point of this approach is to show that {\it realistic} circuit sizes suffice for low degree designs, as our approach for general moments involves reducing to the 1D spectral gap, along with the large constants involved.
    
    For an $n_{\star}$-vertex star graph $G_{\star}$, we denote the Hamiltonian by $H(G_{\star}, n_{\star}, k)$ and its spectral gap by $\Delta(G_{\star}, n_{\star}, k)$. 
    In \autoref{prop:exact_star_k2_n3} and \autoref{prop:exact_star_k2_n4}, we compute the first few spectral gaps of the star graph Hamiltonian for $k=2$. This result in turn implies that local RQCs on arbitrary connected graphs with $|E|$ edges and maximum degree $3$ form approximate $2$-designs in $O(|E|n)$ size. 
    We then numerically compute the spectral gaps for star graphs with various values of $k, q$ and $n_{\star}$, and our findings are reported in \autoref{tab:star_gaps} (in \autorefapp{app:num_dat}). 
    From that table, for each particular triplet of $k, q$ and $n_{\star}$, we can infer that local RQCs on arbitrary connected graphs with $|E|$ edges, maximum degree equal to $n_{\star} - 1$ and local dimension $q$ form $\varepsilon$-approximate unitary $k$-designs in $O(|E|n)$ size. 
    This result rigorously extends the previous result with the help of numerics to indicate that RQCs on arbitrary graphs of maximum degree $38$ form approximate unitary $2$-designs (\autoref{cor:size_bounddeg_knabe}).
    
    In \autoref{fig:plot_star_gaps}, we plot the numerically computed star graph gaps. The observation that for $k < 4$, spectral gaps are strictly greater than $1/2$, leads us to state the following conjecture,
    \begin{conjecture}
        \label{con:star_gap}
        The second and third moment ($k=2$ and $k=3$) star graph gaps $\Delta(G_{\star}, n_{\star}, k)$ are strictly greater than $1/2$.
    \end{conjecture}
    \ni We also formulate a somewhat looser conjecture, applicable to higher moments:
    \begin{conjecture}
        \label{con:kstar_gap}
        Below a certain threshold for the ratio, $k^2 / q^{n_\star}$, the star graph spectral gaps $\Delta(G_{\star}, n_{\star}, k)$ are strictly greater than $1/2$.
    \end{conjecture}
    \begin{figure}[ht!]
        \centering
        \begin{tikzpicture}
            \node[anchor = base] (a) at (0, 0) {\includegraphics[width = 0.48\textwidth]{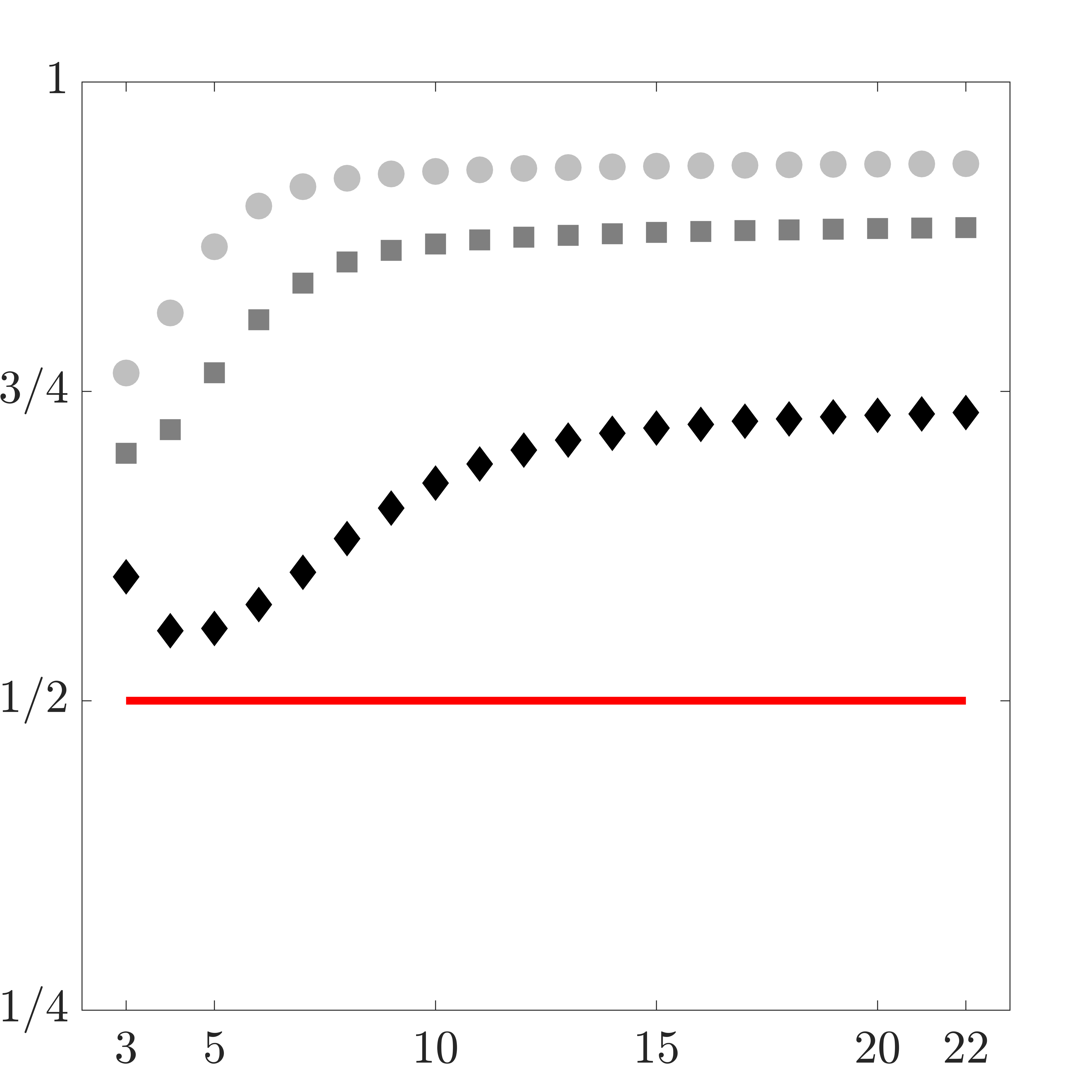}};
            \node[anchor = base, align = left, rotate = 90] at (a.west) {$\Delta(G_{\star}, n_{\star}, k)$};
    	    \node at (a.south) {$n_{\star}$};
        \end{tikzpicture}
        \caption{The vertical axis represents the spectral gaps of the Hamiltonian defined in \autoref{def:ham_gnk} on star graphs with $n_{\star}$ vertices for $k = 2$. The horizontal axis tracks the number of vertices in the star graph. The thick red line denotes \autoref{con:star_gap}. The black diamonds, dark gray squares and light gray circles are the numerically computed values of the said spectral gaps for local dimension $q = 2$, $q = 3$ and $q = 4$, respectively. The data points are provided in \autoref{tab:star_gaps} along with a remark on the method for computing the same.}
        \label{fig:plot_star_gaps}
    \end{figure}
	
    If true, the conjecture along with \autoref{thm:knabe_any_connected_graph} immediately proves that RQCs with $O(|E|nk)$ size form $\varepsilon$-approximate unitary $k$-designs on arbitrary graphs $G(V, E)$. 
    A full proof of \autoref{con:star_gap} or \autoref{con:kstar_gap} is currently not at hand. 
    Instead, we directly prove bounds for the spectral gap $\Delta(G, n, k)$ for the Hamiltonian $H(G, n, k)$ on arbitrary connected graphs $G$, as is summarized in the following results.\\

    \ni \textbf{Result~3:}\ In \autoref{cor:dl_cst_rreg}, \autoref{cor:dl_st_rreg_log_dep} and \autoref{cor:dl_st_rreg_con_dep}, we identify a large class of graphs and prove lower bounds on their spectral gaps, $\Delta(G, n, k)$, such that local RQCs on such graphs form $\varepsilon$-approximate unitary $k$-designs in $O(\mathrm{poly}(n,k))$ circuit size (\autoref{thm:size_bounddeg_r_dep_logn} and \autoref{thm:size_any_graph_anyq}). In particular, we show an $O(|E| n k^{5 + o(1)})$ bound on the circuit size for local RQCs on connected graphs with $|E|$ edges and with spanning trees of constant maximum degree and height. For a weaker constraint of $O(\log{(n)})$ instead of constant height, we show $O(\mathrm{poly}(n,k))$ bound on the circuit size albeit with higher polynomial scaling in $n$. We extend the same bound for graphs whose spanning trees can be ``compressed'' (in a way that we define later \autoref{sec:det_lem_app} and \autorefapp{app:dep_and_cst}) to constant maximum degree and $O(\log{(n)})$ height. Note that previously, $O(\mathrm{poly}(n))$ circuit size upper bounds to form approximate unitary designs with RQCs were only available for $\1D$, complete or $D$-dimensional architectures. 
    Our result rigorously establishes that $O(\mathrm{poly}(n))$ circuit size suffices to form approximate unitary designs on a large class of graphs, which directly includes all the graphs for which results were previously known in the literature.
    
    We consider local RQCs (\autoref{def:loc_rqc}) whereas Ref.~\cite{HarrowAram2023AU} considers parallel RQCs with a specific arrangement of gates. This parallelized architecture seems crucial to the proof technique and it is not clear how their analysis would change if one were to pick any edge of, say, a 2D grid uniformly at random at each time step (that is, if one were to construct a local RQC like the ones we consider).
    
    An important feature of the circuit size bounds of our result is that the $k$ dependence is directly inherited from the similar bound for RQCs with $\1D$ architecture. Thus, any improvement in the latter implies corresponding improvement in the former.\\

    \ni \textbf{Result~4:}\ In \autoref{thm:dl_any_graph}, we show that for local RQCs on arbitrary connected graphs, $\Delta(G, n, k) = \Omega(1 / n^{O(\log (n))})$. This implies that local RQCs on arbitrary connected graphs form $\varepsilon$-approximate unitary $k$-designs in $O(n^{O(\log (n))} \mathrm{poly}(k))$ circuit size (\autoref{thm:size_any_graph} and \autoref{thm:size_any_graph_anyq}). 
    This is an extremely weak bound on circuit size, as it is quasi-polynomial in $n$ and does not prove our expectation that $O(n^2)$ circuit size suffices to form unitary designs with RQCs on arbitrary architectures. 
    Nonetheless, it is the first result that provides a non-trivial and rigorous upper bound on the circuit size of RQCs on completely arbitrary connected architectures to form approximate unitary designs. 
    As in the previous result, the circuit size bound of this result inherits its $k$-dependence directly from the corresponding bound for RQCs on $\1D$ graphs. 
    Therefore, a consequence of this result is that if the strong version of Brown-Susskind conjecture about the linear complexity growth of unitaries sampled by RQCs is true for RQCs on $\1D$ graphs, then, so is it true for RQCs on arbitrary connected graphs.\\

    \ni \textbf{Result~5:}\ In \autoref{cor:opt_size}, we prove that when measuring the approximate convergence to unitary $k$-designs using spectral norms (\autoref{def:weak_aud}), the least upper bound we could have proven is $O(n^2).$ This indicates two things; first, that one would need to measure convergence in a stronger norm, such as the diamond norm, to prove sub-quadratic size upper bounds; and, second, that when restricted to spectral norm definition of approximate unitary designs, in Result~3, we identify the optimal size upper bound of $O(n^2)$ for a strictly larger class of graphs than for which such results were known before.\\

    \ni \textbf{Result~6:} Lastly, we provide a short proof that complete-graph RQCs form approximate 2-designs in $O(n^2)$ size. To show this we derive a finite-size criteria for complete graphs (\autoref{thm:cgknabe}). Computing the exact $n=3$ complete graph second moment gap, by diagonalizing the moment operator, and inserting into the Knabe bound gives a complete-graph gap lower bound for all $n$ (\autoref{thm:cgk2gap}). This approach provides an alternative proof to that first presented in Ref.~\cite{HL08}.\\

    Our motivating question remains open, but we prove rigorous non-trivial size upper bounds for RQCs on arbitrary connected graphs to generate approximate unitary designs. 
    Thus, our results rigorously justify the use of non-local circuit architectures for any application that requires approximate unitary design property of RQCs. 
    In particular, our results provide support to quantum advantage experiments on wide variety of circuit architectures, beyond those implemented by Google and USTC. 
    Furthermore, our results reduce the proof of the strong version of the Brown-Susskind conjecture for any connected architecture RQC to that for $\1D$ RQCs. 
    In the following sections, we provide the formal statements of theorems and the summaries of our approaches to derive the same. 
    In \autoref{sec:knabe_any_g}, we present \autoref{thm:knabe_any_connected_graph}, \autoref{prop:exact_star_k2_n3} and \autoref{prop:exact_star_k2_n4} with their short proofs. 
    In \autoref{sec:det_lem_app}, we provide \autoref{thm:dl_any_graph}, \autoref{cor:dl_cst_rreg},
    \autoref{cor:dl_st_rreg_log_dep}, \autoref{cor:dl_st_rreg_con_dep}, \autoref{thm:size_any_graph}, \autoref{thm:size_bounddeg_r_dep_logn} and \autoref{thm:size_any_graph_anyq}. We defer the proofs of this section to the appendices. In \autoref{sec:upb_gaps}, we state and prove \autoref{cor:opt_size}.

\section{Gaps for arbitrary connected graphs from stars \label{sec:knabe_any_g}}
In this section, we extend the method of finite-size criteria of proving spectral gaps of frustration-free local Hamiltonians on $\1D$ lattices due to Knabe \cite{knabe1988energy} to those on arbitrary graphs. Previously, Knabe's method had been generalized beyond 1D spin chains with periodic boundary conditions to other boundary condition and higher-dimensional lattices \cite{GM15,LM18,lemm2019higherD,anshu2020higherD,Lemm22}. We derive a finite-size criteria, which is applicable to {\it all} graphs, that uses the spectral gaps on star graphs (neighborhood of a vertex in a graph) to lower bound the gap on the graph of interest. We emphasize that for a general frustration-free Hamiltonian, the criteria might be too strong to say anything interesting as one needs to control all gaps on substars up to the maximal degree of the graph and show they are strictly greater than $1/2$. Nevertheless, this criteria is enough for us to prove lower bounds on the spectral gaps of Hamiltonians as defined in \autoref{def:ham_gnk} on arbitrary graphs $\g$ and derive size bounds for RQCs on bounded degree graphs to generate approximate unitary designs. 

Previously, Knabe's method has been used to prove spectral gaps of frustration-free local Hamiltonians on regular lattices, where there exists a clear notion of repeating unit cells. That notion is absent for generic connected graphs and requires keeping track of several possible distinct building blocks of the graphs. The main ideas behind our extension of the method is to first realize that star graph neighborhoods of every vertex in an arbitrary graph is a useful decomposition of the arbitrary graph, and then keep track of combinatorial factors that appear in adding up the Hamiltonians on those neighborhoods.

\subsection{Bound for any connected graph}
Let $\g(V,E)$ be an $n$-vertex undirected connected graph. We can formulate a Knabe-type finite-size criteria for frustration-free Hamiltonians on any connected graph. 
\begin{theorem}
\label{thm:knabe_any_connected_graph}
Let $H(G,n,k) = \sum_{(i,j)\in E} h_{i,j}$ be a frustration-free Hamiltonian on a connected graph $\g$ with $n$ sites, where the local terms are specified by a fixed 2-site projector $P$ as $h_{i,j} = (P)_{i,j}\otimes \iden_{[n]\setminus\{i,j\}} $. Let $\Delta(G,n,k) = \Delta(H(G,n,k))$ denote its spectral gap, then
\begin{equation}
    \Delta(G,n,k) \geq 2\bigg( \min_{v\in V}\,\Delta\big(G_{\star}, \deg(v)+1 ,k\big) - \frac{1}{2}\bigg)\,,
\end{equation}
where $\Delta(G_\star, m ,k)$ is the spectral gap on the star graph with $m$ sites, i.e.\ one internal node and $m-1$ leaves, and $\deg(v)$ is the degree of the vertex $v\in\g$.
\end{theorem}

In the following, the value of the moment $k$ will not play a role and, thus, suppressing it, we write $H(G,n)=H(G,n,k)$.

\begin{proof}[Proof of \autoref{thm:knabe_any_connected_graph}]
Following the approach in \cite{knabe1988energy}, we proceed by lower bounding the square of the Hamiltonian $H(G, n)$, observing that finding a lower bound in terms of the Hamiltonian $H(G, n)^2\geq \gamma H(G, n)$ implies a lower bound on the spectral gap as $\Delta(G, n)\geq \gamma$.
We start by squaring the Hamiltonian
\begin{equation}
    \big(H(G,n)\big)^2 = \sum_{(i,j)\in\g}h_{i,j} + \sum_{\substack{(i,j),(k,l)\in E\\ |\{i,j\}\cap \{k,l\}|=1}}\{h_{i,j},h_{k,l}\} + \sum_{\substack{(i,j),(k,l)\in E\\ \{i,j\}\cap \{k,l\}=0}}\{h_{i,j},h_{k,l}\} = H(G,n) + Q + R\,,
\end{equation}
where $Q$ contains all the anti-commutators of the overlapping Hamiltonian terms (that is those $h_{i, j}$ and $h_{k, l}$ such that either $i$ or $j$ equals $k$ or $l$) and $R$ contains all the anti-commutators of the non-overlapping terms (that is those $h_{i, j}$ and $h_{k, l}$ such that none of $i$ or $j$ equals $k$ or $l$). We define the subsystem operator by squaring all the edges at a vertex and then summing over all vertices as
\begin{equation}
    A(G,n) = \sum_{i\in V} \bigg( \sum_{j:\, (i,j)\in E} h_{i,j}\bigg)^2 = 2H(G,n) + Q\,,
\end{equation}
as every edge contains exactly two vertices and every non-commuting term arises from the overlap at exactly one vertex. We now expand the operator $A(\g,n)$ as a sum over all vertices of a fixed degree
\begin{align}
    A(G,n) &= \sum_{g\in \{\deg(v)\,:\, v\in V\}} \sum_{\substack{i\in V\\ \deg(i)=g}} \bigg( \sum_{j: (i,j)\in E} h_{i,j}\bigg)^2 \\
    &\geq \sum_{g\in \{\deg(v)\,:\, v\in V\}}  \Delta\big(H(G_\star, g+1)\big) \Bigg(\sum_{\substack{i\in V\\ \deg(i)=g}} \sum_{j: (i,j)\in E} h_{i,j}\Bigg)\,,
\end{align}
where $\big( \sum_{j: (i,j)\in E} h_{i,j}\big)^2\geq \Delta(H(G_\star, g+1)) \big( \sum_{j: (i,j)\in E} h_{i,j}\big) $ for vertices of $\deg(i)=g$ and where $\Delta(H(G_\star, m))$ is the spectral gap of the star graph Hamiltonian on $m$ sites. Lower bounding $\Delta(H(G_\star, g+1))$ by taking the minimum over all star gaps, we then find that
\begin{equation}
    A(\g,n) \geq 2H(\g,n) \min_{v\in V}\Delta\big(H(G_\star,\deg(v)+1)\big)\,,
\end{equation}
from which the theorem follows.
\end{proof}

The result is that the spectral gap of a Hamiltonian $H(G,n)$ on any connected graph $G$ may be lower bounded in terms of the spectral gaps of substars. If one can control all gaps $\Delta(G_*,n)$ up to the maximal degree of the graph $G$ and show that they are strictly greater than $1/2$, then the above theorem implies a spectral gap lower bound $H(G,n)$. Similarly, computing all gaps for star graphs up to $\maxdeg+1$ vertices and showing that $\min_{1\leq g \leq \maxdeg} \Delta(H(G_*,g+1))>1/2$ proves that the Hamiltonian is gapped on all graphs of maximal degree $\kappa$.

Turning back to random quantum circuits, we formulate a conjecture for the second and third moment star graph gaps which would allow us to prove that RQCs form approximate designs on all graphs. The low moment star graph conjecture, stated in \autoref{con:star_gap}, is that the second and third moment ($k=2$ and $k=3$) star graph gaps $\Delta(G_{\star}, n, k)$ are strictly greater than $1/2$ for all $n$.

In the following we give evidence for this conjecture, both from numerics of these low moment spectral gaps in conjunction with a Knabe bound for star graphs and asymptotics from the semiclassical limit of a spin model, as well as analytic results for low moment gaps.

We also show why the conjecture cannot be true for arbitrary moments, as at $k=4$ for local qubits the star graph gaps already start decreasing. Still, we further conjecture that the star gaps remain greater than 1/2 so longer as the local dimension $q$ is taken to be large with respect to the moment.

\subsection{Evidence for the low-moment star graph conjecture}
Numerically (refer to \autoref{tab:star_gaps}), we have computed the star graph gaps of $H(\g_{\star}, n_{\star})$ up to $n_{\star}=22$ for $k = 2$, up to $n_{\star} = 9$ for $k = 3$, up to $n_{\star} = 5$ for $k = 4$ and for $n_{\star} = 3$ for $k = 5$. Thus, we have a lower bound on $\Delta(H(\g, n))$ so long that the maximal degree of $\g$ is less than $22$ for $k = 2$, $9$ for $k = 3$ and $5$ for $k = 4$. Note that $n = 3$ gap for $k = 5$ is not useful because its value is $1/2$ to numerical precision and thus does not provide a non-trivial lower bound by inserting in \autoref{thm:knabe_any_connected_graph}. The values of spectral gaps $\Delta(H(\g, n, k))$ for arbitrary graphs $\g$ with bounded degree are provided in \autoref{tab:any_g_from_star_gaps}.
\begin{table}[ht!]
    \centering
    \begin{tabular}{|| r | r | r | r ||}
        \hline
        $k$ & $n_\star$ & $\Delta(\g_{\star}, n_\star, k)$ & $\Delta(\g, n, k)$ \\
        \hline
        2 & 22 & 0.7328 & 0.4656 \\
        \hline
        3 & 9 & 0.6556 & 0.3112 \\
        \hline
        4 & 5 & 0.5583 & 0.1166 \\
        \hline
    \end{tabular}
    \caption{The first column represents the choice of $k$, the second column specifies the maximum system size $n_\star$ that we were able to compute spectral gaps for, the third column contains the corresponding spectral gaps of $n_\star$-vertex star graphs Hamiltonians and the final column lists the lower bound on the spectral gap for Hamiltonians on arbitrary graphs $\g$ with maximum degree equal to $n_\star - 1$. Note that there is no restriction on the number of vertices in $\g$.}
    \label{tab:any_g_from_star_gaps}
\end{table}
\begin{table}[ht!]
    \centering
    \begin{tabular}{|| r | r | r | r | r | r ||}
        \hline
        $k$ & $m_\star$ & $\Delta(\g_{\star}, m_\star, k)$ & Boosted $n_\star$ & Boosted $\Delta(\g_{\star}, n_\star, k)$ & $\Delta(\g, n, k)$ \\
        \hline
        2 & 22 & 0.7328 & 39 & 0.5057 & 0.0114 \\
        \hline
        3 &  9 & 0.6556 & 12 & 0.5080 & 0.0160 \\
        \hline
    \end{tabular}
    \caption{The first column represent the choice of $k$ and the second column specifies the maximum system size $m_\star$ that we were able to compute spectral gaps for (these are same as in \autoref{tab:any_g_from_star_gaps}). The third column contains the maximum $n_\star$, referred to as ``boosted'' $n_\star$ such that the corresponding boosted spectral gap in the fourth column, denoted by ``boosted $\Delta(H(\g_\star, n_\star, k))$,'' computed using Eq.~\eqref{eq:star_g_knabe} is greater than the finite-size-criteria of \autoref{thm:knabe_any_connected_graph}. The final column lists the lower bound on the spectral gap of Hamiltonians on arbitrary graphs $\g$ with maximum degree equal to ``boosted'' $n_\star - 1$. Note that there is no restriction on the number of vertices in $\g$. For $k = 4$ and $m_\star = 3$, the numerically computed spectral gap is $1/2$ and for all values of $n_\star > m_\star$ the boosted gap is below $1/2$ and, hence, do not satisfy the finite-size-criteria of \autoref{thm:knabe_any_connected_graph}.}
    \label{tab:boosted_any_g_from_star_gaps}
\end{table}

In the above, we directly input the numerically computed gaps in \autoref{thm:knabe_any_connected_graph}. We can instead input the numerically computed gaps in for $m$-vertex star graphs in \autoref{thm:star_g_knabe},
\begin{align}
    \label{eq:star_g_knabe}
    \Delta(H(\g_{\star}, n)) \geq \frac{n-2}{m-2} \left( \Delta(H(\g_{\star}, m)) - \frac{n-m}{n-2}\right),
\end{align}
to find lower bounds on spectral gaps for larger $n_\star$-vertex star graphs and then input those gaps in \autoref{thm:knabe_any_connected_graph}. This way we can extend the applicability of our numerical calculations and conclude about spectral gaps of higher bounded-degree graphs. The values of spectral gaps $\Delta(H(\g, n, k))$ for arbitrary graphs $\g$ with higher bounded degree by this method are provided in \autoref{tab:boosted_any_g_from_star_gaps}.

The spectral gaps presented in \autoref{tab:any_g_from_star_gaps} and \autoref{tab:boosted_any_g_from_star_gaps}
lead us to state the following corollary about size bounds for RQCs on arbitrary architectures,
\begin{corollary}
    \label{cor:size_bounddeg_knabe}
    RQCs on arbitrary architectures with maximum degree $\maxdeg$ and total edges $|E|$ form $\varepsilon$-approximate unitary $k$-designs in circuit size $\tau$, according to the following table
    \begin{center}
        \begin{tabular}{||c|c|c||}
            \hline
            $k$ & $\maxdeg$ & $\tau$ \\
            \hline
            $2$ & $38$ & $ 90 |E| (4 n + \log(1 / \varepsilon))$\\
            \hline
            $3$ & $12$ & $ 64 |E| (6 n + \log(1 / \varepsilon))$\\
            \hline
            $4$ & $5^*$ & $ 9 |E| (8 n + \log(1 / \varepsilon))$\\
            \hline
        \end{tabular}
    \end{center}
\end{corollary}
The above table gives size bounds for arbitrary $n$-vertex bounded degree graphs of maximum degree $\maxdeg$ total edges $|E|$ to form $\varepsilon$-approximate unitary $k$-designs. ${}^*$For $k = 4$, size bound is valid only for arbitrary architecture with no degree 2 vertex, as explained in the proof.
\begin{proof}[Proof of \autoref{cor:size_bounddeg_knabe}]
    In general, the upper bound $\tau$ on the circuit size of RQCs to form $\varepsilon$-approximate unitary $k$-designs on $n$-vertex arbitrary graphs $G$ with $|E|$ edges is given by,
    \begin{align}
        \tau = \frac{|E|}{\Delta(H(G, n, k))} (2 n k + \log(1/\varepsilon)),
    \end{align}
    where $\Delta(H(G, n, k))$ is the spectral gap of the corresponding Hamiltonian as defined in \autoref{def:ham_gnk} (See, for instance, \cite{BrandaoFernandoGSL2016LR}, or in the proof of \autoref{thm:size_any_graph}). We substitute the values of gaps from \autoref{tab:boosted_any_g_from_star_gaps} for $k = 2$ and $3$ and from \autoref{tab:any_g_from_star_gaps} for $k = 4$. For $k = 4$ and $n = 3$, spectral gap was found to be $1/2$ to numerical precision, which does not cross the finite-size-criterion set in \autoref{thm:knabe_any_connected_graph}. Therefore, the size bound for $k = 4$ applies to arbitrary architectures so long as they do not contain any degree 2 vertex.
\end{proof}

Analytically, one can compute the spectral gaps by constructing an operator basis to diagonalize the moment operator \cite{BrandaoFernandoGSL2010EQ,HaferkampJonas2021IS}. Here we give the first two star graph gaps. Note that the $n=3$ star graph Hamiltonian is equivalent to the 1D $n=3$ Hamiltonian with open boundary conditions.
\begin{proposition}[Theorem 4 and Eq.~35 of \cite{HaferkampJonas2021IS}]
\label{prop:exact_star_k2_n3}
For $k=2$, the $n=3$ star-graph Hamiltonian $H(G_\star,3,2) = h_{1,2}+h_{2,3}$, where $h_{i,j} = (\id - (P_H)_{i,j}\otimes \id_{[3]\backslash \{i,j\}})$ and $(P_H)_{i,j} = \int d\mu_{\rm Haar}\, U^{\otimes 2,2}_{i,j}$ has a spectral gap
\begin{equation}
    \Delta(H(G_\star,3,2)) = 1-\frac{q}{q^2+1}\,.
\end{equation}
For local qubits, the gap is $\Delta(H(G_\star,3,2)) = 3/5$.
\end{proposition}

\begin{proposition}
\label{prop:exact_star_k2_n4}
For $k=2$ and $n=4$, the star-graph Hamiltonian $H(G_\star,4,2) =h_{1,4}+h_{2,4}+h_{3,4}$, where $h_{i,j}=\id - (P_H)_{i,j}\otimes \id_{[4]\backslash \{i,j\}}$ and $(P_H)_{i,j} = \int d\mu_H\, U^{\otimes 2,2}_{i,j}$ has a spectral gap
\begin{equation}
    \Delta(H(G_\star,4,2)) = \frac{3}{2}-\frac{\sqrt{q^4+18q^2+1}}{2(q^2+1)}\,,
\end{equation}
which for $q=2$ is $\Delta(H(G_\star,4,2))=\frac{3}{2}-\frac{\sqrt{89}}{10}\approx 0.556602$ and for large $q$ approaches 1.
\end{proposition}
We opt not to give a full proof of this proposition, as it follows similarly to Theorem 4 of \cite{HaferkampJonas2021IS}, as well as the proof of \autoref{prop:exactcggap}, where we analytically construct a set of orthonormal basis operators for the moment operator using projectors onto the symmetric and antisymmetric subspaces. In the $n=3$ case, each term is rank 2, and thus the moment operator is rank 4. Exactly diagonalizing the operator then gives its eigenvalues, and specifically the gap, from which the Hamiltonian gap follows. For the $n=4$ star graph Hamiltonian, we again construct an operator basis for the moment operator, but its rank is now 18. We can write out the basis operators and symbolically diagonalize the moment operator in Mathematica, finding analytic expressions for its eigenvalues. The expression for the $n=4$ star graph gap is thus given above.

While the $n=3$ and $n=4$ star graph gaps turn out to be readily expressible as nice functions of $q$, there is no reason to expect that this persists for larger $n$. While the 1D second moment gaps can be exactly derived via a mapping to an integrable spin model, the resulting spin model for the star graph Hamiltonian contains an integrability breaking term, as discussed in \autorefapp{app:semiclass}. However, we perform a semiclassical approximation to determine the asymptotic expression for the star graph gaps for $k = 2$. We find an asymptotic gap of $1 - 1 / q^2$, which seems to be in agreement with numerics in so far as our finite size numerical calculations of spectral gaps are indicative of the asymptotic gaps.

From the analytic gaps in \autoref{prop:exact_star_k2_n3} and \autoref{prop:exact_star_k2_n4}, we have a {\it non-numerical} result that $n$-qubit random quantum circuits on any graph of maximum degree 3 forms an $\ep$-approximate unitary 2-design in depth $\tau = 9|E| (4n+\log(1/\ep))$.

\subsection{Knabe bounds for star graphs}

In this subsection, we derive the Knabe bound for frustration-free Hamiltonians on star graphs used above. Unfortunately, the bound cannot establish an asymptotic lower bound on star graph gaps as the threshold tends to one, and any frustration-free Hamiltonian on a star graph will likely have $\Delta(G_\star, n)\leq 1$. This is proved in \autoref{thm:max_spec_gap} for any translation-invariant frustration-free Hamiltonian on a star graph with some conditions on the 2-site projector. Nevertheless, a Knabe bound on star graphs still allows us to `boost' our numerically computed gaps to get gap lower bounds on larger systems, as in the previous subsection.

As before, the proof of the Knabe bound on star graphs does not depend on the moment $k$, we suppress the dependence for convenience as $H(G_\star,n)=H(G_\star,n,k)$.
\begin{theorem}
\label{thm:star_g_knabe}
Let $H(G_\star,n)=\sum_{i=1}^{n-1} h_{i,n}$ be a frustration-free Hamiltonian defined on a star graph, where the Hamiltonian terms are local projectors $h^2_{i,j}= h_{i,j}$, and let $\Delta(G_\star,n)=\Delta(H(G_\star,n))$ denote the spectral gap of the Hamiltonian. For $n\geq m\geq 3$, the star-graph Hamiltonian gaps obey
\begin{equation}
\Delta(G_\star,n) \geq \frac{n-2}{m-2} \left( \Delta(G_\star,m) - \frac{n-m}{n-2}\right)
\end{equation}
\end{theorem}

\begin{proof}
We again proceed by lower bounding the square of the Hamiltonian, as $H(G,n)^2\geq \gamma H(G,n)$ implies a spectral gap lower bound $\Delta(G,n)\geq \gamma$. Consider
\begin{equation}
    H(G,n)^2 = \sum_{i=1}^{n-1}h_{i,n} + \sum_{i\neq j}\{h_{i,n},h_{j,n}\} = H(G,n)+Q\,,
\end{equation}
and define $Q:=\sum_{i\neq j}\{h_{i,n},h_{j,n}\}$. Let $s(\{1,\ldots,n-1\}) := \binom{\{1,\ldots,n-1\}}{m-1}$ be the set of all size $(m-1)$ subsets of $\{1,\ldots,n-1\}$. We then consider the following subsystem operator, squaring the Hamiltonian on a subsystem of size $m-1$, a star graph on $m-1$ sites, and summing over all possible subsystems
\begin{equation}
\sum_{s\in s(\{1,\ldots,n-1\})} \bigg(\sum_{i=1}^{m-1} h_{s(i),n}\bigg)^2 = \binom{n-2}{m-2}H(G_\star,n) + \binom{n-3}{m-3}Q\,.
\end{equation}
Now note that the same quantity can be lower bounded using the spectral gap of the subsystem $H(G_\star,m)^2\geq \Delta(G_\star,m) H(G_\star,m)$ as
\begin{equation}
\sum_{s\in s(\{1,\ldots,n-1\})} \bigg(\sum_{i=1}^{m-1} h_{s(i),n}\bigg)^2 \geq \binom{n-2}{m-2}\Delta(G_\star,m) H(G_\star,n)\,,
\end{equation}
which implies the desired result.
\end{proof}

\section{Detectability lemma approach for lower bounds on spectral gaps of arbitrary connected graphs \label{sec:det_lem_app}}
    Our goal is to derive circuit size bounds for RQCs with arbitrary architectures to form unitary designs. And, the approach we are interested in uses spectral gaps of Hamiltonians as defined in \autoref{def:ham_gnk}. To that effort, in \autoref{sec:knabe_any_g}, we lower bounded spectral gaps $\Delta(G, n, k)$ for arbitrary graphs $G$ by spectral gaps for star graphs using the Knabe method (Results 1 and 2 from \autoref{sec:mot_and_res}). In this section, we provide a new method to find spectral gap lower bounds for arbitrary graphs by relating those to spectral gap for $\1D$ graphs (Results 3 and 4 from \autoref{sec:mot_and_res}) using the Detectability Lemma and Quantum Union Bound (\autoref{def:dl} and \autoref{def:qub}). We begin with stating and discussing the formal results about the gap bounds and the consequent size bounds for local RQCs. The following theorem provides a lower bound on the spectral gap of Hamiltonians as defined in \autoref{def:ham_gnk} on arbitrary connected graphs with prime power local dimension $q$,
    
    \begin{theorem}
        \label{thm:dl_any_graph}
         For all connected graphs $G$, the spectral gap $\Delta(G, n, k)$ is lower bounded by\linebreak $\Omega(1 / ( n^{4 \log_2 (n + 1) + 2} k^{4+o(1)}))$.
    \end{theorem}
    Using the lower bound on the spectral gap given by \autoref{thm:dl_any_graph} to derive an upper bound on the size of local RQCs in the same manner as is done in Refs.~\cite{ZnidaricMarko2008EC,BrownWinton2010CR,BrandaoFernandoGSL2010EQ,BrandaoFernandoGSL2016LR}, we arrive at the following theorem,
    \begin{theorem}
        \label{thm:size_any_graph}
        Local random quantum circuits on all $n$-vertex connected graphs $G$ with $|E|$ edges form $\varepsilon$-approximate $k$-designs in circuit size $O(|E| n^{3 + 4 \log_2 (n)} k^{5+o(1)})$.
    \end{theorem}
    By replacing the upper bound on the maximum degree of $n$-vertex graphs that appear in the proof of \autoref{thm:dl_any_graph} from the trivial bound of $n-1$ to a constant bound for certain graphs, we can improve the lower bound of \autoref{thm:dl_any_graph} from 1 over quasi-polynomial in $n$ to 1 over polynomial in $n$. This is the content of \autoref{cor:dl_cst_rreg} and \autoref{cor:dl_st_rreg_log_dep} that are given next,  
    \begin{corollary}
        \label{cor:dl_cst_rreg}
        For all connected graphs $G$ with corresponding ``compressed'' spanning tree $CST$ of constant maximum degree $\maxdeg$, the spectral gap $\Delta(G, n, k)$ is lower bounded by $\Omega(1 / ( n^{4 (1 + \log_2 (\maxdeg + 1))} k^{4+o(1)}))$.
    \end{corollary}
    
    \ni ``Compressed'' spanning tree is defined informally in Step 3 in \autoref{subsec:sketch_lem}, and formally in \autorefapp{app:dep_and_cst}.

    \begin{corollary}
        \label{cor:dl_st_rreg_log_dep}
        For all $n$-vertex connected graphs $G$ with spanning trees of bounded degree $\maxdeg$ and $\log(n)$ height, the spectral gap $\Delta(G, n, k)$ is lower bounded by $\Omega(1 / ( n^{2\log{(2)} (1 + \log_2 (\maxdeg + 1))} k^{4+o(1)}))$.
    \end{corollary}
    
    Furthermore, if we consider graphs with bounded degree and constant height, then we can improve the lower bound of \autoref{thm:dl_any_graph} by removing the $n$ dependence altogether.

    \begin{corollary}
        \label{cor:dl_st_rreg_con_dep}
        For all $n$-vertex connected graphs $G$ with spanning trees of bounded degree $\maxdeg$ and height, the spectral gap $\Delta(G, n, k)$ is lower bounded by $\Omega(1/k^{4+o(1)}))$.
    \end{corollary}

    Using the same methods as earlier, we can derive upper bounds on circuit sizes of local RQCs on graphs identified in \autoref{cor:dl_cst_rreg}, \autoref{cor:dl_st_rreg_log_dep} and \autoref{cor:dl_st_rreg_con_dep},
    
    \begin{theorem}
        \label{thm:size_bounddeg_r_dep_logn}
        Local random quantum circuits on all $n$-vertex connected graphs $G$ with $|E|$ edges and with \begin{enumerate}
            \item constant maximum degree $\maxdeg$ of their ``compressed'' spanning tree,
            \item spanning tree with constant maximum degree $\maxdeg$ and $\log(n)$ height,
            \item spanning tree with constant maximum degree $\maxdeg$ and constant height,
        \end{enumerate} form $\varepsilon$-approximate $k$-designs in circuit size, denoted by $\tau$, such that 
        \begin{enumerate}
            \item $\tau = O(|E| n^{5 + 4\log_2 (\maxdeg + 1)} k^{5+o(1)})$,
            \item $\tau = O(|E| n^{2\log{(2)} + 1 + 2\log{(2)}\log_2 (\maxdeg + 1)} k^{5+o(1)})$,
            \item $\tau = O(|E|n k^{5+o(1)})$,
        \end{enumerate}
        respectively.
    \end{theorem}

    The $k$ dependence of size bounds in \autoref{thm:size_any_graph} and \autoref{thm:size_bounddeg_r_dep_logn} is inherited from the $k$ dependence of the corresponding results for local random quantum circuits on $\1D$ graphs, where we have used the improved $k$ dependence of $\approx k^{5+o(1)}$ for prime power local dimension following Ref.~\cite{HaferkampJonas2022RQ}. For general (non prime power) local dimensions, we can instead use the 1D gaps in Ref.~\cite{BrandaoFernandoGSL2016LR}, which simply alters the $k$ dependence to $\approx k^{7.5}$. We provide analogous theorems to those above which hold for any local dimension $q$, where the proof is identical and is thus omitted.
    
    \begin{theorem}
        \label{thm:size_any_graph_anyq}
        Local random quantum circuits on all $n$-vertex connected graphs $G$ with $|E|$ edges form $\varepsilon$-approximate $k$-designs in circuit size $O(|E| n^{3 + 4 \log_2 (n)} k^{7.5})$ for any local dimension $q$. Similarly, the following hold for all local dimensions $q$.
        If the graph $G$ has maximum degree $\maxdeg$ of its ``compressed'' spanning tree then the circuit size at which they form approximate $k$-designs is $O(|E| n^{5 + 4\log_2 (\maxdeg + 1)} k^{7.5})$. 
        If the graph $G$ has spanning tree with maximum degree $\maxdeg$ and $\log(n)$ height, then the circuit size at which they form designs is $O(|E| n^{2\log{(2)} + 1 + 2\log{(2)}\log_2 (\maxdeg + 1)} k^{7.5})$.
        If the graph $G$ has spanning tree with maximum degree $\maxdeg$ and height $\maxh$, then the circuit size at which they form designs is $O(|E| n k^{7.5})$.
    \end{theorem}

    The main ingredient of spectral gap results \autoref{thm:dl_any_graph}, \autoref{cor:dl_cst_rreg} and \autoref{cor:dl_st_rreg_log_dep} and, thus, of the size bounds of \autoref{thm:size_any_graph} and \autoref{thm:size_bounddeg_r_dep_logn}, is the following lemma, which relates the spectral gap of Hamiltonians defined in \autoref{def:ham_gnk} on arbitrary graphs to spectral gap of the $\1D$ Hamiltonian well studied in Refs.~\cite{ZnidaricMarko2008EC,BrandaoFernandoGSL2010EQ,BrownWinton2010CR,BrandaoFernandoGSL2016LR,HunterJonesNicholas2019UD,HaferkampJonas2021IS,BensaJas2021FL,HaferkampJonas2022RQ}
    \begin{lemma}
        \label{lem:any_g_1d}
        Consider the spectral gaps $\Delta(G, n, k)$ and $\Delta(\1D, n, k)$ of the Hamiltonians $H(G, n, k)$ and $H(\1D, n, k)$ defined as in \autoref{def:ham_gnk} on an arbitrary $n$-vertex connected graph $G(V, E)$ and $\1D$ graph with open boundary conditions, respectively. The following lower bound holds,
        \begin{equation}
            \Delta(G, n, k) \geq \cfrac{35}{768} \cfrac{\Delta(\1D, n, k)}{(4(g+1)^2)^{d}}\,,
        \end{equation}
        where $g$ is the maximum degree of the ``compressed,'' spanning tree, denoted by $CST^{(0)}$, that is derived from the spanning tree $ST$ of $G$ (refer to Step~3 below) and $d$ denotes the ``depth'' of $ST$ (refer to Step~2 below).
    \end{lemma}
    \begin{remark}
        \begin{itemize}
            \item[]
            \item The words ``depth'' and ``compressed'' are quoted because they are not standard terms and their informal definitions are provided in Step 2 and 3 in \autoref{subsec:sketch_lem}, while their formal definitions are presented in \autorefapp{app:dep} and \autorefapp{app:dep_and_cst}, respectively. 
            \item The maximum degree $g$ of the compressed spanning tree, $CST^{(0)}$, may in general be greater than the maximum degree of the corresponding spanning tree $ST$.
            \item The constants in the general gap lower bound come from the Quantum Union Bound.
        \end{itemize}
    \end{remark}
    Other than comprising of \autoref{lem:any_g_1d}, the proof of \autoref{thm:dl_any_graph}, \autoref{cor:dl_cst_rreg} and \autoref{cor:dl_st_rreg_log_dep} require bounds on the graph properties $g$ and $d$ in \autoref{lem:any_g_1d} in terms of the number of vertices $n$ of the graph $G$. Those graph properties are stated and proved in \autorefapp{app:tree_graph_prop}. Meanwhile, the proofs of Theorems and Corollaries stated in this section are provided in \autoref{subsec:short_proofs}. 

    \subsection{Sketch of the proof of \autoref{lem:any_g_1d} \label{subsec:sketch_lem}}
    We sketch out the proof of \autoref{lem:any_g_1d} in six broad steps given below;\\

\ni \textbf{Step~1:}\ We recall an observation from Refs.~\cite{BrandaoFernandoGSL2010EQ,BrownWinton2012SS,Onorati17,oszmaniec2022saturation} in \autoref{lem:low_bnd_spt}. It not only provides us the starting point for our proof of lower bounds on spectral gaps for large class of graphs, but also points out how that class of graphs is more general than those considered in literature.
    \begin{lemma}
        \label{lem:low_bnd_spt}
        The spectral gap $\Delta(G, n, k)$ of the Hamiltonian $H(G, n, k)$ defined as in \autoref{def:ham_gnk} on a connected graph $G(V, E)$ is lower bounded by the spectral gap $\Delta(ST, n, k)$ of the restriction of $H(G, n, k)$ to the spanning tree $ST$ of $G$.
    \end{lemma} 
    \begin{proof}
        This is a direct consequence of the proof for the explicit expression for the ground state vectors given in Lemma 17 in Ref.~\cite{BrandaoFernandoGSL2016LR}, which relies on Schur-Weyl duality (see Ref.~\cite{RobertsDanielA2017CA}) and the fact that two qudit unitaries on connected graphs generate the set of all unitaries on $n$ qudits \cite{BarencoAdriano1995EG}. The argument is three fold,
        \begin{enumerate}
            \item The ground state space (and thus the excited state space) for the considered $n$-qudit Hamiltonian is the same regardless of the architecture as long as its connected.
            \item The Hamiltonian that we consider is composed of positive semi-definite local terms, therefore, Hamiltoinan on a graph $G$, $H(G, n, k)$, has more local terms than the Hamiltonian on its spanning tree, $H(ST, n, k)$. 
            \item Since the spectral gaps of the Hamiltonians are their minimization in their excited state space, which are identical due to point 1, the minimization of the $H(G, n, k)$ is greater than $H(ST, n, k)$ due to point 2, from which the claim follows.
        \end{enumerate} 
    \end{proof}
    
    Circuit size bounds on RQCs forming approximate unitary designs  (local RQCs or otherwise with specific procedure of applying gates) are known for $\1D$, complete and $D$-dimensional graphs. Complete and hypercubic lattices have Hamiltonian paths, that is, they admit a $\1D$ graph as their spanning tree. Therefore, by Step~1, the spectral gap of the Hamiltonian on a complete graph or $D$-dimensional lattice is trivially lower bounded by the spectral gap of the Hamiltonian on $\1D$ graphs with identical number of vertices. Consequently the size bounds are easily determined to be $O(\mathrm{poly}(n))$. We will not be restricting ourselves to graphs with $\1D$ spanning trees (Hamiltonian paths), instead our approach to find spectral gaps and, thus, size bounds will work for graphs with arbitrary spanning trees. This argument illustrates the generality of the graphs that we consider compared to those studied in literature.
    In the following, remainder of the steps for the proof of \autoref{lem:any_g_1d} are described succinctly with details in \autorefapp{app:proof_lem_3}.\\
    
    \ni \textbf{Step~2:}\ We provide an algorithmic definition of the ``depth'' of a tree graph and differentiate the ``depth'' from the usual notion of ``height'' of a tree graph. Given a root vertex of a tree graph, the height of the graph is the maximum number of edges one needs to traverse to get to a leaf vertex in the graph. In contrast, the ``depth'' of a tree graph can be intuitively understood as the number of depth first searches (not all starting from the root vertex) to find every leaf in the tree graph. More precisely, ``depth'' is the defined as the minimum over all outcomes of the maximum provisional ``depth'' assigned to each vertex by \autoref{alg:pseudocode} (note that, \autoref{alg:pseudocode} as given is only heuristic and the detailed algorithm in presented in \autorefapp{app:dep}). The first input to the algorithm is the set $\mathcal{S}$ containing the description of the tree graph (or, said simply, containing the tree graph) and a counter variable $d = 0$.

    \begin{algorithm}[ht!]
        \caption{A recursive algorithm to define ``depth'' of a tree graph}\label{alg:pseudocode}
        \begin{algorithmic}[1]
            \item[$\triangleright$] Receive a set of graphs $\mathcal{S}$ and a counter variable $d$.
            \For{each graph in the input set of graphs}
                \State Assign root vertex depth $d$ ($d = 0$ when this line is executed for the first time).
                \State Assign all vertices connected to it depth $d + 1$.
                \State Depth first search a path from the root to a leaf. 
                \State Mark all the vertices along the path by depth $d + 1$.
                \State Increment $d = d + 2$.
                \State Do \autoref{alg:pseudocode} on the set of graphs that is obtained from the current graph by subtracting the vertices (and all edges connected to those vertices)  marked in the current recursion.
            \EndFor
        \end{algorithmic}
    \end{algorithm}

    \ni Since the algorithm depends on depth first searches, for different searches the algorithm might find different final (provisional) depths. Thus, to have a well defined notion of ``depth'' we take it to be the minimum over all possible maximum depth assignments made by the algorithm for a given graph and root vertex. In \autorefapp{app:dep}, we further motivate the definition of ``depth'' by recounting the example of a ``Y'' shaped graph (refer to \autoref{fig:y_graph}) where the more familiar notion of ``height'' of a tree graph is insufficient for our purposes. We demonstrate \autoref{alg:pseudocode} on that graph in \autoref{fig:demo_pseudoalg}. Note that the height of the ``Y'' shaped graph given in  \autoref{fig:y_graph} is $5$ but its ``depth'' with respect to the central vertex is only $3$ (refer to \autoref{fig:demo_pseudoalg}(e)).\\
    \begin{figure}[ht!]
        \centering
        \includegraphics[width=0.53\linewidth]{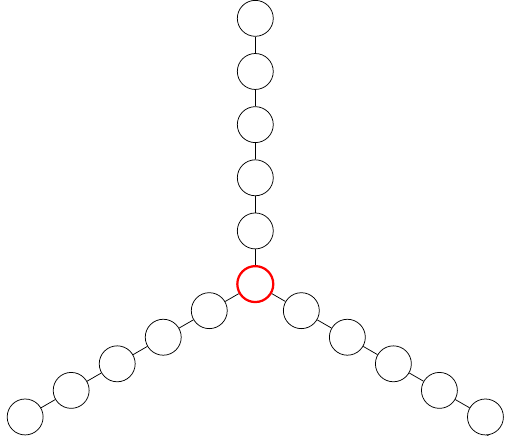}
        \caption{A ``Y'' shaped connected graph which is its own spanning tree. The red vertex denotes our choice of the root vertex for the visual demonstration of \autoref{alg:pseudocode} in \autoref{fig:demo_pseudoalg}. Since, we define ``depth'' of a spanning tree with respect to a root vertex, it could change for different choices of root vertex. However, all of our spectral gap and size bound result remain unaffected because they functionally depend on the ``depth.''}
        \label{fig:y_graph}
    \end{figure}

    \begin{figure}[ht!]
        \centering
        \includegraphics[width=0.97\linewidth]{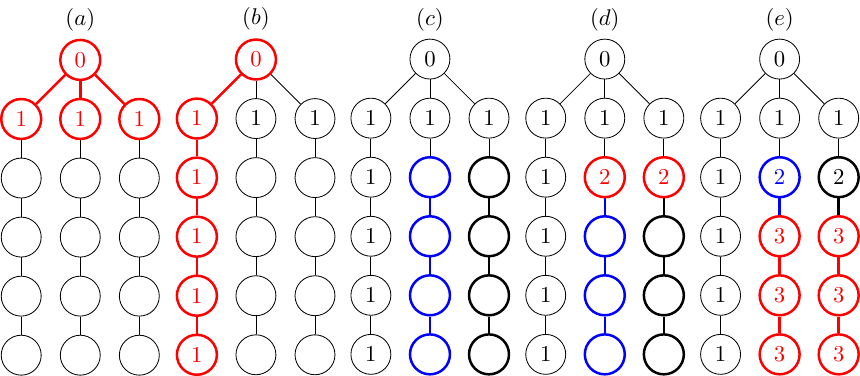}
        \caption{Step by step example of \autoref{alg:pseudocode}, for a ``Y'' shaped graph (refer to \autoref{fig:y_graph}), whose central vertex is identified as the root vertex. Part $(a)$ and $(b)$ depict the action of the lines $2$ and $3$ and lines $4$ and $5$ of \autoref{alg:pseudocode} in red. In Parts $(c), (d)$ and $(e)$, the blue and black sub-graphs in a set of graphs that is passed to the next recursion in line $7$ of \autoref{alg:pseudocode}. Part $(d)$ and $(e)$ depict the action of the lines $2$ and $3$ and lines $4$ and $5$ of \autoref{alg:pseudocode} on each sub-graph in red. Note that for a different choice of the root vertex the \autoref{alg:pseudocode} could evaluate to a different depth. However, we will only consider the case where a root vertex is fixed and depth is calculated with respect to that root vertex.}
        \label{fig:demo_pseudoalg}
    \end{figure}

    \ni \textbf{Step~3:}\ We define a new $n$-vertex tree graph that we derive from the spanning tree $ST$ of $G$. Roughly speaking, the new graph is obtained from the $ST$ by taking all line segments with the same depth assignments from Step~2 and compressing them to star graphs about the first vertices in the line segments. In \autoref{fig:cst_y_graph}, we draw such a graph for the example of the ``Y'' shaped graph given in \autoref{fig:y_graph}. Compare the structure of \autoref{fig:cst_y_graph} with ``depth'' assignments in \autoref{fig:demo_pseudoalg}(e) to see the above construction in play. We define the Detectability Lemma norms for the new graph and $ST$ in terms of the complements of the local terms of the Hamiltonian defined on those graphs (refer to \autoref{def:dl} and the text below it). We conclude this step by equating the two Detectability Lemma norms. We refer to this step as ``compression,'' correspondingly, we refer to the new tree graph defined in this step as the ``compressed'' spanning tree and denote it by $CST^{(0)}$. The superscript is clarified in the next step.\\

    \begin{figure}[ht!]
        \centering
        \includegraphics[width=0.56\linewidth]{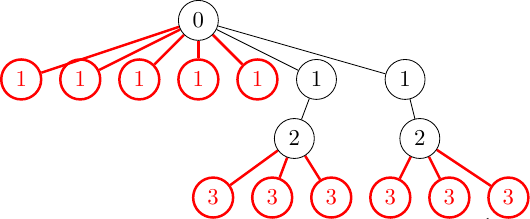}
        \caption{The new $n$-vertex graph for the original ``Y'' shaped graph of \autoref{fig:y_graph}. This graph is referred to as the ``compressed'' spanning tree of the original graph, corresponding to the ``depth'' assignments in \autoref{fig:demo_pseudoalg}, which in turn, rely on the choice of the root vertex.}
        \label{fig:cst_y_graph}
    \end{figure}

    \newpage
    \ni \textbf{Step~4:}\ We show that the Detectability Lemma norm for $CST^{(0)}$ is equal to the Detectability Lemma norm for a graph that results from $CST^{(0)}$ by ``flattening out'' a part of that graph. The ``flattened out'' graph is referred to as $CST^{(1)}$. In \autoref{fig:cst1_y_graph}, we depict $CST^{(1)}$ for our running example of the ``Y'' shaped graph. \\

    \begin{figure}[ht!]
        \centering
        \includegraphics[width=0.9\linewidth]{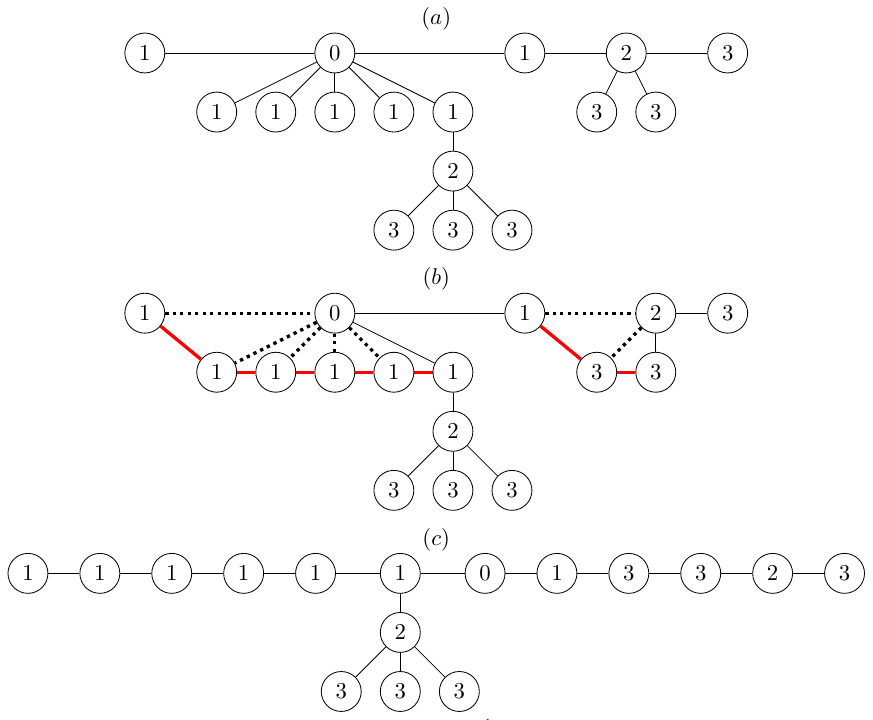}
        \caption{In part $(a)$, we draw the same graph as in \autoref{fig:cst_y_graph} albeit slightly differently to visually identify a $\1D$ piece in that graph (which is drawn as a straight line on the stop of the graph). In part $(b)$, we depict how we ``flatten'' out the graph. We denote the edges that are broken by the dotted lines and the new edges that are added in their place by the red lines. Part $(c)$ of the figure illustrates the final result of one step of ``flattening'' out procedure. The graph in part $(c)$ is referred to as $CST^{(1)}$ for the ``Y'' shaped graph.}
        \label{fig:cst1_y_graph}
    \end{figure}
    
    \ni \textbf{Step~5:}\ We repeat the ``flattening out'' procedure for $d - 1$ times, where $d$ is the depth (w.r.t. definition from Step~2) of $ST$. On the $d^{\text{th}}$ iteration, we show that the Detectability Lemma norm for $CST^{(d - 1)}$ is equal to the Detectability Lemma norm for the most ``flattened out'' tree graph, that is, the $n$-vertex $\1D$ line graph with open boundary conditions. In \autoref{fig:cstd_y_graph}, we illustrate the repetitions of Step~4 for our example of $CST^{(1)}$ (refer to \autoref{fig:cst1_y_graph}) of the ``Y'' shaped graph.\\

    \begin{figure}
        \centering
        \includegraphics[width=0.98\linewidth]{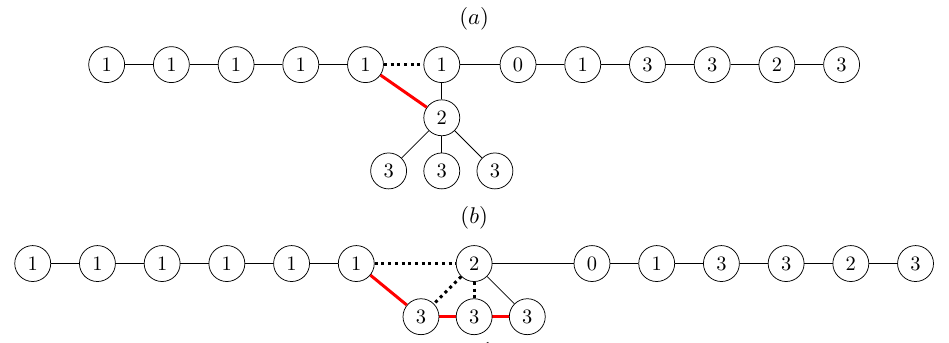}
        \caption{After repeating the ``flattening'' out procedure two more times on the ``compressed'' spanning tree $CST^{(1)}$ of the our running example of the ``Y'' shaped graph, we arrive at the $\1D$ graph. In part $(a)$ and $(b)$, we depict the second and the third repetition of the ``flattening'' out procedures, respectively, using the same notation as in \autoref{fig:cst1_y_graph}. On the third iteration the ``flattened'' graph is the $\1D$ graph.}
        \label{fig:cstd_y_graph}
    \end{figure}
    
    \ni \textbf{Step~6:}\ We put the results of the previous steps together with the help of Quantum Union Bound and, thus, upper bound the Detectability Lemma norm for $ST$ in terms of Detectability Lemma norm for $\1D$ line graph. Afterwards, we again use the Quantum Union Bound to infer a lower bound on $\Delta(H(ST, n, k))$ in terms of $\Delta(H(\1D, n, k))$. Combining the result of this step with that of Step~1 we conclude with a lower bound on $\Delta(H(G, n, k))$ in terms of $\Delta(H(\1D, n, k))$.\\
    
    Our main insight to prove \autoref{lem:any_g_1d} is a relationship between two Detectability Lemma operators; one composed of projectors corresponding to a $\1D$ graph, and the other composed of projectors corresponding to a star graph. The above relationship comes at the cost of left/right multiplying one of the two Detectability Lemma operators by a unitary transformation that preserves the orthogonality between the ground state and excited state spaces of the Hamiltonian that is the sum of the complements of the projectors used in the Detectability Lemma operator. We dub transforming a Detectability Lemma on $\1D$ graph to star graph as ``compression'' and the other way around as ``flattening out.'' This insight is put together into two lemmas, \autoref{lem:1d_to_star} and \autoref{lem:star_to_1d}. Using the ``compression'' and ``flattening out'' procedures recursively, we equate Detectability Lemma norms between a graph at a present step to that on the successive step till the graph at the successive step is the $\1D$ line graph. Finally, we put together the equalities from each iteration using the Quantum Union Bound and the Dectectability Lemma. 
    In this step, we incur an exponentially small in the ``depth,'' $d$, lower bound, which eventually causes the size bound in \autoref{thm:size_any_graph} to be super-polynomial in $n$. 
    However, we believe that such a worsening of the gap and size bounds are artifacts of the proof technique. 
    It might be possible to improve the lower bound in \autoref{lem:any_g_1d} by improving the Detectability Lemma upper bound. As given in \autoref{def:dl}, the choice of projectors is arbitrary and the upper bound holds for all excited states of the Hamiltonian constructed from those projectors. Perhaps the Detectability Lemma upper bound can be improved for our choice of projectors and for first excited state. 
    This hope for improvement does not contradict the insight in Ref.~\cite{AnshuAnurag2016SP} that there exists an example of local ground state projectors such that the usual Detectability Lemma upper bound is tight. 
    The existence of such an example does not preclude the existence of a better Detectability Lemma upper bound for our particular choice of projectors. We leave this direction of thought for future work. 
    In \autorefapp{app:proof_lem_3}, we provide the details of the proof of \autoref{lem:any_g_1d} in the six steps mentioned above.

    \subsection{Proofs of \autoref{thm:dl_any_graph}, \autoref{cor:dl_cst_rreg}, \autoref{cor:dl_st_rreg_log_dep} and \autoref{thm:size_bounddeg_r_dep_logn} \label{subsec:short_proofs}}
    
    \begin{proof}[Proof of \autoref{thm:dl_any_graph}]
        In the result of \autoref{lem:any_g_1d}, substitute,
        \begin{enumerate}
            \item the upper bound for the maximum degree of an $n$-vertex spanning tree, $g \leq n - 1$, and
            \item the upper bound on the depth $d$ of an $n$ vertex spanning tree from \autoref{cor:dep_upper_bound}, $d \leq (2 / \log(2)) \log(n + 1) - 1$,
            \item the spectral gap $\Delta(\1D, n, k)$ from \cite{BrandaoFernandoGSL2016LR} for $n \geq \lceil 2.5 \log_q (4k) \rceil$,
            \begin{equation}
                \Delta(\1D, n, k) \geq \frac{1}{4 \lceil 2.5 \log_q (4k) \rceil} \Delta(\1D, \lceil 2.5 \log_q (4k) \rceil, k)\,,
            \end{equation}
            where
            \begin{equation}
                \Delta(\1D, \lceil 2.5 \log_q (4k) \rceil, k) \geq \frac{1}{\lceil 2.5 \log_q (4k) \rceil ((q^2+1) e)^{\lceil 2.5 \log_q (4k) \rceil}}\,.
                \label{eq:1DBHHgap}
            \end{equation}
            Note that the lower bound for $\Delta(\1D, n, k)$ is independent of $n$ and greater than a constant times $1/k^{2.5 \log_q((q^2+1)e)}$, thus, $\Delta(\1D, n, k) = \Omega(1/k^{9.5})$ for qubits.
            \item For local qubits $q=2$, the result from Ref.~\cite{HaferkampJonas2022RQ} can be used instead of the one given in the previous point, where the unconditional gap lower bound is improved to
            \begin{equation}
                \Delta(\1D, \lceil 2.5 \log_2 (4k) \rceil, k) \geq \frac{1}{120000} \frac{1}{(\lceil 2.5 \log_2 (4k) \rceil)^4 2^{2 \lceil 2.5 \log_2 (4k) \rceil}}\,.
                \label{eq:1Ducgap}
            \end{equation}
            Combined with an improved Nachtergaele bound, the 1D gap lower bound is then $\Delta(\1D, n, k)\geq \big(C \log^5 (k) k^{4+3/\sqrt{\log_2(t)}}\big)^{-1} = \Omega(1/k^{4+o(1)})$. The proof technique in Ref.~\cite{HaferkampJonas2022RQ} works for any prime power local dimension $q$, up to a potentially different constant in Eq.~\eqref{eq:1Ducgap}.
        \end{enumerate}
        Using the gap lower bound from Ref.~\cite{HaferkampJonas2022RQ}, we find
        \begin{align}
            \Delta(H(G, n, k)) &\geq \cfrac{35}{768} \cfrac{\Omega(1/k^{4+o(1)})}{(2n)^{(4 / \log(2)) \log(n + 1) - 2}} \\
            &= \cfrac{35}{192} \cfrac{\Omega(1/k^{4+o(1)})}{\cfrac{(n + 1)^4}{n^2} n^{4 \log_2 (n + 1)}} \\
            \label{eq:any_g_gap}
            &\geq \cfrac{35}{384} \cfrac{\Omega(1/k^{4+o(1)})}{n^{4 \log_2 (n + 1) + 2}} \\
            &= \Omega(1 / ( n^{4 \log_2 (n + 1) + 2} k^{4+o(1)}))\,,
        \end{align}
        where in the second inequality we used the fact that $(n + 1)^2 \leq 2n^2$ for $n \geq 3$.
    \end{proof}

    \begin{proof}[Proof of \autoref{thm:size_any_graph}]
        Consider an $n$-vertex connected graph $G(V, E)$ with $V$ as the set of vertices and $E$ as the set of unordered pair of vertices that share an edge in $G$. Let $|E|$ denote the number of edges in $G$. Let $t$ denote the size of the local random quantum circuit after which the difference between $k^{\text{th}}$ moment of the unitary group on $n$-qudits computed using the random quantum circuit ensemble and the Haar measure are $\varepsilon$-close in diamond norm. Then,
        following Ref.~\cite{BrandaoFernandoGSL2016LR}, we find
        \begin{align}
            &\left(1 - \frac{\Delta(G, n, k)}{|E|}\right)^t \leq \frac{\varepsilon}{q^{2nk}} \\
            \label{eq:basic_size_bound}
            \implies &t \geq \frac{|E|}{\Delta(G, n, k)} \left(2nk \log(q) + \log (1 / \varepsilon)\right)\,,
        \end{align}
        where  we substitute the lower bound on $\Delta(G, n, k)$ from Eq.~\eqref{eq:any_g_gap} to find the upper bound on the minimal circuit size $\tau$ at which we form an $\ep$-approximate $k$-design
        \begin{align}
            &\tau \leq \frac{384}{35} |E| n^{4 \log_2 {(n + 1)} + 2} O(k^{4+o(1)}) \left(2nk \log(q) + \log (1 / \varepsilon)\right) \\
            &\tau = O\left(|E| n^{3 + 4\log_2 (n + 1)} k^{5+o(1)}\right)\,.
        \end{align}
    \end{proof}
    
    \begin{proof}[Proof of \autoref{cor:dl_cst_rreg}]
        The proof is identical to that for \autoref{thm:dl_any_graph} except that instead of the upper bound in the first point of the proof, we use $g \leq \maxdeg$,
        \begin{align}
            \Delta(H(G, n, k)) &\geq \cfrac{35}{768} \cfrac{\Omega(1/k^{4+o(1)})}{(2(\maxdeg + 1))^{(4 / \log(2)) \log(n + 1) - 2}} \\
            &= \cfrac{35}{192} \cfrac{(\maxdeg + 1)^2 \Omega(1/k^{4+o(1)})}{(n + 1)^4 (\maxdeg + 1)^{4 \log_2 (n + 1)}} \\
            \label{eq:sub_in_des_dep_1}
            &= \cfrac{35}{192} \cfrac{(\maxdeg + 1)^2 \Omega(1/k^{4+o(1)})}{(n + 1)^{4(1 + \log_2 (\maxdeg + 1))}} \\
            &= \Omega(1 / ( n^{4 (1 + \log_2 (\maxdeg + 1))} k^{4+o(1)}))\,,
        \end{align}
        where in the last line we use the fact that for $n \geq 3$, there exists a constant $c$ such that $c n^a \geq (n + 1)^a$ for any $a > 0$.
    \end{proof}

    \begin{proof}[Proof of \autoref{cor:dl_st_rreg_log_dep}]
        We introduced the concept of ``compressed'' spanning tree in order to relate the spectral gap for any spanning tree to that for a tree of $\log(n)$ height (where ``height'' carries its usual meaning from graph theory). Therefore, if the spanning tree of a graph is already of $\log(n)$ height, then Step~3 in the derivation of \autoref{lem:any_g_1d} can be skipped, and the rest of the proof can be worked out beginning with the original spanning tree in place of the ``compressed'' spanning tree, $CST^{(0)}$. The bound stated in this corollary looks slightly different from the one in \autoref{cor:dl_cst_rreg}, because we specify an upper bound on the depth (in this case equal to the height) $d \leq \log(n)$. And, we do not need the extra step of recursion mentioned in Step~6, below Eq.~\eqref{eq:one_less_recur}. The extra step relates the spectral gap of the spanning tree to that of the ``compressed'' spanning tree $CST^{(0)}$, which are one and the same graph in the present consideration. Thus, that step is redundant and can be omitted. Following the proof for \autoref{thm:dl_any_graph}, and substituting $g \leq \maxdeg$ and $d \mapsto d - 1 \leq \log(n) - 1$, we find,
        \begin{align}
            \Delta(H(G, n, k)) &\geq \cfrac{35}{768} \cfrac{\Omega(1/k^{4+o(1)})}{(2(\maxdeg + 1))^{2\log{(n)} - 2}} \\
            &= \cfrac{35}{192} \cfrac{(\maxdeg + 1)^2 \Omega(1/k^{4+o(1)})}{n^{2\log{(2(\maxdeg + 1))}}} \\
            \label{eq:sub_in_des_dep_2}
            &= \cfrac{35}{192} \cfrac{(\maxdeg + 1)^2 \Omega(1/k^{4+o(1)})}{n^{2\log{(2)}(1 + \log_2{(\maxdeg + 1))}}} \\
            &= \Omega(1 / ( n^{2 \log {(2)} (1 + \log_2 (\maxdeg + 1))} k^{4+o(1)}))\,.
        \end{align}
    \end{proof}
    
    \begin{proof}[Proof of \autoref{cor:dl_st_rreg_con_dep}]
        The proof is identical to the proof of \autoref{cor:dl_st_rreg_log_dep}, but with the upper bound on $d$ given by, $d - 1 \leq \maxh - 1$, where $\maxh$ denotes the constant height, which effectively implies replacing $\maxh$ in place of $\log{(n)}$ in the proof of \autoref{cor:dl_st_rreg_log_dep},
        \begin{align}
            \label{eq:sub_in_des_dep_3}
            \Delta(H(G, n, k)) &\geq \cfrac{35}{768} \cfrac{\Omega(1/k^{4+o(1)})}{(2(\maxdeg + 1))^{2\maxh - 2}} \\
            &= \Omega(1 / k^{4+o(1)}))\,.
        \end{align}
    \end{proof}
    
    \begin{proof}[Proof of \autoref{thm:size_bounddeg_r_dep_logn}]
        We prove the statement for case (i), and identify what changes in the proof for (ii) and (iii). Consider an $n$-vertex connected graph $G(V, E)$ with $\maxdeg$ as the constant maximum degree of its ``compressed'' spanning tree, $V$ as the set of vertices and $E$ as the set of unordered pair of vertices that share an edge in $G$. Let $t$ denote the depth of the local random quantum circuit after which the difference between $k^{\text{th}}$ moment of the unitary group on $n$-qudits computed using the random quantum circuit ensemble and the Haar measure are $\varepsilon$-close in diamond norm. Proceeding, we find
        \begin{align}
            &\left(1 - \frac{\Delta(H(G, n, k))}{|E|}\right)^t \leq \frac{\varepsilon}{q^{2nk}} \\
            \implies &t \geq \frac{|E|}{\Delta(H(G, n, k))} \left(2nk \log(q) + \log (1 / \varepsilon)\right)\,,
        \end{align}
        where we substitute the lower bound on $\Delta(H(G, n, k))$ from Eq.~\eqref{eq:sub_in_des_dep_1} to arrive at the upper bound on the minimal design size $\tau,$
        \begin{align}
            &\tau \leq \frac{192}{35} \frac{|E|}{(\maxdeg + 1)^2} (n + 1)^{4(1 + \log_2 (\maxdeg + 1))} O(k^{4+o(1)}) \left(2nk \log(q) + \log (1 / \varepsilon)\right) \\
            &\tau = O\left(|E| n^{5 + 4\log_2 (\maxdeg + 1))} k^{5+o(1)}\right)\,.
        \end{align}
        Next, consider case (ii) where $G$ is such that its spanning tree is constant maximum degree $\maxdeg$ and $\log(n)$ height. Then, instead of substituting Eq.~\eqref{eq:sub_in_des_dep_1}, we substitute Eq.~\eqref{eq:sub_in_des_dep_2} (which only differ in the exponent of $n$) to find
        \begin{align}
            &\tau \leq \frac{192}{35} \frac{|E|}{(\maxdeg + 1)^2} (n + 1)^{2\log{(2)}(1 + \log_2 (\maxdeg + 1))} O(k^{4+o(1)}) \left(2nk \log(q) + \log (1 / \varepsilon)\right) \\
            &\tau = O\left(|E| n^{2\log{(2)} + 1 + 2\log{(2)}\log_2 (\maxdeg + 1))} k^{5+o(1)}\right)\,.
        \end{align}
        Finally, consider case (iii) where $G$ is such that its spanning tree is of constant maximum degree $\maxdeg$ and height $\maxh$. Then, instead of substituting Eq.~\eqref{eq:sub_in_des_dep_1}, we substitute Eq.~\eqref{eq:sub_in_des_dep_3} and find
        \begin{align}
            &\tau \leq \frac{768}{35} (2(\maxdeg + 1))^{2 h - 2} |E| O(k^{4+o(1)}) \left(2nk \log(q) + \log (1 / \varepsilon)\right) \\
            &\tau = O\left(|E| n k^{5+o(1)}\right)\,.
        \end{align}
    \end{proof}

    The $k$ dependence of all the results presented here are inherited from the $k$ dependence of the corresponding results for local random quantum circuits on $\1D$ line graphs, which is $k^{5+o(1)}$ following Ref.~\cite{HaferkampJonas2022RQ} for prime power local dimensions. For general (non prime power) local dimensions $q$, the $k$ dependence can be taken to be $\approx k^{10.5}$. But, noting that the exponent in Eq.~\eqref{eq:1DBHHgap} improves if we increase $q$, the first non prime power local dimension is $q=6$. Thus, we may take the $k$-dependence to be $\approx k^{7.5}$ for general local dimensions, as presented in \autoref{thm:size_any_graph_anyq}.

\section{Upper bound on arbitrary graphs gaps \label{sec:upb_gaps}}
In this section, we prove the least upper bound on the minimal circuit size to form designs that one could derive using the spectral gap method. This is the last result in \autoref{sec:mot_and_res} and the content of \autoref{cor:opt_size}. To reach this result, we prove spectral gap upper bounds for frustration-free Hamiltonian $H(G)$ on arbitrary $n$-vertex graphs $G(V, E)$ when local terms in $H(G)$ satisfy a certain property. Since that property holds for $H(G, n, k)$ as defined in \autoref{def:ham_gnk}, for arbitrary $n$ and $k$, it follows that the spectral gap upper bound provided here applies for all values of $k$ (unconditional on $q$ or $n$). Using that upper bound on the spectral gap of $H(G, n, k)$, we prove \autoref{cor:opt_size}. In the following, we will use $|V|$ to denote the number of vertices in the graph $G(V, E)$. When the graph is an $n$-vertex graph, we will use $n$ instead.
    
    Define a two site projector $P^{(2)} _{(v, w)} \in \mathbb{C}^{d^2 \times d^2}$ acting non-trivially on the sites corresponding to the vertices that share an edge in the graph, that is, $(v, w) \in E$. Similarly, let $P^{(2), \perp} _{(v, w)} \in \mathbb{C}^{d^2 \times d^2}$ denote the projectors on the orthogonal complement of $\mathrm{ran}\, P^{(2)} _{(v, w)}$. Consider a frustration-free $2$-local Hamiltonian $H(G)$ defined on the graph $G(V, E)$ with local terms given by $P^{(2), \perp} _{(v, w)}$, where $(v, w) \in E$,  
    \begin{align}
        \label{eq:ff_ham}
        H(G) := \sum_{(v,  w) \in E} P^{(2), \perp} _{(v, w)} \otimes \id_{V \backslash \{v, w\}}\,,
    \end{align}
    where $\id_{X}$ denotes the identity operator on Hilbert spaces corresponding to the vertices in the set $X$. Without loss of generality, we can consider the ground state energy of $H(G)$ to be $0$ because if it were not, then it could be made such by adding a term proportional to the identity to Eq.~\eqref{eq:ff_ham}. Having defined $H(G)$, $\{P^{(2)} _{(v, w)}\}_{(v, w) \in E}$ can be identified as the set of ground state projectors and $\{P^{(2), \perp} _{(v, w)}\}_{(v, w) \in E}$ as the set of local terms in $H(G)$. Consider a corresponding operator $M(G)$ defined as follows
    \begin{align}
        \label{eq:inv_ham}
        M(G) := \sum_{(v,  w) \in E} P^{(2)} _{(v, w)} \otimes \id_{V \backslash \{v, w\}}\,.
    \end{align}

    By their definitions, $H(G) := |E| \cdot \id - M(G)$, where $|E|$ denotes the number of edges in the graph $G(V, E)$. Make two observations: first, by frustration-freeness of $H(G)$, there exists a state $\ket{\psi} \in (\mathbb{C}^d)^{\otimes n}$ that minimizes the energy contribution to the value $0$ from each local term in $H(G);$ and second, such a state simultaneously maximizes the expectation value of $M(G)$ by maximizing the contribution from each local term in $M(G)$ to $1$ because each of those terms is the complement of the corresponding local term in $H(G)$. Therefore, the highest eigenvalue of $M(G)$ is $|E|$ and the associated eigenspace is the ground state space of $H(G)$. It follows that the second highest eigenvalue of $M(G)$ corresponds to the spectral gap of $H(G)$, denoted by $\Delta(H(G))$, and is given by
    \begin{align}
        \label{eq:def_gap}
        \Delta (H(G)) = |E| - \lVert M(G) - |E| \cdot GSP \rVert_{\infty}\,,
    \end{align}
    where the operator norm $\lVert M(G) - |E| \cdot GSP \rVert_{\infty}$ is the second largest eigenvalue of $M(G)$ and where $GSP$ is the projector onto the ground state space of $H(G)$.
    
    \begin{theorem}
        \label{thm:max_spec_gap}
        Consider an arbitrary $n$-vertex graph $G(V, E)$ with minimum degree of the graph equal to $\mindeg$. The spectral gap $\Delta(H(G))$ of a frustration-free $2$-local Hamiltonian $H(G)$ on $G(V, E)$, is bounded above by $\mindeg$, that is,
        \begin{equation}
            \Delta(H(G)) \leq \mindeg\,,
        \end{equation}
        if there exists $P^{(1)} _{u} \in \mathbb{C}^{d \times d}$ acting on the site corresponding to the vertex $u \in V$, such that $\mathrm{ran} (P^{(2)} _{(v, w)}) \subset \mathrm{ran}(P^{(1)} _{v} \otimes P^{(1)} _{w})$ and $0 < \mathrm{rank}(P^{(1)} _{u}) < d$.
    \end{theorem}
    \begin{remark}
        The method for upper bounding the spectral gaps given in this section is quite general and can be used to find non-trivial upper bounds on the $H(G)$ as given in Eq.~\eqref{eq:ff_ham} even when local terms are weighted by coefficients.
    \end{remark}
    
    \begin{corollary}
        \label{cor:cg_gap_ub}
        The upper bound on the spectral gap $\Delta(\cg, n, k)$ as defined in \autoref{def:ham_gnk}, where $\cg$ is the $n$-vertex complete graph is $n - 1$ for all values of local dimension $q$ and arbitrary $k$, that is,
        \begin{equation}
            \Delta(\cg, n, k) \leq n - 1\,,
        \end{equation}
        for all values of local dimension $q$ and arbitrary $k$.
    \end{corollary}
    \begin{proof}[Proof of \autoref{cor:cg_gap_ub}]
    In the case of $H(\cg, n, k)$ as defined in \autoref{def:ham_gnk}, we know that $P^{(2)} _{(v, w)}$ can be written as a sum of permutation operators on $k$ copies of $2$-site Hilbert space (refer to \cite{BrandaoFernandoGSL2010EQ}). Since the range of permutation operators of this kind is a subset of the range of tensor product of two permutation operators each defined on $k$ copies of single site Hilbert space, the condition on $P^{(2)}$ required for \autoref{thm:max_spec_gap} is satisfied. Furthermore, $H(\cg, n, k)$ is $2$-local and frustration free \cite{BrandaoFernandoGSL2016LR}. Therefore, the result of \autoref{thm:max_spec_gap} can be applied to $H(\cg, n, k)$. Regardless of $q$ or $k$, the minimum degree of $\cg$ is $n - 1$, hence,
    \begin{equation}
        \Delta(\cg, n, k) \leq n - 1\,.
    \end{equation}
    \end{proof}

    \begin{corollary}[\textbf{Optimality of spectral gap approach}]
        \label{cor:opt_size}
        The smallest upper bound on circuit size using the approach of lower bounding the spectral gaps of Hamiltonians as defined in \autoref{def:ham_gnk} is,
        \begin{align}
            n^2 k + \frac{n}{2}\log(1 / \varepsilon).
        \end{align}
    \end{corollary}
    \begin{proof}[Proof of \autoref{cor:opt_size}]
        Consider an arbitrary connected graph $\g(V, E)$ and the Hamiltonian $H(G, n, k)$ as defined in \autoref{def:ham_gnk}. We suppress the $k$ dependence in the notation, $H(G, n) = H(G, n, k)$, for it will not be required in this proof. Also, consider the Hamiltonian on the $n$-vertex complete graph, $H(\cg, n)$. Note that $[H(\cg, n), \mathrm{SWAP}_{(i, j)}] = 0$, for all $(i, j) \in E$, where $\mathrm{SWAP}_{(i, j)}$ represents the usual swap operator across the $i^{\text{th}}$ and $j^{\text{th}}$ sites. Therefore, the first excited state of $H(\cg, n)$, denote it by $\ket{\psi_{\cg}}$ is also an eigenvector of $\{\mathrm{SWAP}_{(i, j)}\}_{(i, j) \in E}$ with usual eigenvalues $+1$ or $-1$. Denote the local terms in the Hamiltonians by $h_{(i, j)}$ for $(i, j) \in E$. Begin with taking the expectation value of $H(G, n)$ with respect to $\ket{\psi_{\cg}}$.
        \begin{align}
            \bra{\psi_{\cg}} H(G, n) \ket{\psi_{\cg}} &= \bra{\psi_{\cg}} \left(\sum_{(i, j) \in E} h_{(i, j)} \right) \ket{\psi_{\cg}} \\
            &= \sum_{(i, j) \in E} \bra{\psi_{\cg}} h_{(i, j)} \ket{\psi_{\cg}} \\
            &= \sum_{(i, j) \in E} \bra{\psi_{\cg}} \mathrm{SWAP}_{(1, i)} \mathrm{SWAP}_{(2, j)} h_{(1, 2)} \mathrm{SWAP}_{(1, i)} \mathrm{SWAP}_{(2, j)} \ket{\psi_{\cg}} \\
            &= \sum_{(i, j) \in E} \bra{\psi_{\cg}} h_{(1, 2)} \ket{\psi_{\cg}} \\
            \label{eq:exp_with_cg}
            &= |E| \bra{\psi_{\cg}} h_{(1, 2)} \ket{\psi_{\cg}},
        \end{align}
        where in the second last equation we used the fact that $\ket{\psi_{\cg}}$ is an eigenvector of \\
        $\{\mathrm{SWAP}_{(i, j)}\}_{(i, j) \in E}$ with eigenvalues $+1$ or $-1$. Using the same logic as above but for $ \Delta(\cg, n) = \bra{\psi_{\cg}} H(\cg, n) \ket{\psi_{\cg}}$, we find that,
        \begin{align}
            \Delta(\cg, n) &= \bra{\psi_{\cg}} H(\cg, n) \ket{\psi_{\cg}} \\
            &= \frac{n(n - 1)}{2} \bra{\psi_{\cg}} h_{(1, 2)} \ket{\psi_{\cg}} \\
            \label{eq:sub_in_exp_cg}
            \implies \bra{\psi_{\cg}} h_{(1, 2)} \ket{\psi_{\cg}} &= \frac{2 \Delta(\cg, n)}{n(n - 1)}.
        \end{align}
        Substituting Eq.~\eqref{eq:sub_in_exp_cg} in Eq.~\eqref{eq:exp_with_cg}, we find that
        \begin{align}
            \label{eq:sub_cg_ub}
            \bra{\psi_{\cg}} H(G, n) \ket{\psi_{\cg}} = \frac{2|E|}{n(n - 1)} \Delta(\cg, n).
        \end{align}
        Substitute the upper bound from \autoref{cor:cg_gap_ub} in Eq.~\eqref{eq:sub_cg_ub},
        \begin{align}
            \label{eq:sub_in_min_ub}
            \bra{\psi_{\cg}} H(G, n) \ket{\psi_{\cg}} \leq \frac{2|E|}{n}.
        \end{align}
        Note that both $H(G, n)$ and $H(\cg, n)$ share the same excited state space (refer to \autoref{lem:low_bnd_spt}), therefore, the expectation of $H(G, n)$ with the vector $\ket{\psi_{\cg}}$ is an upper bound on its spectral gap $\Delta(G, n)$,
        \begin{align}
            &\Delta(G, n) \leq \bra{\psi_{\cg}} H(G, n) \ket{\psi_{\cg}} \leq \frac{2|E|}{n}.
        \end{align}
        Since the basic minimal circuit size upper bound found using the spectral gap method (refer to Eq.~\eqref{eq:basic_size_bound}) is
        \begin{align}
            \frac{|E|}{\Delta(G, n, k)} \left(2nk \log(q) + \log (1 / \varepsilon)\right),
        \end{align}
        its minimum value is obtained when the inequality in  Eq.~\eqref{eq:sub_in_min_ub} is saturated.
    \end{proof}

    \begin{proof}[Proof of \autoref{thm:max_spec_gap}]
        Define one site projectors $P^{(1)} _{u} \in \mathbb{C}^{d \times d}$ acting on the site corresponding to the vertex $u \in V$, such that $\mathrm{ran} (P^{(2)} _{(v, w)}) \subset \mathrm{ran}(P^{(1)} _{v} \otimes P^{(1)} _{w})$ and $0 < \mathrm{rank}(P^{(1)} _{u}) < d$. This condition will be useful later in proving \autoref{lem:max_rgve}, which is in turn required to prove \autoref{lem:max_rgve}. Define the projectors onto the orthogonal complement of $\mathrm{ran}(P^{(1)} _{u})$ by $P^{(1), \perp} _{u} \in \mathbb{C}^{d \times d}$. Recall the operator $M(G)$, defined as
        \begin{align}
            \tag{\ref{eq:inv_ham}}
            M(G) &:= \sum_{(v,  w) \in E} P^{(2)} _{(v, w)} \otimes \id_{V \backslash \{v, w\}}\,.
        \end{align}
    
        For each vertex $v \in V$, decompose $\id_{\{v\}} = P^{(1)} _{v} + P^{(1), \perp} _{v}$. We substitute this decomposition of the identity operator in Eq.~\eqref{eq:inv_ham},
        \begin{align}
            M(G) &= \left[ \sum_{(v,  w) \in E} P^{(2)} _{(v, w)} \left( \bigotimes_{u \in V \backslash \{v, w\}} P^{(1)} _{u} \right) \right] \oplus  R(G), \\
            \label{eq:dec_inv_ham}
            &=: Q(G) \oplus  R(G)\,,
        \end{align}
        where $Q(G)$ is as defined above and $R(G)$ is an operator that is sum of terms each of which is a tensor product with at least one multiplicand as $P^{(1), \perp} _{v}$ for some $v \in V$. We note a crucial observation (that we shall invoke again later in the proof for \autoref{lem:max_rgve}. The decomposition in Eq.~\eqref{eq:dec_inv_ham} is useful because $\mathrm{ran}(R(G)) \perp \mathrm{ran}(GSP)$, or, equivalently, $\mathrm{ran}(GSP) \subseteq \mathrm{ran}(Q(G))$. Using this observation, we see that subtracting $GSP$ from Eq.~\eqref{eq:dec_inv_ham} only changes the eigenvalues of the first term in the direct sum and leaves $R(G)$ as is
        \begin{align}
            M(G) - |E| \cdot GSP &= \left( Q(G) - |E| \cdot GSP \right) \oplus  R(G)\,.
        \end{align}
        Taking the operator norm and using the definition of $\Delta(H(G))$ from Eq.~\eqref{eq:def_gap}, we find
        \begin{align}
            \lVert M(G) - |E| \cdot GSP \rVert_{\infty} &= \max \left\{ \left\lVert Q(G) - |E| \cdot GSP \right\rVert_{\infty}, \lVert R(G) \rVert_{\infty} \right\}\\
            |E| - \lVert M(G) - |E| \cdot GSP \rVert_{\infty} &= \min \left \{ |E| - \left\lVert Q(G) - |E| \cdot GSP \right\rVert_{\infty}, |E| - \lVert R(G) \rVert_{\infty} \right \}\\
            \label{eq:apply_lem_max_rgve}
            \Delta(H(G)) &\leq |E| - \lVert R(G) \rVert_{\infty}\,.
        \end{align}
    
        Now, we note a lemma about $R(G)$ that we prove next. 
    
        \begin{lemma}
            \label{lem:max_rgve}
            The operator norm of $R(G)$ is equal to $|E| - \mindeg$, where $\mindeg$ is the minimum degree of the graph $G(V, E)$, that is,
            \begin{equation}
                \lVert R(G) \rVert_{\infty} = |E| - \mindeg\,.
            \end{equation}
        \end{lemma}
    
        Applying \autoref{lem:max_rgve} to Eq.~\eqref{eq:apply_lem_max_rgve}, completes the proof that
        \begin{equation}
            \Delta(H(G)) \leq \mindeg\,.
        \end{equation}
    \end{proof}

    \begin{proof}[Proof of \autoref{lem:max_rgve}]
        Recall that $E$ is the set of unordered pair of vertices that share an edge in the graph $G(V, E)$. For a subset of vertices $U$, define a subset of $E$ by
        \begin{equation}
            E_{\text{not } U} := \{(v, w) \in E : v \notin U \text{ and } w \notin U\}\,.
        \end{equation}
        Since the full expression for $R(G)$ is complicated, we break it down into $|V| - 2$ parts and refer to each part by $S_i$. Consider the expression for $S_i$,
        \begin{align}
            S_i &= \bigoplus_{U \subset V : |U| = i} \left( \bigotimes_{u \in U} P^{(1), \perp} _{u} \right) \otimes \left[ \sum_{(v,  w) \in E_{\text{not } U}} P^{(2)} _{(v, w)} \otimes \left( \bigotimes_{x \in V - (\{v, w\} \cup U)} P^{(1)} _{x} \right) \right]\\
            \label{eq:inv_ham_rep}
            &= \bigoplus_{U \subset V : |U| = i} \left( \bigotimes_{u \in U} P^{(1), \perp} _{u} \right) \otimes Q(G(V - U, E_{\text{not } U}))\,,
        \end{align}
        where we have expressed $S_i$ as the direct sum over all subsets $U$ of $V$ of order $i$, for each $i \in \{1, 2, \dots, |V| - 2\}$. Furthermore, the terms of that direct sum are proportional to the operator $Q$ (refer to Eq.~\eqref{eq:dec_inv_ham}) but on the sub-graph $G(V - U, E_{\text{not } U})$ of the original graph $G(V, E)$ that neither contains the vertices in $U$ nor edges that connect to vertices in $U$. The expression for $R(G)$ in terms of $S_i$ reads,
        \begin{equation}
            R(G) := \bigoplus_{i = 1} ^{|V| - 2} S_i\,.
        \end{equation}
        Now, we evaluate the operator norm of $R(G)$,
        \begin{align}
            \lVert R(G) \rVert_{\infty} &= \left \lVert \bigoplus_{i = 1} ^{|V| - 2} S_i \right \rVert_{\infty}\\
            &= \max_{i \in \{1, 2, \dots, |V| - 2\} } \lVert S_i \rVert_{\infty}\,,
        \end{align}
        where we used the fact that the operator norm of a direct sum of terms is the max of of the operator norms of the terms. Continuing with the derivation, we substitute Eq.~\eqref{eq:inv_ham_rep} in place of $S_i$ and again make use of the last fact above,
        \begin{align}
            \lVert R(G) \rVert_{\infty} &= \max_{i \in \{1, 2, \dots, |V| - 2\} } \left \lVert \bigoplus_{U \subset V : |U| = i} \left( \bigotimes_{u \in U} P^{(1), \perp} _{u} \right) \otimes Q(G(V - U, E_{\text{not } U})) \right \rVert_{\infty} \\
            &= \max_{i \in \{1, 2, \dots, |V| - 2\} } \max_{U \subset V : |U| = i} \left \lVert \left( \bigotimes_{u \in U} P^{(1), \perp} _{u} \right) \otimes Q(G(V - U, E_{\text{not } U})) \right \rVert_{\infty}\,.
        \end{align}
        Since the operator norm of a tensor product of terms is equal to the product of the operator norms of the terms
        \begin{align}
            \label{eq:sub_norm_q}
            \lVert R(G) \rVert_{\infty} &= \max_{i \in \{1, 2, \dots, |V| - 2\} } \max_{U : |U| = i} \left \lVert Q(G(V - U, E_{\text{not } U})) \right \rVert_{\infty}\,,
        \end{align}
        where some of the terms in the tensor product were projectors and, thus, did not contribute to its operator norm. Consider the following Lemma proved in the following subsection,
        \begin{lemma}
            \label{lem:sub_ham_ff}
            If $H(G(V, E))$ is frustration-free with $0$ ground state energy, then $H(G(V - U, E_{\text{not } U}))$ is also frustration free with $0$ ground state energy.
        \end{lemma}
    
        Observe that $Q(G(V - U, E_{\text{not } U}))$ is a sum of tensor products of one and two site ground state projectors of the Hamiltonian $H(G(V - U, E_{\text{not } U}))$, that is, projectors from $\{P^{(1)} _{u}\}_{u \in V - U}$ and $\{P^{(2)} _{(v, w)}\}_{(v, w) \in E_{\text{not } U}}$. By \autoref{lem:sub_ham_ff},  there exists a state $\ket{\phi} \in (\mathbb{C}^d)^{\otimes |V - U|}$ which lies in the kernel of each local term of $H(G(V - U, E_{\text{not } U}))$. Because of the conditions demanded in \autoref{thm:max_spec_gap}, there exists a state $\ket{\phi} \in (\mathbb{C}^d)^{\otimes |V - U|}$ that lies in the range of each term in the sum that is $Q(G(V - U, E_{\text{not } U}))$. Together with the fact that there are $|E_{\text{not } U}|$ terms in $Q(G(V - U, E_{\text{not } U}))$, we have $\lVert Q(G(V - U, E_{\text{not } U}))\rVert_{\infty} = |E_{\text{not } U}|$. Substituting this observation in Eq.~\eqref{eq:sub_norm_q}, we find
        \begin{align}
            \label{eq:apply_ff}
            \lVert R(G) \rVert_{\infty}  &= \max_{i \in \{1, 2, \dots, |V| - 2\} } \max_{U : |U| = i} |E_{\text{not } U}|\,.
        \end{align}
        Removing vertices can only reduce the number of edges in the graph, therefore, the maximum of Eq.~\eqref{eq:apply_ff} must be attained for $i = 1$
        \begin{align}
            \label{eq:few_ver_few_eds}
            \lVert R(G) \rVert_{\infty} &= \max_{U : |U| = 1} |E_{\text{not } U}|\,.
        \end{align}
        Finally, the $\mathrm{r.h.s.}$ of Eq.~\eqref{eq:few_ver_few_eds} asks, ``What is the maximum number of edges that remain in the graph after removing a single vertex?'' The answer to this question is equal to the number of edges of the graph $|E|$ minus the fewest number of edges that must be deleted to disconnect a vertex of the graph. Later is, by definition, the minimum degree $\mindeg$ of the graph. Therefore, 
        \begin{equation}
            \lVert R(G(V, E)) \rVert_{\infty} = |E| - \mindeg\,.
        \end{equation}
    \end{proof}

    \begin{proof}[Proof of \autoref{lem:sub_ham_ff}]
        Since $H(G(V, E))$ is frustration free, there exists a state $\ket{\psi} \in (\mathbb{C}^d)^{\otimes n}$ such that $\bra{\psi} P^{(2), \perp} _{(v, w)} \ket{\psi} = 0$, for all $(v, w) \in E$. In particular, there exists a state $\ket{\psi} \in (\mathbb{C}^d)^{\otimes n}$ such that $\bra{\psi} P^{(2), \perp} _{(v, w)} \ket{\psi} = 0$, for all $(v, w) \in E_{\text{not } U}$. Observe that, for all $(v, w) \in E_{\text{not } U}$,
        \begin{equation}
            \bra{\psi} P^{(2), \perp} _{(v, w)} \ket{\psi} = \mathrm{tr}\big(\ket{\psi} \bra{\psi} P^{(2), \perp} _{(v, w)}\big) = \mathrm{tr}_{V - U}(\mathrm{tr}_{U}(\ket{\psi} \bra{\psi} P^{(2), \perp} _{(v, w)})) = \mathrm{tr}_{V - U}(\mathrm{tr}_{U}(\ket{\psi} \bra{\psi}) P^{(2), \perp} _{(v, w)})\,.
        \end{equation}
        Define the density matrix $\rho_{V - U} := \mathrm{tr}_{U}(\ket{\psi} \bra{\psi}) \in (\mathbb{C}^d)^{\otimes |V - U|} \times (\mathbb{C}^d)^{\otimes |V - U|}$, with eigendecomposition, $\rho_{V - U} := \sum_{i = 1} ^{r} \lambda_{i} \ket{\lambda_i} \bra{\lambda_i}$, where $\lambda_i > 0$, $r$ is the Schmidt rank of $\rho_{V - U}$ and $\ket{\lambda_i} \in (\mathbb{C}^d)^{\otimes |V - U|}$ are pure states. Since, for all $(v, w) \in E_{\text{not } U}$, $\bra{\psi} P^{(2), \perp} _{(v, w)} \ket{\psi} = 0$,
        \begin{align}
            &\mathrm{tr}_{V - U} (\sum_{i = 1} ^{r} \lambda_i \ket{\lambda_i} \bra{\lambda_i} P^{(2), \perp} _{(v, w)}) = 0 \\
            \implies & \sum_{i = 1} ^{r} \lambda_i  \bra{\lambda_i} P^{(2), \perp} _{(v, w)} \ket{\lambda_i} = 0 \\
            \label{eq:not_eig_vec}
            \implies & \bra{\lambda_i} P^{(2), \perp} _{(v, w)} \ket{\lambda_i} = 0 \\
            \label{eq:eig_vec_bec_prj}
            \implies & \bra{\lambda_i} P^{(2), \perp} _{(v, w)} P^{(2), \perp} _{(v, w)} \ket{\lambda_i} = 0 \\
            \label{eq:sub_ham_ff_con}
            \implies & \lVert P^{(2), \perp} _{(v, w)} \ket{\lambda_i} \rVert_2 ^2 = 0,
        \end{align}
        for each $i \in \{1, 2, \dots, r\}$ and for all $(v, w) \in E_{\text{not } U}$. Note that in general from an equation of the form of Eq.~\eqref{eq:not_eig_vec}, we can not conclude that $\ket{\lambda_i}$ are eigenvectors of $P^{(2)} _{(v, w)}$ with eigenvalue $0$. However, in this case, since $P^{(2)} _{(v, w)}$ are projectors, we can write Eq.~\eqref{eq:eig_vec_bec_prj} and transform the $\mathrm{l.h.s.}$ into a vector norm. Since vector norm is $0$ iff all elements of the vector are $0$, $\ket{\lambda_i}$ are in fact eigenvectors of $P^{(2)} _{(v, w)}$ with eigenvalue $0$. Since $H(G(V - U, E_{\text{not } U})) \geq 0$ and $\{\ket{\lambda_i}\}_{i = 1} ^r$ is the set of eigenstates with eigenvalue $0$, $\{\ket{\lambda_i}\}_{i = 1} ^r$ is the set of ground states of the Hamiltonian $H(G(V - U, E_{\text{not } U}))$. Since for each $i \in \{1, 2, \dots, r\}$, $\ket{\lambda_i}$ is an eigenvector of $P^{(2)} _{(v, w)}$ with eigenvalue $0$ for all $(v, w) \in E_{\text{not } U}$ by Eq.~\eqref{eq:sub_ham_ff_con}, $H(G(V - U, E_{\text{not } U}))$ is frustration free with $0$ ground state energy. 
    \end{proof}

\section{Knabe for complete graph RQCs: A short proof of Harrow-Low \label{sec:short_hl08}}
In this section we prove another result for complete-graph random quantum circuits using Knabe bounds, which constitutes a short proof of the result by Harrow and Low \cite{HL08} that complete-graph RQCs form approximate 2-designs. Their proof, with an error fixed in \cite{DinizComm11}, focused on bounding the mixing time of the Markov chain to prove convergence to approximate 2-designs. Here we point out that one can derive a Knabe bound for frustration-free Hamiltonians on complete graphs and then, by explicitly computing the $n=3$ Hamiltonian gap, use the finite-size criteria to establish the same convergence to approximate 2-designs in $t=O(n^2)$ depth. We emphasize that it is believed that complete-graph RQCs are believed to mix even faster, and that this, potentially weak, behavior is a consequence of bounding diamond norms with operator norms. If the true behavior is $t=O(n\log(n))$, as is often conjectured \cite{HaydenPatrick2007BH,SekinoSusskind,BrownWinton2012SS,HarrowAram2023AU,HaferkampJonas2021IS}, then such a scaling is not present in the operator norm or, equivalently, the spectral gaps. 

\begin{theorem}[complete-graph gaps]\label{thm:cgk2gap}
    The complete-graph Hamiltonian gaps for the second moment are $\Delta(\cg, n, k=2) = \Theta(n)$. Specifically, the $k=2$ spectral gaps are lower bounded as
    \begin{equation}
        \Delta(\cg,n,2)\geq (n-2) \left(1-\frac{2q}{q^2+1}\right)\,.
    \end{equation}
\end{theorem}

\begin{proof}
    The upper bound of $\Delta(\cg,n,k)\leq (n-1)$ holds for all moments, as was proved in \autoref{cor:cg_gap_ub}. To prove the lower bound on the complete-graph spectral gap, we combine a newly derived Knabe bound with the exact calculation of a finite-size gap. Specifically, in \autoref{thm:cgknabe} we prove a Knabe-type bound which allows us to related the complete-graph Hamiltonian spectral gap on $n$ sites to the gap on a subsystem of $m$ sites for any moment. For subsystem size $m=3$, the bound is
    \begin{equation}
        \Delta(\cg,n,k)\geq (n-2) \big( \Delta(\cg,3,k) - 1\big)\,.
    \end{equation}
    For the second moment, we can analytically diagonalize the moment operator and in \autoref{prop:exactcggap} find the exact $k=2$ 3-site gap
    \begin{equation}
        \Delta(\cg,3,2) = 2\left(1-\frac{q}{q^2+1}\right)\,.
    \end{equation}
    Combining these two the theorem then follows.
\end{proof}

Recalling how the spectral gap controls the design depth, we then conclude:
\begin{corollary} Random quantum circuits on $n$ qudits with local dimension $q$ on a complete graph form $\ep$-approximate unitary $2$-designs when the circuit depth $t$ is
\begin{equation}
    t\geq n \left(1-\frac{2q}{q^2+1}\right)^{-1} (n\log q + \log 1/\ep)\,.
\end{equation}

\end{corollary}
This is the same result proved in \cite{HL08}, albeit using an entirely different approach. Again, we emphasize that the true design depth for complete-graph circuits might be sub-quadratic, but that any such improvement cannot be seen from spectral gaps, and thus in an approach using the operator norm.

\subsection{Finite-size criteria for non-local systems}

\begin{theorem}\label{thm:cgknabe}
    Let $H(\cg, n)=\sum_{i>j} h_{i,j}$ be a frustration-free Hamiltonian defined on a complete graph, where the Hamiltonian terms are local projectors $h^2_{i,j}= h_{i,j}$, and let $\Delta(\cg, n)$ denote the spectral gap of $H(\cg, n)$. For $n\geq m\geq 3$, the complete-graph Hamiltonian gaps obey
    \begin{equation}
        \Delta(\cg, n) \geq \frac{n-2}{m-2}\left(\Delta(\cg, m)-\frac{n-m}{n-2}\right)\,.
    \end{equation}
\end{theorem}

\begin{proof}[Proof of \autoref{thm:cgknabe}]
Again, as in \cite{knabe1988energy}, we proceed by lower bounding the square of $H(\cg, n)$, noting that $H(\cg, n)^2\geq \gamma H(\cg, n)$ implies a spectral gap lower bound $\Delta(\cg, n)\geq \gamma$. For the complete-graph Hamiltonian, we have
\begin{equation}
    (H(\cg, n))^2 = \sum_{i<j}h_{i,j} + \sum_{\substack{i<j,k<l\\ |\{i,j\}\cap \{k,l\}| = 1}}\{h_{i,j},h_{k,l}\} + \sum_{\substack{i<j,k<l\\ \{i,j\}\cap \{k,l\}=0}}\{h_{i,j},h_{k,l}\} = H(\cg, n) + Q + R\,,
\end{equation}
where $\{h_{i,j},h_{k,l}\} $ denotes the anticommutator of Hamiltonian terms and $Q=\sum_{|\{i,j\}\cap \{k,l\}|=1}\{h_{i,j},h_{k,l}\}$ are the anticommutators of terms which overlap on a single site and $R=\sum_{\{i,j\}\cap \{k,l\}= 0}\{h_{i,j},h_{k,l}\}$ and the anticommutators of non-overlapping Hamiltonian terms. 

Let $s_2(\{j_1,j_2,\ldots j_m\}):=\binom{\{j_1,j_2,\ldots j_m\}}{2}$ denote the set of all length-2 subsets of $\{j_1,j_2,\ldots j_m\}$. Now consider the following operator, defined as the sqaure of the complete graph Hamiltonian on a subset of $m$ of the sites, summed over all possible choices of $m$ sites. We find
\begin{equation}
    \sum_{j_1<j_2<\ldots<j_m} \bigg(\sum_{s\in s_2(\{j_1,j_2,\ldots j_m\})}h_{s(1),s(2)}\bigg)^2 = \binom{n-2}{m-2}H(\cg,n) + \binom{n-3}{m-3}Q + \binom{n-4}{m-4}R\,.
\end{equation}
We further note that this sum of all possible $m$-site complete-graph Hamiltonians can be lower bounded as
\begin{align}
    \sum_{j_1<j_2<\ldots<j_m} \bigg(\sum_{s\in s_2(\{j_1,j_2,\ldots j_m\})}h_{s(1),s(2)}\bigg)^2 &\geq \Delta(\cg,m) \sum_{j_1<j_2<\ldots<j_m} \bigg(\sum_{s\in s_2(\{j_1,j_2,\ldots j_m\})}h_{s(1),s(2)}\bigg)\\
    &= \Delta(\cg,m) \binom{n-2}{m-2} H(\cg, n)\,,
\end{align}
where $\Delta(\cg,m)$ is the $m$-site complete-graph spectral gap. We are simply using that $(H(\cg, m))^2\geq \Delta(\cg,m) H(\cg, m)$ for the subsystem Hamiltonians in terms of the subsystem gaps. The two above equations subsequently imply that
\begin{equation}
    Q + \frac{m-3}{n-3}R \geq \frac{n-2}{m-2} (\Delta(\cg,m)-1) H(\cg, n)\,,
\end{equation}
which in turn gives a lower bound on $(H(\cg, n))^2$. The desired Knabe bound for the complete-graph Hamiltonian then follows.
\end{proof}

For subsystem size $m=3$, the above complete-graph Knabe bound is simply
\begin{equation}
    \Delta(\cg,n) \geq (n-2)\left(\Delta(\cg,3)-\frac{n-3}{n-2}\right)\geq (n-2)\left(\Delta(\cg,3)-1\right)\,.
\end{equation}
Therefore, if we compute a subsystem gap and the gap is strictly greater than the threshold value of one, then we prove a gap lower bound for all $n>3$.

\subsection{Exact computation of finite-size gaps for the complete graph}
Turning back to our frustration-free Hamiltonian from random quantum circuits on a complete graph, we can exactly compute the $n=3$ Hamiltonain gap for the second moment.

\begin{proposition}\label{prop:exactcggap}
For $k=2$, the $n=3$ complete-graph Hamiltonian $H(\cg, 3,2) = h_{1,2}+h_{2,3}+h_{3,1}$, where $h_{i,j} = (\id - (P_H)_{i,j}\otimes \id_{[3]\backslash \{i,j\}})$ and $(P_H)_{i,j} = \int d\mu_{\rm Haar}\, U^{\otimes 2,2}_{i,j}$ has a spectral gap
\begin{equation}
    \Delta(\cg, 3,2) = 2\left(1-\frac{q}{q^2+1}\right)\,.
\end{equation}
For local qubits, the gap is $\Delta(\cg, 3,2) = 6/5$.
\end{proposition}

\begin{proof}[Proof of \autoref{prop:exactcggap}]
Consider the difference of the $n=3$ complete-graph second moment operator and the Haar moment operator
\begin{equation}
    M(\cg, 3,2) - (P_H)_{1,2,3} = \frac{1}{3}\big((P_H)_{1,2}\otimes\id_3+(P_H)_{2,3}\otimes\id_1+(P_H)_{1,3}\otimes\id_2\big) - (P_H)_{1,2,3}\,,
\end{equation}
where $P_H$ is the Haar projector on the specified sites. For simplicity in the following proof, we compress the notation as $P_{12} = (P_H)_{1,2}$. As proven in \cite{BrandaoFernandoGSL2010EQ}, it is sufficient to diagonalize the operator 
\begin{equation}
X = P_{12} \otimes P_3 - P_{123} + P_{23} \otimes P_1 - P_{123} + P_{13} \otimes P_2 - P_{123}\,,
\end{equation}
which has the same nonzero and nonunital eigenvalues as $3(M(\cg, 3,2) - (P_H)_{1,2,3})$. Now we want to write down an explicit operator basis for the operator.  

Let $\cV^{([n])} = \spn\{\id^{\otimes n},S^{\otimes n}\} $. First note that a basis for this space of operators is given by the projectors on to the symmetry and anti-symmetric subspaces, $P_\pm^{(n)} = \frac{1}{2} (\id^{\otimes n}\pm S^{\otimes n})$. For $\cV^{(1)}$, we have $P_\pm = \frac{1}{2} (\id\pm S)$, for which
\begin{equation}
    \tr(P_+P_-)=\tr(P_-P_+)=0 \quad{\rm and}\quad \tr(P_+)=\frac{q(q+1)}{2}\,, \quad \tr(P_-)=\frac{q(q-1)}{2}\,.
\end{equation}

To diagonalize the operator $P_{12}\otimes P_3 - P_{123}$, it suffices to construct a basis for $\cV^{(12)}\otimes \cV^{(3)}$ which is orthogonal to $\cV^{(123)}$, which is rank 2. Denote these two basis operators as $V^{(12,3)}_a$ and $V^{(12,3)}_b$. Similarly for $P_{23}\otimes P_1 - P_{123}$ and $P_{13}\otimes P_2 - P_{123}$ we have $V^{(23,1)}_{a,b}$ and $V^{(13,2)}_{a,b}$. These operators can be written explicitly in terms of the tensored projectors on to the symmetric and antisymmetric subspaces as
\begin{align*}
V_a^{(12,3)} &= \left( \frac{2(q^2+1)}{q^3 (q^3+1)} \right)^{1/2} \left( \frac{q-1}{q^2+1}  (P_{+++} + P_{--+}) - \frac{q+1}{q^2-1} (P_{+--} + P_{-+-})\right)\\
V_b^{(12,3)} &= \left( \frac{2(q^2+1)}{q^3 (q^3-1)}\right)^{1/2} \left( \frac{q+1}{q^2+1} (P_{++-} + P_{---}) - \frac{q-1}{q^2-1} (P_{+-+} + P_{-++})\right)\\
V_a^{(23,1)} &= \left( \frac{2(q^2+1)}{q^3 (q^3+1)} \right)^{1/2} \left( \frac{q-1}{q^2+1}  (P_{+++} + P_{+--}) - \frac{q+1}{q^2-1} (P_{-+-} + P_{--+})\right)\\
V_b^{(23,1)} &= \left( \frac{2(q^2+1)}{q^3 (q^3-1)}\right)^{1/2} \left( \frac{q+1}{q^2+1} (P_{-++} + P_{---}) - \frac{q-1}{q^2-1} (P_{++-} + P_{+-+})\right)\\
V_a^{(13,2)} &= \left( \frac{2(q^2+1)}{q^3 (q^3+1)} \right)^{1/2} \left( \frac{q-1}{q^2+1}  (P_{+++} + P_{-+-}) - \frac{q+1}{q^2-1} (P_{+--} + P_{--+})\right)\\
V_b^{(13,2)} &= \left( \frac{2(q^2+1)}{q^3 (q^3-1)}\right)^{1/2} \left( \frac{q+1}{q^2+1} (P_{+-+} + P_{---}) - \frac{q-1}{q^2-1} (P_{++-} + P_{-++})\right)\,,
\end{align*}
where for convenience we have denoted $P_{\pm\pm\pm} = P_\pm \otimes P_\pm \otimes P_\pm$.

Computing the matrix of inner products of all the basis operators, ordered as\\ $\{V_a^{(12,3)},V_b^{(12,3)},V_a^{(23,1)},V_b^{(23,1)},V_a^{(13,2)},V_b^{(13,2)}\}$, we find
\begin{equation}
\begin{pmatrix}
 1 & 0 & \frac{q}{q^2+1} & 0 & \frac{q}{q^2+1} & 0 \\
 0 & 1 & 0 & -\frac{q}{q^2+1} & 0 & -\frac{q}{q^2+1} \\
 \frac{q}{q^2+1} & 0 & 1 & 0 & \frac{q}{q^2+1} & 0 \\
 0 & -\frac{q}{q^2+1} & 0 & 1 & 0 & -\frac{q}{q^2+1} \\
 \frac{q}{q^2+1} & 0 & \frac{q}{q^2+1} & 0 & 1 & 0 \\
 0 & -\frac{q}{q^2+1} & 0 & -\frac{q}{q^2+1} & 0 & 1 \\
\end{pmatrix}\,.
\end{equation}

Writing the matrix as $X = \sum_{i=a,b}\ketbra{V_i^{(12,3)}}+\ketbra{V_i^{(23,1)}}+\ketbra{V_i^{(13,2)}}$
we can express a state as $\ket{\lambda} = \sum_i \alpha_i \ket{V_i}$ and compute the eigenvalues of $X$ to find
\begin{equation}
\lambda = \left\{\frac{(q+1)^2}{q^2+1}, \frac{q^2+q+1}{q^2+1}, \frac{q^2+q+1}{q^2+1}, \frac{q^2-q+1}{q^2+1}, \frac{q^2-q+1}{q^2+1},\frac{(q-1)^2}{q^2+1}\right\}\,,
\end{equation}
the first of which is the largest eigenvalue of $X$. The largest eigenvalue of $X/3$ equals the second largest eigenvalue of $M(\cg, 3,2)$, as we subtract off the projector onto the highest eigenvalue eigenspace, we compute the spectral gap of $H(\cg,3,2)$ simply by noting that $H(\cg, 3,2) = 3(\iden-M(\cg, 3,2))$.
\end{proof}

We note that this exact complete-graph gap for the second moment agrees with the gap for the  1D Hamiltonian with periodic boundary conditions (which is the same Hamiltonian) computed for local qubits, $q=2$, in Ref.~\cite{ZnidaricMarko2008EC}. There the 1D Hamiltonian is rewritten in terms of an integrable spin system. While the integrability is lost for the complete-graph Hamiltonian, a semiclassical limit allows one to estimate the large $n$ asymptotic value of the gap. We comment on this and the rigorous lower bound from Knabe at the end of the section.

\begin{figure}[ht!]
    \centering
    \begin{tikzpicture}
        \node[anchor = base] (a) at (0, 0) {\includegraphics[width=0.54\linewidth]{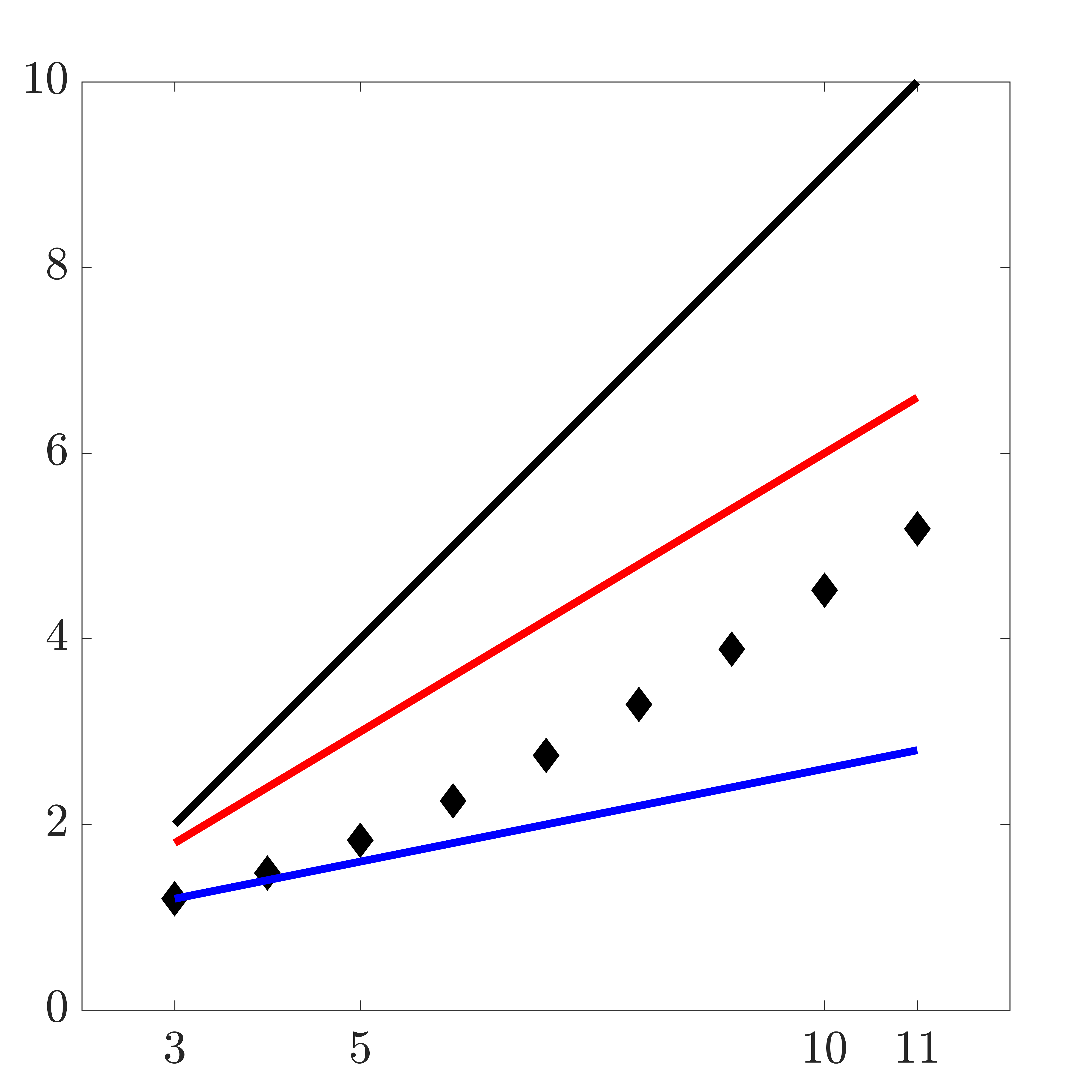}};
        \node[anchor = base, align = left, rotate = 90] at (a.west) {$\Delta(\cg, n, k = 2)$};
        \node at (a.south) {$n$};
    \end{tikzpicture}
    \caption{The numerical $q=2$ complete-graph gaps for the second moment are plotted (black diamonds) along with the upper bound from \autoref{cor:cg_gap_ub} (black line) and lower bound from \autoref{thm:cgk2gap} (blue line) and semiclassical asymptotic behavior (red line).}
    \label{fig:cggaps}
\end{figure}

\subsection{Asymptotic complete-graph gaps}
For the second moment, we find $\Delta(\cg,n,2) = \Theta(n)$, with an upper bound proved in \autoref{cor:cg_gap_ub} and lower bound proved in \autoref{thm:cgk2gap}. For local qubits $q=2$, the bounds imply $n/5\leq \Delta(\cg,n,2)\leq n$. In Ref.~\cite{ZnidaricMarko2008EC}, it was argued using a semiclassical approximation that the asymptotic $q=2$ complete-graph gaps for the second moment should be, in our notation,
\begin{equation}
    \Delta(\cg,n,2) \sim 3/5\,.
\end{equation}
(Note that there the moment operator is normalized differently, contributing a factor of 2.) By numerically computing the second moment complete graph gaps, we find that the behavior is consistent with the (nonrigorous) semiclassical asymptotic. The numerical gaps are plotted in \autoref{fig:cggaps} along with upper and lower bounds, as well as the asymptotic behavior. Furthermore, we can insert the numerical gaps into the finite-size criteria for increasing subsystem size and find $\Omega(n)$ lower bounds on the gaps with increasing slopes, all consistent with the large $n$ behavior.

\section{Outlook \label{sec:outlook}}
We considered quantum circuits with randomly drawn gates on arbitrary architectures. Our motivation was to understand the effect of the architecture on the design size, and, in particular, if certain architectures enable faster convergence to unitary designs compared to other architectures. To that end, we proved that random quantum circuits on a very general set of graphs form approximate unitary $k$-designs in polynomial in $n$ and $k$ number of gates. We further identified a large set of graphs for which the design size is $O(n^2 {\rm poly}(k))$, which typically corresponds to linear depth depending on the architecture. Since the $k$ dependence of our design size bounds found via our Detectability Lemma method is inherited from the corresponding bound for 1D, any improvement in the $k$-dependence for the 1D case would directly imply a corresponding improvement for all architectures. Using the same approach, we show for the first time a design size upper bound for random quantum circuits on all connected architectures, albeit with quasi-polynomial dependence on $n$. We believe that this suboptimal scaling without any assumptions about the connected graph is an artifact of our proof technique. We emphasize that our results also hold for random quantum circuits with local gates drawn from any universal gate set.

In establishing the desired spectral gap bounds, we developed new techniques that we expect to be broadly applicable in proving spectral gap lower bounds for local, frustration free Hamiltonians defined on arbitrary graphs. These methods include lower bounding the Hamiltonian gap using the Detectability Lemma and a Knabe bound which holds on all connected graphs.

Our approach involves bounding the spectral norm to show convergence to approximate unitary designs. We established that the $n$ dependence cannot be improved past $O(n^2)$ gates using this approach. Therefore, proving that random quantum circuits form unitary designs in sub-linear depth for generic local architectures necessitates moving away from spectral gaps and to stronger norms, for which far fewer techniques are known.

\vspace*{12pt}
\ni {\it Note added:} After the completion of this work we learned of Ref.~\cite{ClarkToAppear}, which independently studies the design properties of RQCs on general architectures and will appear in the same arXiv posting. Their models consist of deterministic arrangements of gates, as opposed to our randomly placed gate per time step.

\subsection*{Acknowledgments}
The authors would like to acknowledge helpful discussions with Anurag Anshu, Jonas Haferkamp, Brian Kennedy, Marius Lemm, and David Wendt. 
NHJ thanks the Kavli Institute for Theoretical Physics (supported by NSF Grant PHY-1748958) and the Aspen Center for Physics (supported by NSF grant PHY-2210452) for hospitality during the completion of part of this work. NHJ acknowledges prior support from the Stanford Q-FARM Bloch Fellowship in Quantum Science and Engineering.
SM is supported in part by an OGS Fellowship at UT Austin.

\appendix

\section{Proof of \autoref{lem:any_g_1d} in six steps \label{app:proof_lem_3}}

\subsection{Step~1: Spectral gaps for spanning trees are sufficient}
        One of the main features of the Hamiltonians that are the focus of our work is that regardless of the graph (as long as it is connected) on which they are defined, they have identical ground and excited state spaces. The proof is essentially the same as the proof for the explicit expression for the ground state vectors given in Lemma 17 in Ref.~\cite{BrandaoFernandoGSL2016LR}, and relies on Schur-Weyl duality (see Ref.~\cite{RobertsDanielA2017CA}) and the fact that two qudit unitaries on connected graphs generates the set of all unitaries on $n$ qudits (the result of Ref.~\cite{BarencoAdriano1995EG} implies this). Since the excited state space is the complement of the ground state space in the $n$-qudit Hilbert space, if the latter are identical across all connected graphs, then so will be the former.
        \begin{remark}
            Results for convergence of random quantum circuits (local or otherwise with specific procedure of applying gates) to approximate designs are known for $\1D$, complete and higher-dimensional lattice graphs. Since, complete and hypercubic lattice graphs admit a $\1D$ graph as their spanning tree, by Step~1, the spectral gap of the Hamiltonian on those graphs is lower bounded by the spectral gap of the Hamiltonian on $\1D$ graphs with identical number of vertices. However, such gap lower bounds for complete and hypercubic lattice graphs give weaker bounds than those presented in Refs.~\cite{HaferkampJonas2021IS,HarrowAram2023AU}. This is because the circuit size is proportional to the number of edges in a graph, $|E|$, divided by the spectral gap. Since, $|E|$ then contributes additional factors of $n$ in the circuit size.
        \end{remark}
    \subsection{Step~2: A definition of ``depth'' of a spanning tree \label{app:dep}}
        
        Consider an $n$-vertex tree graph with height $h$ and identify a path between two leaves. Remainder of the tree graph is a sub-tree graph with its root vertex on the chosen path and captures the deviation of the tree from an $n$-vertex $\1D$ line graph. Our strategy to prove the lower bound of \autoref{lem:any_g_1d} is to recursively relate Detectability Lemma norm on a tree with that on another tree with a longer path and a shorter sub-tree. The lower bound becomes smaller exponentially in the number of recursions required, which is equal to the height of the sub-tree. For example, consider an $n$-vertex $G$ in the shape of the letter `y', where a single degree three vertex is connected to a $\1D$ graph at each of its neighboring vertices. The height of the sub-tree in this graph is order $n$ and, according to our method of showing a lower bound, implies an exponentially small in  $n$ lower bound on $\Delta(G, n, k)$. Therefore, we questioned if ours was a viable method to prove a lower bound in \autoref{lem:any_g_1d} that was better than exponentially small in $n$. We answer that question in the affirmative by equating the Detectability Lemma on any $n$-vertex tree graph $T$ of height $h = O(n)$ to that on another tree, denote it by the ``compressed'' tree, of height $O(\log(n))$. This equality between Detectability Lemma norms allows for the possibility of a polynomially small instead of exponentially small in $n$ lower bound on $\Delta(T, n, k)$. We introduce a new quantity that we refer to as the ``depth'' of a tree and define it to be the height of the corresponding ``compressed'' tree. We give the definition of the ``depth'' of a tree in this section, and explicate the structure of the corresponding ``compressed'' tree in the following section. Later in \autorefapp{app:tree_graph_prop}, we prove that the height of the ``compressed'' tree is at most $O(\log{(n)})$.
        
        Consider an arbitrary $n$-vertex connected graph $G$ with a spanning tree denoted by $ST$. We provide an algorithmic function $f$ to define ``depth'' of $ST$. Select any vertex $r \in ST$ as the root of the tree. Then, we denote the depth of $ST$ with respect to $r$ by $\mathrm{Depth}(ST, r)$ and define it to be,
        \begin{align}
            \label{eq:def_depth}
            \mathrm{Depth}(ST, r) := \min_{\mathrm{choices}} \max f(\{ST\}, A, 0),
        \end{align}
        where, 
        \begin{align*} 
            \mathrm{choices} :=\  &\text{The set of all of the following; choices of paths made at step 2.(d).i.A. of $f(\{ST\}, A, d)$} \\
            &\text{including those in every recursion of $f$.}
        \end{align*} 
        $A$ is an array each of whose elements is uniquely mapped to a vertex of the original ST, and $d$ is a counter variable that begins at $0$. By definition, the label for $r$ in $A$ is $0$. 

        \begin{table}[ht!]
            \centering
            \begin{tabular}{|l|}
                \hline
                \parbox{0.9\textwidth}{ \vspace{0.5em} Define a function $f(\mathcal{S}, A, d)$ \vspace{0.5em}} \\
                \hline 
                
                \parbox{0.9\textwidth}{
                \vspace{1em} \% $f$ takes as input a set $\mathcal{S}$ of graphs, a variable array $A$ and a counter variable $d$. $f$ returns a modified version of $A$.
                \begin{enumerate}
                    \item If $\mathcal{S} = \emptyset$,
                    
                    \begin{enumerate}
                        \item Return $A$.
                    \end{enumerate}
                    
                    \item Else,
                    
                    \begin{enumerate}
                        \item For all $s \in \mathcal{S}$, change labels in $A$ for the root vertices $r_s$ to $d$. 
            
                        \item Increment value of $d$ by one.
                        
                        \item If for all $s \in \mathcal{S}, V_s \backslash \{r_s\} = \emptyset$
                        \begin{enumerate}
                            \item Return $A$.
                        \end{enumerate}
                        
                        \item Else, 
                        \begin{enumerate} 
                            \item For all $s \in \mathcal{S} : V_s \backslash \{r_s\} \neq \emptyset$,
                            \begin{enumerate}
                                \item Select a path $V_s$ (a set of connected vertices) from $r_s$ to a leaf in $s$.
                                    
                                \item Change labels in $A$ for vertices in $V_s \backslash \{r_s\}$ to value $d$.
                                
                                \item Change labels in $A$ for vertices connected to $r_s$ to value $d$.
                                
                                \item Increment value of $d$ by one.
                            
                                \item Define sets $CV_s$ of vertices connected to $r_s$.
                            \end{enumerate}
                            
                            \item Define a new set of connected graphs\\
                            $\mathcal{S}' := \bigcup_{s \in \mathcal{S} : V_s \backslash \{r_s\} \neq \emptyset} \{s - (V_s \bigcup CV_s)\}$.
                                
                            \item execute $f(\mathcal{S}', A, d)$.
                        \end{enumerate}
                    \end{enumerate}
                \end{enumerate}
                } \\
                \hline
            \end{tabular}
            \label{tab:def_f_alg}
            \caption{An algorithmic, recursive function $f$ for the definition of the depth of a tree.}
        \end{table}
    
    \subsection{Step~3: A new ``compressed'' tree graph from the spanning tree \label{app:dep_and_cst}}
        \begin{enumerate}
            \item According to the definition of $f$, we select paths to leaves at step 2.(d).i.A. Let us consider those choices of paths such that when those choices are made, $\max f(\{ST\}, \mathcal{A}, 0) = \mathrm{Depth}(ST, r)$. 
            \item For each path $V_s$ that is selected at the end of step 2.(d).i.A., rearrange the vertices in $V_s \backslash \{r_s\}$ in a star graph arrangement around $r_s$. We arrive at a graph that we refer to as the compressed spanning tree or, $CST^{(0)}$, where the superscript denotes the number of iterations of Step 4 presented in the next section and, hence, is $0$ at this step.
            
            \item For a single execution of the algorithm $f$, across all recursions of $f$ and for all $s$ at each recursion, the paths $V_s$ are not intersecting. This is because within each recursion, $s$ labels disconnected graphs and paths from disconnected graphs can not intersect. And, at step 2.(d).iii., each successive recursion of $f$ receives a set of graphs that do not intersect with the paths (given by $V_s$ for any $s \in \mathcal{S}$) from the current recursion.

            \item Suppose that the local term of the moment operator as defined in \autoref{def:resc_mop}, also known as the local ground state space projector, is given by $m_{(\mu, \nu)}$, where $m$ acts non-trivially on the qudit Hilbert spaces corresponding to vertices $\mu$ and $\nu$. 
            
            \item Then, as a consequence of the previous point, two local ground state projectors commute if they act on Hilbert spaces corresponding to edges in two different non-intersecting paths.

            \item Consider the set of non-intersecting paths $\{P_i\}_{i = 1} ^{N}$, where $N$ is the number of non-intersecting paths, and for each $i \in \{1, 2, \dots, N\}$, $P_i$ is a $\1D$ line graph with open boundary conditions with $|P_i|$ number of vertices. 

            \item We consider the product of operators in the Detectability Lemma norm in two parts. First, let $\Pi_{\text{others}}$ be the product of local ground state projectors corresponding to the edges in the graph $ST - \bigcup_{i = 1} ^N P_i$, that is those edges that are not taken into account in any of the non-intersecting paths. Second, consider the product $\Pi_{\text{paths}}$ given as follows
            \begin{align}
                \Pi_{\text{paths}} := \prod_{i = 1} ^{N} m_{(v_{1} ^{(i)} , v_{2} ^{(i)})} m_{(v_{2} ^{(i)} , v_{3} ^{(i)})} \dots m_{(v_{|P_i| - 1} ^{(i)} , v_{|P_i|} ^{(i)})}\,,
            \end{align}
            where the vertices in $P_i$, connected sequentially with each other, are given by $v_{1} ^{(i)}, v_{2} ^{(i)}, \dots, v_{|P_i|} ^{(i)}$. Then, the Detectability Lemma norm for $ST$, is given by,
            \begin{align}
                \label{eq:dl_st}
                \lVert \DL^{ST} \ket{\psi} \rVert = \lVert \Pi_{\text{others}} \Pi_{\text{paths}} \ket{\psi} \rVert\,,
            \end{align}
            where the Detectability Lemma operator is defined as, $\DL^{ST} := \Pi_{\text{others}} \Pi_{\text{paths}}$, and  $\ket{\psi} \in \mathcal{G}_{\perp}$, the excited state space of $H(ST, n, k)$.
            \item Consider the following Lemma:
            \begin{lemma}
                \label{lem:1d_to_star}
                Suppose $P$ is a path graph with vertices that are sequentially connected given by $\{v_{1}, v_{2}, \dots, v_{|P|}\}$, where $|P|$ is the total number of vertices in $P$, then, 
                \begin{align}
                    \label{eq:1d_to_star}
                    m_{(v_{1}, v_{2})} m_{(v_{2}, v_{3})} \cdots m_{(v_{|P| - 1}, v_{|P|})}
                    &= m_{(v_{1}, v_{2})} m_{(v_{1}, v_{3})} \cdots m_{(v_{1}, v_{|P|})} W_{P}\,,
                \end{align} 
                where $W_{P}$ is a cyclic permutation over the Hilbert spaces corresponding to the vertices in $P$ such that it takes $v_{1} \rightarrow v_{|P|}, v_{2} \rightarrow v_{1}, v_{3} \rightarrow v_{2}$ and so forth.
            \end{lemma}
            \begin{remark}
                Note that the proof of Eq.~\eqref{eq:1d_to_star} requires that $m = m \mathrm{SWAP} = \mathrm{SWAP} m$, where we omit the subscripts and take it understood that $m$ and $\mathrm{SWAP}$ act on the same 2 qudit Hilbert spaces. This property is made available to us by left/right invariance of the Haar measure that is invoked in the definition of $m = \int_{\mathcal{U}(q^2)} dU\ U^{\otimes k} \otimes U^{\dagger, \otimes k}$, a local term of the moment operator as defined in \autoref{def:resc_mop}.
            \end{remark}

            \item Using \autoref{lem:1d_to_star}, we can re-write the sequential products in the multiplicand for each $i \in \{1, 2, \dots, N\}$ of Eq.~\eqref{eq:1d_to_star} as a product of local ground state projectors on a star graph centered at $v_{1} ^{(i)}$. For each $i$, this rewriting comes at a cost of cyclic permutation unitary, denoted by $W_{P_i}$, that is right-multiplied to the multiplicand corresponding to $i$ and that acts on all the Hilbert spaces corresponding to all vertices in $P_i$
            \begin{align}
                \label{eq:w_com}
                \Pi_{\text{paths}} := \prod_{i = 1} ^{N} m_{(v_{1} ^{(i)} , v_{2} ^{(i)})} m_{(v_{1} ^{(i)} , v_{3} ^{(i)})} \dots m_{(v_{1} ^{(i)} , v_{|P_i|} ^{(i)})} W_{P_i}\,.
            \end{align}

            \item Note that for each $i \in \{1, 2, \dots, N - 1\}$, $W_{P_i}$ commutes with all terms to its right in Eq.~\eqref{eq:w_com}. This is because $W_{P_i}$ acts on the Hilbert spaces corresponding to vertices in $P_i$ and all terms to the right of $W_{P_i}$ act on Hilbert spaces corresponding to vertices from non-intersecting paths different from $P_i$. Since non-intersecting paths do not have common vertices by definition, $W_{P_i}$ and all terms to its right act on different Hilbert spaces and, hence, commute. Consequently, we can pull all $W_{P_i}$ to the right of the product  
            \begin{align}
                \Pi_{\text{paths}} &:= \prod_{i = 1} ^{N} m_{(v_{1} ^{(i)} , v_{2} ^{(i)})} m_{(v_{1} ^{(i)} , v_{3} ^{(i)})} \dots m_{(v_{1} ^{(i)} , v_{|P_i|} ^{(i)})}  \prod_{i = 1} ^{N} W_{P_i} \\
                \label{eq:sub_in_new_dl}
                &=: \Pi_{\text{stars}} \prod_{i = 1} ^{N} W_{P_i}\,.
            \end{align}

            \item Substituting Eq.~\eqref{eq:sub_in_new_dl} in Eq.~\eqref{eq:dl_st}, we find
            \begin{align}
                \label{eq:sub_dl_cst0}
                \lVert \DL^{ST} \ket{\psi} \rVert = \lVert \Pi_{\text{others}} \Pi_{\text{stars}} \ket{\widetilde{\psi}} \rVert\,,
            \end{align}
            where $\ket{\widetilde{\psi}} := \prod_{i = 1} ^{N} W_{P_i} \ket{\psi}$ and both $\ket{\psi}$ and $\ket{\widetilde{\psi}}$ belong to $\mathcal{G}_{\perp}$. We see this from the invariance of the orthogonality between ground and excited states under permutations. We know that for $\ket{\phi} \in \mathcal{G}$, $\left(\prod_{i = 1} ^{N} W_{P_i}\right)^\dagger \ket{\phi} = \ket{\phi}$, because $\left(\prod_{i = 1} ^{N} W_{P_i}\right)^\dagger$ is a permutation over $n$ Hilbert spaces and the ground states are span of states that are identical across all $n$ Hilbert spaces. Suppose $\ket{\phi} \in \mathcal{G}$, $\ket{\psi} \in \mathcal{G}_{\perp}$ and $\ket{\widetilde{\psi}} = \prod_{i = 1} ^{N} W_{P_i} \ket{\psi}$, then $\langle \phi | \widetilde{\psi} \rangle = \bra{\phi} \left(\prod_{i = 1} ^{N} W_{P_i}\right) \ket{\psi} = \langle \phi | \psi \rangle = 0$. 

            \item Finally, we realize that the product of operators in the $\mathrm{r.h.s.}$ of Eq.~\eqref{eq:sub_dl_cst0} is a product of local ground state projectors corresponding to a graph with all non-intersecting paths replaced by stars, which was precisely the definition of the compressed spanning tree $CST^{(0)}$. Therefore, we can rewrite the $\mathrm{r.h.s.}$ in terms of the Detectability Lemma norm for $CST^{(0)}$ as follows,
            \begin{align}
                \lVert \DL^{ST} \ket{\psi} \rVert = \lVert \DL^{CST^{(0)}} \ket{\widetilde{\psi}} \rVert\,,
            \end{align}
            where the Detectability Lemma operator for $CST^{(0)}$ is defined by, $\DL^{CST^{(0)}} := \Pi_{\text{others}} \Pi_{\text{stars}}$, and both $\ket{\psi}$ and $\ket{\widetilde{\psi}}$ belong to $\mathcal{G}_{\perp}$.
        \end{enumerate}

        \begin{proof}[Proof of \autoref{lem:1d_to_star}]

        We give a graphical proof for the case of $|P| = 5$ vertices. The generalization to arbitrary number of vertices is straight forward.\\
        We begin with the expression for the sequential product of operators arranged according to the path graph,
        \begin{align}
            \PathGraphRQC\ .
        \end{align}
        Then, we re-write the above sequential product with the help of the permutation operator $W_{P}$,
        \begin{align}
            \label{greq:perm_to_swap_perm}
            \PathGraphRQCPerm\ ,
        \end{align}
        where
        \begin{align}
            \CyclicPermutation = W_{P}.
        \end{align}
        Using the identity $m = m \mathrm{SWAP}$, we re-write the expression in Eq.~\eqref{greq:perm_to_swap_perm} as
        \begin{align}
            \label{greq:swap_perm_to_star}
            \PathGraphRQCPermSWAP\ .
        \end{align}
        Equation \eqref{greq:swap_perm_to_star} simplifies to
        \begin{align}
            \PathGraphRQCStarPerm\ ,
        \end{align}
        where we identify the product of operators arranged in a star graph geometry followed by right multiplication by $W_{P}$, 
        \begin{align}
            \StarGraphRQC\ .
        \end{align}
        \end{proof}
    
    \subsection{Step~4: Detectability Lemma norm equality for ``flattened out'' graph \label{app:step4}}
    \begin{enumerate}
        \item Step 4 will be repeated many times. Let each iteration of Step 4 be labelled by $i \in \{1, \dots, d\}$. After Step 3, we begin Step 4 with $i = 1$. 
            
        \item Only when $i = 1$, we do the following. We choose the root vertex of $CST^{(0)}$ to be the same as the root vertex of $ST$. We note that $CST^{(0)}$ must contain at least $3$ leaves or else it is a $\1D$ line graph. If $CST^{(0)}$ is a $\1D$ line graph, then
        \begin{align}
            \label{eq:triv_upb_dl_st}
            \lVert \DL^{ST} \ket{\psi} \rVert^2 = \lVert \DL^{CST^{(0)}} \ket{\widetilde{\psi}} \rVert^2 = \lVert \DL^{\1D} \ket{\widetilde{\psi}} \rVert^2 \leq \cfrac{1}{1 + \cfrac{\Delta(\1D, n, k)}{4}}\,,
            \end{align}
            where $\ket{\psi}, \ket{\widetilde{\psi}} \in \mathcal{G}_{\perp}$, the excited state space of $H(G, n, k)$. Using the upper bound on $\lVert \DL^{ST} \ket{\psi} \rVert^2$ from Eq.~\eqref{eq:triv_upb_dl_st} in the Quantum Union Bound for $ST$, we find a lower bound on $\Delta(ST, n, k)$ and $\Delta(G, n, k)$ as follows
        \begin{align}
            \label{eq:amb_lb}
            \Delta(G, n, k) \geq \Delta(ST, n, k) \geq \frac{1}{4} \left(1 - \lVert \DL^{CST^{(0)}} \ket{\psi} \rVert\right) \geq \frac{\Delta(\1D, n, k)}{16}\,.
        \end{align}
        In general, $CST^{(0)}$ will contain at least $3$ leaves. When that is the case, we choose a path between two of the leaves that contains the root vertex and define it to be the lowest layer of $CST^{(0)}$.
             
        \item Let $V$ denote the vertices in the lowest layer of $CST^{(i-1)}$, (where we will clarify the notion of the ``lowest layer'' for $i - 1 > 0$ at the end of this Step). Place each $v \in V$ on a horizontal line. Each $v \in V$ is connected with at least 2 other vertices from $V$ in the bulk, and at least 1 other vertex from $V$ at the end points. Placing the vertices in $V$ in a horizontal line provides an ordering of those vertices that helps with clarity of the following discussion.
        
        \item Let $V _{\geq 3} \subseteq V$, such that for each $v \in  V _{\geq 3}$, the degree of the vertex $v$ is greater than or equal to $3$. Take any vertex in $V _{\geq 3}$ as a reference vertex for the following discussion and label it with $v_{0}$.
        
        \item Recall that $N(CST^{(i - 1)}, v)$ is referred to as the neighborhood of the vertex $v$ in the graph $CST^{(i - 1)}$. $N(CST^{(i - 1)}, v)$ denotes the set of vertices containing $v$ and all vertices connected to $v$ in the graph $CST^{(i - 1)}$.
        
        \item Suppose $x, y \in V$, then we write $x - y = z$ to mean that the vertex $x$ is $z$ edges away from $y$ along vertices in $V$. If $z > 0$, $x$ is to the right of $y$, and if $z < 0$, then $x$ is to the left of $y$.
        
        \item For each $v \in V _{\geq 3}$, let $C_v := N(CST^{(i - 1)}, v) \backslash \{u \in V :\ u - v = 1\}$.
        
        \item Let $V _{\mathrm{even}} := \{ u \in V :\ u - v_{0} \in \mathbb{Z}_{\mathrm{even}} \}$, and let $V _{\mathrm{odd}} := \{ u \in V :\ u - v_{0} \in \mathbb{Z}_{\mathrm{odd}} \}$, where $\mathbb{Z}_{\mathrm{even}}$ and $\mathbb{Z}_{\mathrm{odd}}$ denote the even and odd integers, respectively.
            
        \item Let $E (X)$ denote the subset of all edges in $CST^{(i - 1)}$ among the vertices in $X$.
        
        \item Define products of projectors $\Pi_{\mathrm{even}}$ and $\Pi_{\mathrm{odd}}$ as follows,
        \begin{align}
            \label{eq:pi_even}
            &\Pi_{\mathrm{even}} = \prod_{ v \in V _{\geq 3} \cap V _{\mathrm{even}} } \left[ \prod_{ e \in E (C_v) } m_e \right] \\
            \label{eq:pi_odd}
            &\Pi_{\mathrm{odd}} = \prod_{ v \in V _{\geq 3} \cap V _{\mathrm{odd}} } \left[ \prod_{ e \in E (C_v) } m_e \right]\,,
        \end{align}
        where when we write $\prod_{ e \in E (C_v) } m_e$, the specific order of the product is arbitrary as long as the first term in the product does \textit{not} correspond to the edge along $V$.
        
        \item The following lemma allows us to rewrite the product $\prod_{ e \in E (C_v) } m_e$ corresponding to the neighborhood of the vertex $v$, that is a star graph, in terms of product of projectors corresponding to a line graph on the vertices in $C_v$. 
        \begin{lemma}
            \label{lem:star_to_1d}
            Suppose $C_v = \{v, v_1, v_2, \dots, v_{|C_v| - 1}\}$, then
            \begin{align}
                \prod_{ e \in E (C_v) } m_e &:= m_{(v, v_1)} m_{(v, v_2)} m_{(v, v_3)} \cdots m_{(v, v_{|C_v| - 1})} \\
                \label{eq:unitary_shift_right}
                &= m_{(v, v_1)} m_{(v_1, v_2)} m_{(v_2, v_3)} \cdots m_{(v_{|C_v| - 3}, v_{|C_v| - 2})} m_{(v_{|C_v| - 2}, v_{|C_v| - 1})} W_{C_v} ^{\dagger}\\ \text{or, equivalently,}\nn
                \label{eq:unitary_shift_left}
                &= W_{C_v} ^\dagger m_{(v_1, v_2)} m_{(v_2, v_3)} \cdots m_{(v_{|C_v| - 3}, v_{|C_v| - 2})} m_{(v_{|C_v| - 2}, v_{|C_v| - 1})} m_{(v, v_{|C_v| - 1})}\,,
            \end{align} 
            where $W_{C_v}$ is a cyclic permutation over the Hilbert spaces corresponding to the vertices in $C_v$ such that it takes $v \rightarrow v_{|C_v| - 1}, v_1 \rightarrow v, v_2 \rightarrow v_1$ and so forth.
        \end{lemma}
        \begin{remark}
            The proof of Eq.~\eqref{eq:unitary_shift_right} in \autoref{lem:star_to_1d} is identical to that for \autoref{lem:1d_to_star}. They are the same lemma up to multiplication by the cyclic permutation unitary, nonetheless, we give \autoref{lem:star_to_1d} explicitly for clarity of the explanations. The proof for Eq.~\eqref{eq:unitary_shift_left} requires a few extra steps, so we give it at the end of this section. In Step 2, \autoref{lem:1d_to_star} relates products of projectors on $\1D$ graphs to products of projectors on star graphs. Since the later visually look like compressed versions of the former, Step 2 was associated with the procedure of ``compression.'' In Step 3, \autoref{lem:star_to_1d} achieves the opposite task compared to \autoref{lem:1d_to_star} in Step 2, thus, Step 3 is associated with the procedure of ``flattening-out.'' The remark for \autoref{lem:1d_to_star} continues to apply here.
        \end{remark}
        
        \item We apply \autoref{lem:star_to_1d} to Eqs.~\eqref{eq:pi_even} and \eqref{eq:pi_odd}. For each $v \in V_{\geq 3} \cap V_{\mathrm{even}}$ in Eq.~\eqref{eq:pi_even}, we apply Eq.~\eqref{eq:unitary_shift_right} to $\prod_{ e \in E (C_v) } m_e$ and find
        \begin{align}
            \Pi_{\mathrm{even}} &= \prod_{ v \in V _{\geq 3} \cap V _{\mathrm{even}} } \left[ m_{(v, v_1)} m_{(v_1, v_2)} \cdots m_{(v_{|C_v| - 2}, v_{|C_v| - 1})} \right] W_{C_v} ^\dagger \\
            \label{eq:w_com_arb_g}
            &= \prod_{ v \in V _{\geq 3} \cap V _{\mathrm{even}} } \left[ m_{(v, v_1)} m_{(v_1, v_2)} \cdots m_{(v_{|C_v| - 2}, v_{|C_v| - 1})} \right] \prod_{ v \in V _{\geq 3} \cap V _{\mathrm{even}} } W_{C_v} ^\dagger \\
            \label{eq:pi_even_new}
            &=: \widetilde{\Pi}_{\mathrm{even}} \prod_{ v \in V _{\geq 3} \cap V _{\mathrm{even}} } W_{C_v} ^\dagger\,, 
        \end{align}
        where, to arrive at Eq.~\eqref{eq:w_com_arb_g}, we notice that that $W_{C_v} ^\dagger$ for different $v \in V_{\geq 3} \cap V_{\mathrm{even}}$ commute with each other because they are defined on disjoint subsets of local Hilbert spaces. We define $\widetilde{\Pi}_{\mathrm{even}}$ as in Eq.~\eqref{eq:pi_even_new} that will be used later. Similarly, for each $v \in V_{\geq 3} \cap V_{\mathrm{odd}}$ in Eq.~\eqref{eq:pi_odd}, we apply Eq.~\eqref{eq:unitary_shift_left} to $\prod_{ e \in E (C_v) } m_e$ and find
        \begin{align}
            \Pi_{\mathrm{odd}} &= \prod_{ v \in V _{\geq 3} \cap V _{\mathrm{odd}} } W_{C_v} ^\dagger \left[ m_{(v_1, v_2)} \cdots m_{(v_{|C_v| - 2}, v_{|C_v| - 1})} m_{(v, v_{|C_v| - 1})} \right] \\
            &= \prod_{ v \in V _{\geq 3} \cap V _{\mathrm{odd}} } W_{C_v} ^\dagger \prod_{ v \in V _{\geq 3} \cap V _{\mathrm{odd}} } \left[ m_{(v_1, v_2)} \cdots m_{(v_{|C_v| - 2}, v_{|C_v| - 1})} m_{(v, v_{|C_v| - 1})} \right] \\
            \label{eq:pi_odd_new}
            &=: \prod_{ v \in V _{\geq 3} \cap V _{\mathrm{odd}} } W_{C_v} ^\dagger \widetilde{\Pi}_{\mathrm{odd}}\,,
        \end{align}
        where, as before, we observe that $W_{C_v} ^\dagger$ for different $v \in V_{\geq 3} \cap V_{\mathrm{odd}}$ commute with each other because they are defined on disjont subsets of local Hilbert spaces, and we defined $\widetilde{\Pi}_{\mathrm{odd}}$ as in Eq.~\eqref{eq:pi_odd_new} to be used later. We can swap our choice of Eqs.~\eqref{eq:unitary_shift_right} and \eqref{eq:unitary_shift_left} as long as we use the same choice for all $\prod_{ e \in E (C_v) } m_e$ for each $v \in V_{\geq 3} \cap V_{\mathrm{even}}$ in $\Pi_{\mathrm{even}}$ and for each $v \in V_{\geq 3} \cap V_{\mathrm{odd}}$ in $\Pi_{\mathrm{odd}}$.
        
        \item All local ground state projectors not accounted for in $\Pi_{\mathrm{even}}$ or $\Pi_{\mathrm{odd}}$, that is, $m$ corresponding to edges not accounted for in $E(C_v)$ for either $v \in V_{\geq 3} \cap V_{\mathrm{even}}$ or $v \in V_{\geq 3} \cap V_{\mathrm{odd}}$, can be multiplied in an arbitrary order to form the product denoted by $\Pi_{\mathrm{others}}$.
        
        \item We write the Detectability Lemma norm for $CST^{(i - 1)}$, denoted by $\lVert \DL^{CST^{(i - 1)}} \lvert \psi \rangle \rVert^2$, as follows
        \begin{align}
            \label{eq:sub_star_to_1d}
            \lVert \DL^{CST^{(i - 1)}} \lvert \psi \rangle \rVert^2 = \lVert \Pi_{\mathrm{odd}} \Pi_{\mathrm{others}} \Pi_{\mathrm{even}} \lvert \psi \rangle \rVert^2\,,
        \end{align}
        where the Detectability Lemma operator is defined by, $\DL^{CST^{(i - 1)}} := \Pi_{\mathrm{odd}} \Pi_{\mathrm{others}} \Pi_{\mathrm{even}}$, and $\ket{\psi} \in \mathcal{G}_{\perp}$. Moving forward, we substitute expressions for $\Pi_{\mathrm{even}}$ and $\Pi_{\mathrm{odd}}$ from Eqs.~\eqref{eq:pi_even_new} and \eqref{eq:pi_odd_new}, respectively. 
        \begin{align}
            \lVert \DL^{CST^{(i - 1)}} \lvert \psi \rangle \rVert^2 &=
            \begin{aligned}[t]
                \mathrm{tr} \left( \left[ \prod_{ v \in V _{\geq 3} \cap V _{\mathrm{odd}} } W_{C_v} ^\dagger \right] \widetilde{\Pi}_{\mathrm{odd}} \Pi_{\mathrm{others}} \widetilde{\Pi}_{\mathrm{even}} \left[ \prod_{ v \in V _{\geq 3} \cap V _{\mathrm{even}} } W_{C_v}  ^\dagger \right] \ketbra{\psi} \right.\\
                \left. \left[ \prod_{ v \in V _{\geq 3} \cap V _{\mathrm{even}} } W_{C_v} ^\dagger \right] ^\dagger \widetilde{\Pi}_{\mathrm{even}} ^\dagger  \Pi_{\mathrm{others}} ^\dagger \widetilde{\Pi}_{\mathrm{odd}} ^\dagger \left[ \prod_{ v \in V _{\geq 3} \cap V _{\mathrm{odd}} } W_{C_v} ^\dagger \right]^\dagger \right)
            \end{aligned}\nn
            &=
            \begin{aligned}[t]
                \mathrm{tr} \Bigg( \widetilde{\Pi}_{\mathrm{odd}} \Pi_{\mathrm{others}} \widetilde{\Pi}_{\mathrm{even}} \left[ \prod_{ v \in V _{\geq 3} \cap V _{\mathrm{even}} } W_{C_v} ^\dagger \right] \ketbra{\psi} \left[ \prod_{ v \in V _{\geq 3} \cap V _{\mathrm{even}} } W_{C_v} ^\dagger \right] ^\dagger \\ \widetilde{\Pi}_{\mathrm{even}} ^\dagger  \Pi_{\mathrm{others}} ^\dagger \widetilde{\Pi}_{\mathrm{odd}} ^\dagger \Bigg)
            \end{aligned} \nn
            \label{eq:def_cst_i}
             &= \mathrm{tr} \Bigg( \widetilde{\Pi}_{\mathrm{odd}} \Pi_{\mathrm{others}} \widetilde{\Pi}_{\mathrm{even}} \,\ketbra{\widetilde{\psi}} \,\widetilde{\Pi}_{\mathrm{even}} ^\dagger  \Pi_{\mathrm{others}} ^\dagger \widetilde{\Pi}_{\mathrm{odd}} ^\dagger \Bigg),
        \end{align}
        where $\ket{\widetilde{\psi}} := \left[ \prod_{ v \in V _{\geq 3} \cap V _{\mathrm{even}} } W_{C_v} ^\dagger \right] \ket{\psi} \in \mathcal{G}_{\perp}$, which follows from the same reasoning as in sub-step 11 in Step 3.

        \item The $\mathrm{r.h.s.}$ of Eq.~\eqref{eq:def_cst_i} is a Detectability Lemma norm for a graph different from $CST^{(i - 1)}$ that we will define to be $CST^{(i)}$. Therefore, the Detectability Lemma operator for $CST^{(i)}$ is defined by, $\DL^{CST^{(i)}} := \widetilde{\Pi}_{\mathrm{odd}} \Pi_{\mathrm{others}} \widetilde{\Pi}_{\mathrm{even}}$, and Eq.~\eqref{eq:def_cst_i} reads,
        \begin{align}
            \label{eq:iter_rhs_dif}
            \lVert \DL^{CST^{(i - 1)}} \lvert \psi \rangle \rVert^2 = \lVert \DL^{CST^{(i)}} \lvert \widetilde{\psi} \rangle \rVert^2\,.
        \end{align}
        All the edges among the $n$ vertices in $CST^{(i - 1)}$ corresponding to the projectors in $\Pi_{\mathrm{others}}$ are present in $CST^{(i)}$. However, the edges that are arranged in star graph fashion in the neighborhood of each $v \in V_{\geq 3}$ in $CST^{(i - 1)}$ are replaced by $\1D$ line graphs in $CST^{(i)}$. This is inferred from the support of the projectors in the products $\widetilde{\Pi}_{\mathrm{even}}$ and $\widetilde{\Pi}_{\mathrm{odd}}$ in Eq.~\eqref{eq:pi_even_new} and Eq.~\eqref{eq:pi_odd_new}, respectively. Alternatively, one may think about $CST^{(i)}$ as follows; recall that we began Step 4, by choosing a path between two leaves in $CST^{(i - 1)}$. Refer to that path as $V$. Then, $CST^{(i)}$ is a graph that results from $CST^{(i - 1)}$ by reconnecting the neighborhood (not adding nor removing edges) of each vertex $v \in V$ of degree $3$ or greater such that all vertices in the neighborhood of $v$ become part of $V$. The resulting path $V$ will be defined as the lowest level of $CST^{(i)}$ for the next iteration of Step 4.

        \begin{proof}[Proof of Eq.~\eqref{eq:unitary_shift_left} of \autoref{lem:star_to_1d}]
        
        We give a proof for $|C_v| = 5$ and the proof for arbitrary number of vertices is a straight forward extension of the given proof. We begin with the left hand side of the equality in the lemma and re-express it as is shown,
        \begin{align}
            \StarGraphRQCLemSix = \left\{ \StarGraphRQCLemSixDag \right\}^\dagger\ .
        \end{align}
        We use the result of Eq.~\eqref{eq:unitary_shift_right} in \autoref{lem:star_to_1d} inside the curly bracket,
        \begin{align}
            \StarGraphRQCLemSix &= \left\{ \PathGraphRQCLemSix \right\}^\dagger \\
            &= \PathGraphRQCLemSixFin\ .
        \end{align}
        \end{proof}
    \end{enumerate}

    \subsection{Step~5: Repeating the ``flattening out'' procedure}
        \begin{enumerate}
            \item It was necessary to arrange the terms in the product in Eq.~\eqref{eq:sub_star_to_1d} in a very specific manner to equate $\lVert \DL^{CST^{(i - 1)}} \lvert \psi \rangle \rVert^2 = \lVert \DL^{CST^{(i)}} \lvert \widetilde{\psi} \rangle \rVert^2$. If we begin with $CST^{(0)}$ that has height $d$ (equal to the ``depth'' of the corresponding $ST$), then at the end of $d$ iterations of Step 4, we would have shown that $\lVert \DL^{CST(d-1)} \ket{\psi} \rVert^2 = \lVert \DL^{CST(d)} \ket{\psi} \rVert^2 = \lVert \DL^{\1D} \ket{\psi} \rVert^2$, where $CST^{(d)}$ is equal to the $n$-vertex $\1D$ line graph with open boundary conditions. The proof of this claim follows from induction and the definition of the ``lowest layer'' (refer to the last point of Step 4).
            \begin{lemma}
                If $CST^{(0)}$ is the compressed spanning tree corresponding to the spanning tree $ST$ of an arbitrary $n$-vertex connected graph $G$ and $CST^{(0)}$ has height $d$, then $CST^{(d)}$ is the $n$-vertex $\1D$ line graph.
            \end{lemma}
            \begin{proof}
                Proof follows from Step 4. $i = 1$ is the base step and is equal to the Step 4 as is written above. The induction hypothesis is that at the end of the $i^{\text{th}}$ iteration of Step 4, $CST^{(i)}$ includes, in its lowest layer, all layers of the tree $CST^{(0)}$ from its lowest to its $i^{\text{th}}$ layer. Consider the $(i + 1)^{\text{th}}$ iteration. Since at this iteration, the neighborhood of all degree $3$ or greater vertices in $CST^{(i)}$ contain vertices from either
                \begin{itemize}
                    \item the lowest layer of $CST^{(i)}$, which by the definition of ``lowest layer'' (at the end of Step 4) and the induction hypothesis contains all vertices up to and including the $i^{\text{th}}$ layer of $CST^{(0)}$, or,
                    \item vertices from the first layer of $CST^{(i)}$, which is equal to the $(i + 1)^{\text{th}}$ layer of $CST^{(0)}$,
                \end{itemize}
                the proof of the induction step is complete. $CST^{(0)}$ has depth $d$ and thus $d$ layers. By the above argument, $CST^{(d)}$ will contain in its lowest layer, all the vertices in $CST^{(0)}$. Since, the ``lowest layer'' is a path graph between two leaves, the prove of the lemma is complete.
            \end{proof}
            
            \begin{remark}
                In Eq.~\eqref{eq:iter_rhs_dif}, $\mathrm{r.h.s.}$ in the $i^{\text{th}}$ iteration and $\mathrm{l.h.s.}$ in the $(i + 1)^{\text{th}}$ of Step 4 are different because the order of product of projectors in the Detectability Lemma operators are different. The affect of this dissimilarity is significant for our eventual lower bound on $\Delta(G, n, k)$. If the stated difference were absent, then the same lower bound as in Eq.~\eqref{eq:amb_lb} would hold because $\lVert \DL^{ST} \ket{\psi} \rVert^2 \leq \lVert \DL^{\1D} \ket{\phi} \rVert^2$, for some $\ket{\psi}, \ket{\phi} \in \mathcal{G}_{\perp}$. Even if the $\mathrm{r.h.s.}$ of Eq.~\eqref{eq:iter_rhs_dif} in the $i^{\text{th}}$ iteration of Step 4 were less than or equal to the $\mathrm{l.h.s.}$ of Eq.~\eqref{eq:iter_rhs_dif} in the $(i + 1)^{\text{th}}$ iteration of the same step, the lower bound of Eq.~\eqref{eq:amb_lb} would hold. However, we do not know of a way to show this. More precisely, if one could show that for every $\ket{\psi} \in \mathcal{G}_{\perp}$, there exists a $\ket{\phi} \in \mathcal{G}_{\perp}$ such that $\lVert \DL^{CST^{(i)}} \ket{\psi} \rVert \leq \lVert \widetilde{\DL}^{CST^{(i)}} \ket{\phi} \rVert$, where $\widetilde{\DL}$ is the Detectability Lemma operator with different order of products of projectors, then lower bound of Eq.~\eqref{eq:amb_lb} holds. We do not know about the truth of falsity of this statement. Instead we take a different approach to relate the $\mathrm{r.h.s.}$ in the $i^{\text{th}}$ iteration and $\mathrm{l.h.s.}$ in the $(i + 1)^{\text{th}}$ of Eq.~\eqref{eq:iter_rhs_dif} of Step 4, because of which we incur an exponentially decreasing in the height $d$ of $CST^{(0)}$ lower bound on $\Delta(G, n, k)$.
            \end{remark}

            \item To repeat the above procedure the order of terms in the product in Eq.~\eqref{eq:def_cst_i} needs to be changed. Suppose a new Detectability Lemma norm for $CST^{(i)}$ is composed with different order of terms in the product, denoted by $\lVert \widetilde{\DL}^{CST^{(i)}} \lvert \psi \rangle \rVert^2$. We can go around the problem of relating $\lVert \DL^{CST^{(i)}} \lvert \psi \rangle \rVert^2$ and $\lVert \widetilde{\DL} ^{CST^{(i)}} \lvert \psi \rangle \rVert^2$, by using the upper bound on $\lVert \DL^{CST^{(i)}} \lvert \psi \rangle \rVert^2$ in terms of the spectral gap $\Delta(CST^{(i)}, n, k)$ of the Hamiltonian $H(CST^{(i)}, n, k)$, which in turn can be lower bounded by $\lVert \widetilde{\DL} ^{CST^{(i)}} \lvert \psi \rangle \rVert^2$ as follows
            \begin{align}
                \lVert \DL^{CST^{(i)}} \lvert \psi \rangle \rVert^2 \leq \cfrac{1}{1 + \cfrac{\Delta(CST^{(i)}, n, k)}{(g^{(i)})^2}} \leq \cfrac{1}{1 + \cfrac{1}{4 (g^{(i)})^2} \left(1 - \lVert \widetilde{\DL} ^{CST^{(i)}} \lvert \psi \rangle \rVert^2\right)}\,,
            \end{align}
            where $g^{(i)}$ is the maximum degree of the graph $CST^{(i)}$. The first, inequality is the usual result of the Detectability Lemma, and the second inequality results from inserting the Quantum Union Bound for $\Delta(CST^{(i - 1)}, n, k)$. This approach works to relate $\lVert \DL^{CST^{(i)}} \lvert \psi \rangle \rVert^2$ and $\lVert \widetilde{\DL} ^{CST^{(i)}} \lvert \psi \rangle \rVert^2$ because the Detectability Lemma upper bound and the Quantum Union Bound lower bound are unconditional on the order of products of projectors in the Detectability Lemma operator.
            
            \item After $d$ iterations of Step 4, we have the following results.
            \begin{enumerate}
                \item For $i \in \{1, 2, \dots, d\}$, 
                \begin{align}
                    \label{eq:ins_1}
                    \lVert \DL ^{CST^{(i - 1)}} \lvert \psi \rangle \rVert^2 = \lVert \widetilde{\DL} ^{CST^{(i)}} \lvert \widetilde{\psi} \rangle \rVert^2\,,
                \end{align}
                where $\lvert \psi \rangle, \lvert \widetilde{\psi} \rangle \in \mathcal{G}_{\perp}$, and $\widetilde{\DL}$ denotes a particular order of product of projectors such that we would have to reorder that product to equate the Detectability Lemma norms on $CST^{(i)}$ and $CST^{(i + 1)}$.
                In particular, $\lVert \DL ^{CST^{(d)}} \lvert \psi \rangle \rVert^2$ is equal to the Detectability Lemma norm on the $n$-vertex 1D chain with open boundary conditions, denoted by $\lVert \DL^{\1D} \lvert \psi \rangle \rVert^2$.
                
                \item For $\DL$ and $\widetilde{\DL}$ that differ in a very specific manner in the order of products of projectors in their respective definitions,
                \begin{align}
                    \label{eq:ins_2}
                    \lVert \widetilde{\DL} ^{CST^{(i)}} \lvert \widetilde{\psi} \rangle \rVert^2 \leq \cfrac{1}{1 + \cfrac{1}{4 (g^{(i)})^2} \left(1 - \lVert \DL^{CST^{(i)}} \lvert \psi \rangle \rVert^2\right)}\,.
                \end{align}
            \end{enumerate}
        \end{enumerate}

    \subsection{Step~6: Final lower bound for the spectral gap of G}
        \begin{enumerate}
            \item From Step 3, we know that $\lVert \DL ^{ST} \ket{\psi} \rVert^2 = \lVert \DL ^{CST^{(0)}} \ket{\widetilde{\psi}} \rVert^2$, for $\ket{\psi}, \ket{\widetilde{\psi}} \in \mathcal{G}_{\perp}$.
            \item We can upper bound $\lVert \DL^{CST^{(0)}} \ket{\psi} \rVert^2$ and, hence, $\lVert \DL^{ST} \ket{\psi} \rVert^2$ in the following iterative manner. Begin with $i = 1$.
            \begin{enumerate}
                \item Use Eq.~\eqref{eq:ins_1} to equate $\lVert \DL ^{CST^{(i - 1)}} \ket{\psi} \rVert^2$ with $\lVert \widetilde{\DL} ^{CST^{(i)}} \lvert \widetilde{\psi} \rangle \rVert^2$. 
                \item Use Eq.~\eqref{eq:ins_2} to upper bound $\lVert \DL ^{CST^{(i - 1)}} \ket{\psi} \rVert^2$ in terms of $\lVert \DL^{CST^{(i)}} \ket{\psi} \rVert^2$
                \item Increment $i = i + 1$. 
            \end{enumerate}
            Repeat points 2(a), (b) and (c) till $i = d - 1$ at the end of 2(c). At the $d^{\text{th}}$ iteration stop after point 2(a) because at that point, $\lVert \DL^{CST^{(d - 1)}} \ket{\psi} \rVert^2 = \lVert \widetilde{\DL} ^{CST^{(d)}} \ket{\widetilde{\psi}} \rVert^2 = \lVert \DL^{\1D} \ket{\widetilde{\psi}} \rVert^2 \leq \cfrac{1}{1 + \cfrac{\Delta(\1D, n, k)}{4}}$, where $\Delta(\1D, n, k)$ is the spectral gap for the Hamiltonian $H(\1D, n, k)$ defined on a $\1D$ line graph with open boundary conditions.
        
            \item Consider a $CST^{(0)}$ of bounded degree $g$. Note that in general, the maximum degree of $CST^{(0)}$ is greater than the maximum degree of the corresponding $ST$. Furthermore, at each flattening out procedure of Step 4, we reduce the degree of vertices in the lowest level of $CST^{(i)}$ but we can increase degree of vertices in its first layer by at most $1$. Degrees of all other vertices remain unchanged. Therefore, $g^{(i)} \leq g + 1 = \widetilde{g}$ for all $i \in \{1, 2, \dots, d\}$, then from the previous item we find
            \begin{align}
                \lVert \DL^{CST^{(0)}} \lvert \psi \rangle \rVert^2 &\leq \cfrac{1}{1 + \cfrac{1}{4\widetilde{g}^2} - \cfrac{1}{4\widetilde{g}^2} \cfrac{1}{1 + \cfrac{1}{4\widetilde{g}^2} - \cfrac{1}{4\widetilde{g}^2} \cfrac{1}{1 + \cfrac{1}{4\widetilde{g}^2} - \cfrac{1}{4\widetilde{g}^2} \cfrac{1}{1 + \dots}}}}\,.
            \end{align}
            For convenience, define
            \begin{equation}
                \alpha := 4\widetilde{g}^2,\quad \beta_1 := 1 - \cfrac{1}{1 + \cfrac{\Delta(\1D, n, k)}{4}},\quad \beta_{i} := 1 - \cfrac{\alpha}{\alpha + \beta_{i - 1}} = \cfrac{\beta_{i - 1}}{\alpha + \beta_{i - 1}}\,,
            \end{equation}
            for all $i \in \{2, 3, \dots, d\}$. If the depth of $ST$ is $d$, then, $\lVert \DL^{CST^{(0)}} \lvert \psi \rangle \rVert^2 \leq 1 - \beta_{d}$. Now note the recursive relationship between $\beta_i$
            \begin{align}
                \beta_{d} &= \cfrac{\beta_{d - 1}}{\alpha + \beta_{d - 1}} \\
                &= \cfrac{\beta_{d - 2}}{\alpha + \beta_{d - 2}} \cfrac{1}{\alpha + \cfrac{\beta_{d - 2}}{\alpha + \beta_{d - 2}}} = \cfrac{\beta_{d - 2}}{\alpha (\alpha + \beta_{d - 2}) + \beta_{d - 2}} \\
                &= \cfrac{\beta_{d - 3}}{\alpha + \beta_{d - 3}} \cfrac{1}{\cfrac{\alpha (\alpha (\alpha + \beta_{d - 3}) + \beta_{d - 3})}{\alpha + \beta_{d - 3}} + \cfrac{\beta_{d - 3}}{\alpha + \beta_{d - 3}}} = \cfrac{\beta_{d - 3}}{\alpha (\alpha (\alpha + \beta_{d - 3}) + \beta_{d - 3}) + \beta_{d - 3}} \\
                \nonumber
                &\vdots \\
                &= \cfrac{\beta_1}{\alpha^{d - 1} + \beta_1 \sum_{i = 0} ^{d - 2} \alpha^i }\,.
            \end{align}
            Using the formula for geometric sum,
            \begin{align}
                \beta_{d} &= \cfrac{\beta_1}{\alpha^{d - 1} (1 + \beta_1 \sum_{i = 1} ^{d - 1} \alpha^{-i})} \\
                &= \cfrac{\beta_1}{\alpha^{d - 1}} \cfrac{1}{1 + \beta_1 \left(\cfrac{1 - \alpha^{-(d - 1)}}{\alpha - 1}\right) } \\
                &\geq \cfrac{\beta_1}{\alpha^{d - 1}} \cfrac{1}{1 + \cfrac{\beta_1}{\alpha - 1} }\,.
            \end{align}
            We note that $\alpha = 4\widetilde{g}^2 \geq 36$, for $g = 2$, and recall that $\beta_1 = 1 - \cfrac{1}{1 + \cfrac{\Delta(\1D, n, k)}{4}} \leq 1$. Therefore,
            \begin{align}
                \beta_{d} &\geq \cfrac{35}{36} \cfrac{\beta_1}{\alpha^{d - 1}} \\
                &\geq \cfrac{35}{36} \cfrac{1}{(4(g+1)^2)^{d - 1}} \left(1 - \left(1 + \frac{\Delta(\1D, n, k)}{4}\right)^{-1}\right) \\
                &\geq \cfrac{35}{36} \cfrac{1}{(4(g+1)^2)^{d - 1}} \left(\frac{\Delta(\1D, n, k)}{4} - \frac{\Delta(\1D, n, k)^2}{16}\right)\\
                &\geq \cfrac{35}{36} \cfrac{1}{(4(g+1)^2)^{d - 1}} \left(\frac{ \Delta(\1D, n, k)}{4} - \frac{\Delta(\1D, n, k)}{16}\right)\\
                &\geq \cfrac{105}{576} \cfrac{\Delta(\1D, n, k)}{(4(g+1)^2)^{d - 1}},
            \end{align}
            which implies
            \begin{align}
                \label{eq:one_less_recur}
                \lVert \DL^{CST^{(0)}} \lvert \psi \rangle \rVert^2 \leq 1 - \cfrac{105}{576} \cfrac{\Delta(\1D, n, k)}{(4(g+1)^2)^{d - 1}}\,.
            \end{align}
            From Step 3, we know that equating the Detectability Lemma norms for $ST$ and $CST^{(0)}$ followed by changing the order of terms in the later to relate it to $CST^{(1)}$ incurs one extra step of item 2 of this section. This amounts to changing the exponent of $4(g + 1)^2$ from $d - 1$ to $d$, thus
            \begin{align}
                \lVert \DL^{ST} \lvert \psi \rangle \rVert^2 \leq 1 - \cfrac{105}{576} \cfrac{\Delta(\1D, n, k)}{(4(g+1)^2)^{d}}\,.
            \end{align}
            Finally, we apply the Quantum Union Bound for $ST$ to lower bound $\Delta(ST, n, k)$ in terms of $\lVert \DL^{ST} \lvert \psi \rangle \rVert^2$ and, hence, in terms of $\Delta(\1D, n, k)$ by the above equation as follows
            \begin{align}
                \Delta(ST, n, k) &\geq \frac{1}{4} \left( 1 - \lVert \DL^{ST} \lvert \psi \rangle \rVert^2 \right)\\
                &\geq \cfrac{35}{768} \cfrac{\Delta(\1D, n, k)}{(4(g+1)^2)^{d}}\,,
            \end{align}
            and use the fact that $\Delta(G, n, k) \geq \Delta(ST, n, k)$ to conclude the proof of \autoref{lem:any_g_1d}
            \begin{align}
                \Delta(G, n, k) \geq \cfrac{35}{768} \cfrac{\Delta(\1D, n, k)}{(4(g+1)^2)^{d}}\,.
            \end{align}
        \end{enumerate}

    \subsection{Logarithmic upper bound on tree depth \label{app:tree_graph_prop}}

    Here we prove that the depth $d$ of a tree $T$ of order $n$ is, $d = O(\log(n))$. Intuitively, the minimum and/or maximum and/or average degree of the tree should appear in the upper bound on the depth. An upper bound on the depth in terms of both the number of vertices and some function of the degree of vertices should improve the lower bound reported in \autoref{lem:any_g_1d}. The improvement could be significant and increase the lower bound from 1 over quasi-polynomial in $n$, which comes from setting $g \leq n$ and $d = O(\log(n))$, to perhaps $1$ over a polynomial in $n$. We leave that improvement to future work. We prove the upper bound on the depth via three lemmas that build on top of one another. In the following, when we say depth of $T$, then we mean the definition of depth as given in \autorefapp{app:dep_and_cst}, and when we say depth of a vertex $v$ in $T$, then we mean the label for $v$ in that output of the function $f(\{T\}, \mathcal{A}, 0)$ whose max is the definition of the depth of $T$.
        \begin{lemma}
            \label{lem:min_deg_3_v}
            If a leaf of a tree graph is at depth $d$, then there must necessarily be at least $\lfloor d / 2\rfloor$ degree $3$ or greater vertices along the path from the root to that leaf.
        \end{lemma}
        
        \begin{proof}
            Consider a tree graph $T$, with a root vertex $r$. Assume to the contrary that there exists a leaf $l$ such that the $r-l$ path contains fewer than $\lfloor d/2 \rfloor$ degree $3$ or greater vertices and there exists a list of choices that could be made in $f$ such that $l$ is labelled with depth $d$.
        
            By the definition of the algorithm $f$, two degree $3$ or greater vertices connected by an arbitrary length path graph can differ from each other by not more than $2$ units. If there are $k \leq \lfloor d/2 \rfloor - 1$ degree $3$ or greater vertices, then the first and the last such vertices differ by at most $2(k - 1)$ units in depth. Connecting an arbitrary length path graph between $r$ and the first and $l$ and the last degree 3 or greater vertices can add at most $1$ and $2$ units of depth to the depth of $l$, respectively. Therefore, with addition of path graphs (no degree 3 or greater vertices), the depth of $l$ can be increased by at most $3$ units. The maximum depth of $l$ could be,
            \begin{align}
                2(k - 1) + 3 \leq 2(\lfloor d/2 \rfloor - 1) + 1 = \begin{cases}
                    d - 1,\ d \in \mathrm{even} \\
                    d - 2,\ d \in \mathrm{odd}
                \end{cases}\,.
            \end{align}
            We arrive at a contradiction about the depth of $l$, hence, the assumption that $r-l$ path could have fewer than $\lfloor d / 2 \rfloor$ degree $3$ or greater vertices is false, and the Lemma is proved.
        \end{proof}
        
        \begin{lemma}
            \label{lem:del_less_deg_3_paths}
            Suppose that the depth of a tree graph is $d$. Suppose further that there are leaves in this tree such that the paths from the root to the leaves contain fewer than $\lfloor d/2 \rfloor$ degree $3$ or greater vertices. If those leaves are deleted from the graph, then a tree graph of smaller order results whose depth remains unchanged at $d$.
        \end{lemma}
        
        \begin{proof}
            Suppose a tree graph $T$ and the set $L_{<d}$ of all such leaves in $T$. Denote the set of all choices made in $f$ for $T$ by $\mathrm{choices}_{\text{with}}$. By \autoref{lem:min_deg_3_v}, for any choice of paths in $\mathrm{choices}_{\text{with}}$, leaves in $L_{<d}$ cannot be labelled by depth greater than or equal to $d$. Thus, if depth of $T$ is $d$, then for each choice in $\mathrm{choices}_{\text{with}}$, there exists a leaf $l \notin L_{<d}$ such that the depth of $l$ is greater than or equal to $d$.
            
            Consider another graph, denoted by $T'$, which is derived from $T$ by recursively subtracting the leaves in $L_{<d}$ and edges that connect to those leaves, till $L_{<d}$ is empty for the resulting graph. Observe that $\mathrm{choices}_{\text{with}}$ contains the set of all choices made in $f$ for $T'$, denoted by $\mathrm{choices}_{\text{without}}$. That is, $\mathrm{choices}_{\text{without}} \subset \mathrm{choices}_{\text{with}}$. Then, from the result established in the last line of the first paragraph above, we conclude that for each choice in the subset $\mathrm{choices}_{\text{without}}$ of $\mathrm{choices}_{\text{with}}$, there exists a leaf in $l$ in $T'$ such that the depth of $l$ is greater than or equal to $d$. Therefore, the depth of $T'$ is also $d$ and its order is strictly less than that of $T$, which completes the proof.
            % The length of the paths from root to leaves not in $L_{<d}$ remain unaffected by removing leaves in $L_{<d}$. Consider two graphs,  original tree graph denoted by $T$ and $T$ minus the leaves in $L_{<d}$ and edges that connect to those leaves. The set of all choices in $f$ for $T$ is denoted by $\mathrm{choices}_{\text{with}}$ and contains the set of all choices in $f$ for $T$ without the leaves in $L_{<d}$, denoted by $\mathrm{choices}_{\text{without}}$. That is, $\mathrm{choices}_{\text{without}} \subset \mathrm{choices}_{\text{with}}$. Because for all choices in $\mathrm{choices}_{\text{with}}$ there exists a leaf not in $L_{<d}$ that is labelled with depth $\geq d$, we infer that for all choices in $\mathrm{choices}_{\text{without}}$ there exists a leaf not in $L_{<d}$ that is labelled with depth $\geq d$. Since, removing a vertex that for no choice in $f$ is labelled by depth $\geq d$ affects the depth of the tree, removing vertices in $L_{<d}$ from the tree do not affect its depth, but only lower its order.
        \end{proof}
        
        \begin{corollary}
            \label{cor:min_deg_3_per_path}
            All paths from root to leaves in a depth $d$ tree graph of minimum order must contain at least $\lfloor d / 2\rfloor$ degree 3 or greater vertices.
        \end{corollary}
        \begin{proof}
            If this were not the case, then from \autoref{lem:del_less_deg_3_paths} a graph of smaller order with same depth could be constructed. Hence, proved by contradiction. 
        \end{proof}
        
        \begin{lemma}
            \label{lem:dep_upper_bound}
            The minimum number $m$ of degree $3$ or greater vertices in a tree graph of depth $d$ is such that $d = O(\log (m))$. 
        \end{lemma}
        \begin{proof}
            Suppose $G$ is a tree graph with minimum number of vertices with depth $d$. By \autoref{lem:del_less_deg_3_paths}, all paths from the root $r$ to leaves $l$ must contain $D \geq \lfloor d/2 \rfloor$ degree $3$ or greater vertices. Pick one leaf $l$, and label the first $\lfloor d/2 \rfloor$ degree $3$ or greater vertices in the directed path from $r$ to $l$ by $v_i$, for $i \in \{1, 2, \dots, \lfloor d/2 \rfloor\}$. When we say ``outgoing'' paths from a vertex $v_i$, we mean to indicate sub-paths from $v_i$ to leaves that do not pass through $v_{i - 1}$.
            
            Since $v_{\lfloor d/2 \rfloor - 1}$ is at least degree $3$, it has at least two outgoing paths, one that contains $v_{\lfloor d/2 \rfloor}$, and the other that does not. To satisfy \autoref{lem:del_less_deg_3_paths}, each outgoing path from $v_{\lfloor d/2 \rfloor - 1}$ must contain at least $L_{\lfloor d/2 \rfloor - 1} := 1$ degree $3$ or greater vertex. (Note that, along one of the outgoing paths from $v_{\lfloor d/2 \rfloor - 1}$, $v_{\lfloor d/2 \rfloor}$ is the required degree 3 or greater vertex.) 
        
             Similarly, since $v_{\lfloor d/2 \rfloor - 2}$ is at least degree $3$, it has at least two outgoing paths, one that contains $v_{\lfloor d/2 \rfloor - 1}$, and the other that does not. We know that an outgoing path from $v_{\lfloor d/2 \rfloor - 2}$ that contains $v_{\lfloor d/2 \rfloor - 1}$ must contain at least $L_{\lfloor d/2 \rfloor - 2} = 1 + 2 L_{\lfloor d/2 \rfloor - 1}$ degree $3$ or greater vertices. Note that, in the sum, the $1$ corresponds to $v_{\lfloor d/2 \rfloor - 1}$ and the $2 L_{\lfloor d/2 \rfloor - 1}$ corresponds to the fewest number of degree $3$ or greater vertices in the outgoing paths from $v_{\lfloor d/2 \rfloor - 1}$ not including $v_{\lfloor d/2 \rfloor - 1}$ itself, which was established in the previous paragraph. The same lower bound of $L_{\lfloor d/2 \rfloor - 2}$ must apply to the other outgoing path from $v_{\lfloor d/2 \rfloor - 2}$, or else, there exists an outgoing path from $v_{\lfloor d/2 \rfloor - 2}$ (that does not contain $v_{\lfloor d/2 \rfloor - 1}$) that violates \autoref{cor:min_deg_3_per_path}.
        
             Continuing by the same logic, $L_{\lfloor d/2 \rfloor - i} = 1 + 2 L_{\lfloor d/2 \rfloor - (i - 1)}$, and the lower bound on the minimum number of degree 3 or greater vertices in $G$ is $L_{-1} = 2L_{0} = 2(1 + 2L_{1}) = 2(1 + 2(1 + 2(1 + \dots))) = \sum_{i = 1} ^{\lfloor d / 2 \rfloor - 1} 2^i = 2^{\lfloor d/2 \rfloor} - 2$. Note that the additive factor of $1$ is missing in $L_{-1}$ because the root of the tree may or may not be a degree $3$ or greater vertex.
             
             Suppose $m$ denotes the number of degree $3$ or greater vertices in $G$, then, 
             \begin{align}
                \label{eq:dep_upper_bound}
                &m \geq 2^{\lfloor d/2 \rfloor} - 2\,, \\
                &d \leq (2 / \log (2)) \log(m + 2) + 1 = O(\log(m))\,.
            \end{align}
        \end{proof}

        \begin{remark}
            \label{rem:oto_map_bin_tree}
            Consider the counting argument in the proof of \autoref{lem:dep_upper_bound}. Alternatively, we could consider a perfect binary tree of height $\lfloor d / 2 \rfloor$ and a one to one map from each of its degree $3$ vertices to a degree $3$ or greater vertex in the tree graph $T$ considered in the proof. Note the domain of the map excludes the root of the perfect binary tree. The co-domain of the map is the set of all degree $3$ or greater vertices in $T$. Since the size of the co-domain of a one-to-one map is greater than or equal to the size of its domain, the number of degree $3$ or greater vertices in $T$ is greater than or equal to number of vertices in the perfect binary tree of height $\lfloor d / 2 \rfloor$ minus $1$, for excluding the root of the binary tree from the domain. That is $m \geq 2^{\lfloor d / 2 \rfloor} - 2$  
        \end{remark}
        
        \begin{corollary}
            \label{cor:dep_upper_bound}
            The order $n$ of a tree graph of depth $d$ and maximum degree $\maxdeg \geq 4$ is such that $d \leq (2 / \log (2)) \log(n + 1) - 1$.
        \end{corollary}
        
        \begin{proof}
            The number of leaves in $T$ is greater than or equal to twice the number of out-going paths from the degree $3$ or greater vertices closest to the leaves. Continue the point of view from \autoref{rem:oto_map_bin_tree}, then pre-images of the degree $3$ or greater vertices closest to the leaves in $T$ are the leaves of the perfect binary tree. Therefore, the number of degree $3$ or greater vertices in $T$ closest to the leaves is $2^{\lfloor d / 2 \rfloor - 1}$, the number of leaves in the perfect binary tree of height $\lfloor d / 2 \rfloor$. Since each degree $3$ or greater vertex closest to the leaves in $T$ contributes at least $2$ leaves to the graph, the total number of leaves in $T$ is at least $2^{\lfloor d / 2 \rfloor}$. The sum total of leaves and a single root vertex is $1 + 2^{\lfloor d / 2 \rfloor}$. Therefore, $n - (1 + 2^{\lfloor d / 2 \rfloor})$ must be greater than or equal to the number of degree $3$ or greater vertices. From Eq.~\eqref{eq:dep_upper_bound}, we find,
            \begin{align}
                &n - (1 + 2^{\lfloor d / 2 \rfloor}) \geq m \geq 2^{\lfloor d / 2 \rfloor} - 2 \\
                &n + 1 \geq 2^{\lfloor d / 2 \rfloor + 1} \\
                &d \leq \frac{2}{\log(2)} \log (n + 1) - 1.  
            \end{align}
        \end{proof}
        \begin{remark}
            The upper bound in \autoref{cor:dep_upper_bound} can be slightly improved if we further suppose that the maximum degree of $T$ is $\maxdeg$. On repeating the proof for that case, we find that at least $\maxdeg-3$ vertices are left out in the mapping from $T$ to the perfect binary tree (consider the degree $\maxdeg$ vertex to correspond to any vertex that is not a leaf in the perfect binary tree). Consequently, $n - (1 + 2^{\lfloor d / 2 \rfloor}) - (\maxdeg - 3)$ must be greater than the number of degree $3$ or greater vertices. This improves the upper bound on $d = O(\log(n - \maxdeg))$. However, this improvement is not sufficient to reduce the quasi-polynomial circuit size in \autoref{thm:size_any_graph} to polynomial circuit size.
        \end{remark}

\section{Semiclassical Approximation \label{app:semiclass}}
We presented rigorous spectral gap lower bounds for the Hamiltonians defined via \autoref{def:ham_gnk}. However, it is interesting to note that for $k = 2$, we can speculate about the behavior of the asymptotic (in $n$) star graph spectral gaps to provide evidence in favor of \autoref{con:star_gap} using semiclassical approximation. Here, we follow the prescription of \cite{ZnidaricMarko2008EC}, where such a semiclassical approximation was done for complete graphs. First, we write the Hamiltonian for star graph for $k = 2$ in terms of usual spin-$1/2$ Pauli matrices. Note that regardless of the local dimension of the qudits, the moment operator for $k = 2$ is analogous to an $n$ spin-$1/2$ interaction Hamiltonian. For star graphs, it is given by, 
\begin{align}
    \label{eq:make_semi_class}
    H(G_\star, n, 2) = \frac{n-1}{2} - \frac{(n-1) q}{2 \left(q^2+1\right)} \sigma_z ^{(0)} -\frac{q}{q^2+1} S_z - \frac{q^2}{q^2+1} \sigma_x ^{(0)} S_x -\frac{1}{q^2+1} \sigma_y ^{(0)} S_y,
\end{align}
where $\sigma_\alpha$ denote the usual Pauli matrices and the superscript of $(0)$ denotes that they are defined on the central spin Hilbert space. $S_\alpha := \frac{1}{2} \sum_{i = 1} ^{n - 1} \sigma_\alpha ^{(i)}$, for $\alpha \in \{x, y, z\}$, where the superscript denotes the Hilbert spaces of the remaining $n - 1$ spins. Since, the Hilbert spaces on which the operators act is clear from Eq.~\eqref{eq:make_semi_class}, we simplify the notation by omitting the superscript over Pauli matrices. Then, we insert the semiclassical approximation by treating total spin of $n-1$ spin-$1/2$ as a classical spin. We demote the operators $S_x, S_y$ and $S_z$ to variables $S\sin{\theta}\cos{\phi}, S\sin{\theta}\sin{\phi}$ and $S\cos{\theta}$, where $S$ denotes the total angular momentum. A heuristic justification of this is that the Hamiltonian should achieve its lowest eigenvalue for the largest value of the total spin \cite{ZnidaricMarko2008EC,Leyvraz2005Large}. Since we are interested in the asymptotic gap, we define a rescaled total spin $\xi = 2S / n$ and take it to be $1$. Then, the resulting semiclassical Hamiltonian in the asymptotic limit, denoted by $H_{\mathrm{s.c.}}(G_{\star})$ reads,
\begin{multline}
    H_{\mathrm{s.c.}}(G_{\star}) = \frac{1}{2} - \frac{q}{2 \left(q^2+1\right)} \sigma _z -\frac{q}{2 \left(q^2+1\right)} \cos (\theta) - \frac{q^2}{2 \left(q^2+1\right)} \sin (\theta) \cos (\phi) \sigma _x \\
    - \frac{1}{2 \left(q^2+1\right)} \sin (\theta)\sin (\phi ) \sigma_y,
\end{multline}
which in its matrix form is given as
\begin{align}
    H_{\mathrm{s.c.}} (G_\star) = 
    \begin{bmatrix}
        \frac{1}{2} - \frac{q}{2(q^2 + 1)} (1 + \cos(\theta)) & \frac{\sin(\theta)}{2(q^2 + 1)} ( - q^2 \cos(\phi) + i \sin(\phi) ) \\
        \frac{\sin(\theta)}{2(q^2 + 1)} ( - q^2 \cos(\phi) - i \sin(\phi) ) & \frac{1}{2} + \frac{q}{2(q^2 + 1)} (1 - \cos(\theta))
    \end{bmatrix}.
\end{align}
We find the spectrum of the semiclassical Hamiltonian in the asymptotic limit, which is parameterized by $\theta$ and $\phi$, then we find the minimum of the two eigenvalues and expand that about its minimum value of $0$ (ground state energy) in the parameter $\theta$. The choice is slightly arbitrary, but the asymptotic gaps seem to change for different values of $q$ as observed in numerics, so we chose to expand in $\theta$ about $\phi = 0$ and $\theta = \arctan(\sqrt{q^2 - 1})$). We do this step via symbolic computing. We expand up to the second order because the first order term vanishes. To evaluate the second order term we need to insert a value for $\theta - \arctan(\sqrt{q^2 - 1}) =: \delta \theta$. We choose $\delta \theta$ as follows: imagine two classical angular momentum vectors whose $z$-components of the angular momentum are the successive eigenvalues of the $S_z$ quantum angular moment operator. We take $\delta \theta$ to be the minimum over all such successive pairs of vectors of the angle between those vectors in the $\phi = 0$ plane. Using this expression for $\delta \theta$ in the expansion of the eigenvalue near its minimum value, we find the asymptotic expression for the gap to be $\Delta(G_\star, n\rightarrow \infty, 2) \approx 1 - \frac{1}{q^2}$. Note that the numerical calculations presented in \autoref{tab:star_gaps} up to $n=22$ are consistent with the $k = 2$ star graph gaps approaching the semiclassical asymptotic gap $1 - 1 / q^2$.

We further hint at a potential connection between the Hamiltonian on star graphs for $k = 2$ to the well-known Jaynes-Cumming (JC) model of a two-level atom interacting with a bosonic field. The original JC model is known to be an integrable quantum Hamiltonian. The Hamiltonian on star graphs that we work with is almost identical with the only significant difference being that instead of having standard raising and lowering operators for the two level system, we have those operators modified such that their commutator is different by a $q$ dependent constant from the standard value. In the past, connections to integrable quantum Hamiltonians have led to \textit{exact} determination of the second moment operator for $\1D$ RQCs. Maybe the above connection to other integrable systems could offer insight to determining the exact spectral gap of the Hamiltonian on star graphs for $k = 2$, which would directly provide the circuit size for forming unitary $2$-designs on arbitrary architectures through \autoref{thm:knabe_any_connected_graph}.

\section{Numerical Calculations}
\label{app:num_dat}

We present our numerical determination of the star graph gaps in \autoref{tab:star_gaps}. At best, we could numerically compute the spectral gap of a $2^{88} \times 2^{88}$ Hamiltonian matrix (spectral gap for $k = 2$, $n = 22$). In order to achieve this, we had to take two steps. First, we projected the Hamiltonian to the relevant subspace that contains the spectral gap as was done in \cite{BrandaoFernandoGSL2010EQ}. This alone was not sufficient to make the spectral gap calculation for $k = 2$ and $n = 22$ tractable. So, next we found a basis in which the Hamiltonian matrix was more sparse than it would have been were it expressed in the computational basis. That basis for $k = 2$ is the single site basis that corresponds to the normalized and vectorized projector on to the symmetric and anti-symmetric sub-spaces. For any higher values of $k$, an analogous single site basis was found by taking the $QR$ decomposition of the matrix whose columns were the ground state vectors.

\begin{center}
    \setlength{\tabcolsep}{1.5pt}
    \setlength{\columnsep}{1.5cm}

\begin{table}[ht!]
    \begin{tikzpicture}[scale=0.8,baseline=0mm]
        \node at (-5,0) {
            \begin{tabular}{|| r | r || r | r || r | r ||}
                \hline
                \multicolumn{6}{|c|}{$k = 2$} \\
                \hline
                \multicolumn{2}{||c||}{$q = 2$} & \multicolumn{2}{c||}{$q = 3$} & \multicolumn{2}{c||}{$q = 4$} \\
                \hline
                \multicolumn{1}{||c|}{$n_{\star}$} & \multicolumn{1}{|c||}{$\Delta(G_{\star}, n_{\star}, k)$} & \multicolumn{1}{c|}{$n_{\star}$} & \multicolumn{1}{|c||}{$\Delta(G_{\star}, n_{\star}, k)$} & \multicolumn{1}{c|}{$n_{\star}$} & \multicolumn{1}{|c||}{$\Delta(G_{\star}, n_{\star}, k)$} \\
                \hline
                3  & 0.6000 & 3  & 0.7000 & 3  & 0.7647 \\
                4  & 0.5566 & 4  & 0.7190 & 4  & 0.8134 \\
                5  & 0.5583 & 5  & 0.7650 & 5  & 0.8668 \\
                6  & 0.5776 & 6  & 0.8078 & 6  & 0.8997 \\
                7  & 0.6038 & 7  & 0.8373 & 7  & 0.9153 \\
                8  & 0.6309 & 8  & 0.8545 & 8  & 0.9222 \\
                9  & 0.6556 & 9  & 0.8638 & 9  & 0.9256 \\
                10 & 0.6759 & 10 & 0.8691 & 10 & 0.9276 \\
                11 & 0.6913 & 11 & 0.8723 & 11 & 0.9290 \\
                12 & 0.7025 & 12 & 0.8745 & 12 & 0.9300 \\
                13 & 0.7105 & 13 & 0.8761 & 13 & 0.9307 \\
                14 & 0.7161 & 14 & 0.8773 & 14 & 0.9314 \\
                15 & 0.7203 & 15 & 0.8783 & 15 & 0.9319 \\
                16 & 0.7233 & 16 & 0.8792 & 16 & 0.9323 \\
                17 & 0.7257 & 17 & 0.8799 & 17 & 0.9327 \\
                18 & 0.7277 & 18 & 0.8805 & 18 & 0.9330 \\
                19 & 0.7293 & 19 & 0.8810 & 19 & 0.9333 \\
                20 & 0.7306 & 20 & 0.8815 & 20 & 0.9336 \\
                21 & 0.7318 & 21 & 0.8819 & 21 & 0.9338 \\
                22 & 0.7328 & 22 & 0.8823 & 22 & 0.9340 \\
                \hline
            \end{tabular}
        };
        \node at (5.8,3.89) {            
            \begin{tabular}{|| r | r || r | r || r | r ||}
                \hline
                \multicolumn{6}{|c|}{$k = 3$} \\
                \hline
                \multicolumn{2}{||c||}{$q = 2$} & \multicolumn{2}{c||}{$q = 3$} & \multicolumn{2}{c||}{$q = 4$} \\
                \hline
                \multicolumn{1}{||c|}{$n_{\star}$} & \multicolumn{1}{|c||}{$\Delta(G_{\star}, n_{\star}, k)$} & \multicolumn{1}{c|}{$n_{\star}$} & \multicolumn{1}{|c||}{$\Delta(G_{\star}, n_{\star}, k)$} & \multicolumn{1}{c|}{$n_{\star}$} & \multicolumn{1}{|c||}{$\Delta(G_{\star}, n_{\star}, k)$} \\
                \hline
                3 & 0.6000 & 3 & 0.7000 & 3 & 0.7647 \\
                4 & 0.5566 & 4 & 0.7190 & 4 & 0.8134 \\
                5 & 0.5583 & 5 & 0.7650 & 5 & 0.8668 \\
                6 & 0.5776 & 6 & 0.8078 & 6 & 0.8997 \\
                7 & 0.6038 & 7 & 0.8373 & 7 & 0.9153 \\
                8 & 0.6309 & 8 & 0.8545 & 8 & 0.9222 \\
                9 & 0.6556 & 9 & -      & 9 & - \\
                \hline
            \end{tabular}
        };
        \node at (5.8,-1.76) {            
            \begin{tabular}{|| r | r || r | r || r | r ||}
                \hline
                \multicolumn{6}{|c|}{$k = 4$} \\
                \hline
                \multicolumn{2}{||c||}{$q = 2$} & \multicolumn{2}{c||}{$q = 3$} & \multicolumn{2}{c||}{$q = 4$} \\
                \hline
                \multicolumn{1}{||c|}{$n_{\star}$} & \multicolumn{1}{|c||}{$\Delta(G_{\star}, n_{\star}, k)$} & \multicolumn{1}{c|}{$n_{\star}$} & \multicolumn{1}{|c||}{$\Delta(G_{\star}, n_{\star}, k)$} & \multicolumn{1}{c|}{$n_{\star}$} & \multicolumn{1}{|c||}{$\Delta(G_{\star}, n_{\star}, k)$} \\
                \hline
                3 & 0.5000 & 3 & 0.7000 & 3 & 0.7647 \\
                4 & 0.5566 & 4 & 0.7190 & 4 & 0.8134 \\
                5 & 0.5583 & 5 & - & 5 & - \\
                \hline
            \end{tabular}
        };
        \node at (5.8,-5.68) {
            \begin{tabular}{|| r | r || r | r || r | r ||}
                \hline
                \multicolumn{6}{|c|}{$k = 5$} \\
                \hline
                \multicolumn{2}{||c||}{$q = 2$} & \multicolumn{2}{c||}{$q = 3$} & \multicolumn{2}{c||}{$q = 4$} \\
                \hline
                \multicolumn{1}{||c|}{$n_{\star}$} & \multicolumn{1}{|c||}{$\Delta(G_{\star}, n_{\star}, k)$} & \multicolumn{1}{c|}{$n_{\star}$} & \multicolumn{1}{|c||}{$\Delta(G_{\star}, n_{\star}, k)$} & \multicolumn{1}{c|}{$n_{\star}$} & \multicolumn{1}{|c||}{$\Delta(G_{\star}, n_{\star}, k)$} \\
                \hline
                3 & 0.5000 & 3 & - & 3 & - \\
                \hline
            \end{tabular}
        };        
    \end{tikzpicture}
    \caption{\label{tab:star_gaps} The table presents our numerical calculations for the spectral gaps, $\Delta(G_{\star}, n_{\star}, k)$, for various values of the parameters $k, q$ and $n_{\star}$. We note that increasing 
    We note that the dimension of the matrices increasing exponentially in $k$ and $n_*$ and the cost to compute the matrix elements grows exponentially in $q$, and thus quickly becomes computationally prohibitive.}
    \end{table}
\end{center}

\printbibliography

\end{document}